\def\IsProofInAppendix{} 

\ifdefined\IsLLNCS
\documentclass{llncs}
\else
\documentclass[11pt]{article}
\usepackage{fullpage}
\fi

\usepackage{amsfonts,amsmath,amssymb,boxedminipage,color,url,nccmath}

\usepackage[bottom]{footmisc} 

\usepackage[pageanchor=false,pdfstartview=FitH,colorlinks,linkcolor=blue,filecolor=blue,citecolor=blue,urlcolor=blue]{hyperref} 

\ifdefined\IsLLNCS\else
\usepackage{amsthm} 
\fi

\usepackage[T1]{fontenc}
\usepackage{lmodern} \normalfont 
\DeclareFontShape{T1}{lmr}{bx}{sc} { <-> ssub * cmr/bx/sc }{}

\usepackage{amsfonts,amssymb,boxedminipage,color,url,nccmath}
\usepackage{enumerate,paralist}
\usepackage{varwidth}
\usepackage{aliascnt}
\usepackage[numbers]{natbib} 
\usepackage{cleveref} 
\usepackage{xspace}
\usepackage{xstring}

\usepackage[nottoc]{tocbibind} 

\usepackage[hypcap=false]{caption}
\usepackage{subcaption}

\usepackage{graphicx}

\usepackage[usestackEOL]{stackengine}
\setstackgap{L}{\normalbaselineskip}

\newcommand{\remove}[1]{}

\newcommand{\LLNCS}[1]{\ifdefined\IsLLNCS #1 \fi}

\newcommand{\TLLNCS}[2]{\ifdefined\IsLLNCS#1\else #2 \fi}

\ifdefined\IsDraft
\newcommand{\authnote}[2]{[{\color{magenta}\textbf{Note(#1):}}~{\color{blue} #2}]}
\else
\newcommand{\authnote}[2]{}
\fi

\newcommand{\radded}[1]{\added{Ran}{#1}}

\ifdefined\IsDraft
\newcommand{\added}[2]{{{\color{blue}\textbf{Added(#1):}}~{\color{blue} #2}}}
\else
\newcommand{\added}[2]{#2}
\fi

\ifdefined\IsDraft
    
\else
    
\fi

\ifdefined\IsDraft
    \newcommand{\deleted}[2]{{{\color{red}\textbf{Deleted(#1):}}~{\color{red} #2 }}}
\else
    \newcommand{\deleted}[2]{}
\fi

\newcommand{\Tableofcontents}{
\thispagestyle{empty}
\pagenumbering{gobble}
\clearpage
{\small{
\tableofcontents
}}
\thispagestyle{empty}
\clearpage
\pagenumbering{arabic}
}

\newtheorem{innercustomthm}{Theorem}
\newenvironment{customthm}[1]
  {\renewcommand\theinnercustomthm{#1}\innercustomthm}
  {\endinnercustomthm}

\newcommand{\sdotfill}{\textcolor[rgb]{0.8,0.8,0.8}{\dotfill}} 

\newcommand{\Ensuremath}[1]{\ensuremath{#1}\xspace}
\newcommand{\MathAlg}[1]{\mathsf{#1}}
\newcommand{\MathAlgX}[1]{\Ensuremath{\MathAlg{#1}}}

\ifdefined\IsLLNCS
\newcommand{\QED}{\qed}
\else
\newcommand{\QED}{}
\fi

\ifdefined\IsLLNCS
\newcommand{\SUBSUBSEC}{.}
\else
\newcommand{\SUBSUBSEC}{}
\fi

\newcommand{\ie}  {i.e.,\xspace}
\newcommand{\eg}  {e.g.,\xspace}

\newcommand{\assign}{\ensuremath{\mathrel{\vcenter{\baselineskip0.5ex \lineskiplimit0pt \hbox{\scriptsize.}\hbox{\scriptsize.}}}=}}
\newcommand{\abs}[1]{\left\lvert #1 \right\rvert}

\newcommand{\ceil}[1]{\left\lceil #1 \right\rceil}

\newcommand{\set}[1]{\left\{#1\right\}}
\newcommand{\sset}[1]{\{#1\}}

\newcommand{\size}[1]{\left|#1\right|}
\newcommand{\ssize}[1]{|#1|}

\newcommand{\SD}{\mathsf{\textsc{SD}}}

\newcommand{\compindist}{\stackrel{\rm c}{\equiv}}
\newcommand{\statclose}{\stackrel{\rm s}{\equiv}}

\renewcommand{\Pr}{{\mathrm {Pr}}}
\newcommand{\pr}[1]{\Pr\left[#1\right]}
\newcommand{\ppr}[2]{\Pr_{#1}\left[#2\right]}

\newcommand{\Uni}{{\mathord{\mathcal{U}}}}

\newcommand{\N}{{\mathbb{N}}}

\newcommand{\zo}{\{0,1\}}
\newcommand{\zn}{{\zo^n}}

\newcommand{\zs}{{\zo^\ast}}
\newcommand{\is}{{i^\ast}}
\newcommand{\Is}{{I^\ast}}
\newcommand{\js}{{j^\ast}}

\newcommand{\ls}{{l^\ast}}

\newcommand{\poly}{\operatorname{poly}}
\newcommand{\polylog}{\operatorname{polylog}}

\renewcommand{\cref}{\Cref}

\TLLNCS{
\newaliascnt{claiml}{theorem}
\newtheorem{claiml}[claiml]{Claim}
\aliascntresetthe{claiml}

\renewenvironment{claim}{\begin{claiml}}{\end{claiml}}
}
{
\newtheorem{theorem}{Theorem}[section]

\newaliascnt{lemma}{theorem}
\newtheorem{lemma}[lemma]{Lemma}
\aliascntresetthe{lemma}

\newaliascnt{claim}{theorem}
\newtheorem{claim}[claim]{Claim}
\aliascntresetthe{claim}

\newaliascnt{corollary}{theorem}
\newtheorem{corollary}[corollary]{Corollary}
\aliascntresetthe{corollary}

\newaliascnt{proposition}{theorem}
\newtheorem{proposition}[proposition]{Proposition}
\aliascntresetthe{proposition}

\newaliascnt{conjecture}{theorem}

\aliascntresetthe{conjecture}

\newaliascnt{definition}{theorem}
\newtheorem{definition}[definition]{Definition}
\aliascntresetthe{definition}

\newaliascnt{remark}{theorem}

\aliascntresetthe{remark}

\newaliascnt{example}{theorem}

\aliascntresetthe{example}
}

\crefname{lemma}{Lemma}{Lemmas}
\crefname{figure}{Figure}{Figures}
\crefname{claim}{Claim}{Claims}
\crefname{corollary}{Corollary}{Corollaries}
\crefname{proposition}{Proposition}{Propositions}
\crefname{conjecture}{Conjecture}{Conjectures}
\crefname{definition}{Definition}{Definitions}
\crefname{remark}{Remark}{Remarks}
\crefname{exmaple}{Example}{Examples}

\crefname{equation}{Equation}{Equations}

\newaliascnt{proto}{theorem}

\newtheorem{proto}[proto]{Protocol}

\aliascntresetthe{proto}
\crefname{proto}{protocol}{protocols}

\newaliascnt{algo}{theorem}
\newtheorem{algo}[algo]{Algorithm}
\aliascntresetthe{algo}
\crefname{algo}{algorithm}{algorithms}

\newaliascnt{expr}{theorem}
\newtheorem{expr}[expr]{Experiment}
\aliascntresetthe{expr}
\crefname{experiment}{experiment}{experiments}

\newcommand \mycaption {\small }     
\newcommand \mylabel {}

\ifdefined\IsLLNCS
    \newenvironment{nfbox}[3]{
    \renewcommand \mycaption {#1}
    \renewcommand \mylabel {#2}
    \begin{center}\small
    \begin{tabular}{|ll|}
    \hline
    \hspace{.3ex}
    \begin{minipage}{.97\linewidth}
         \vspace{0.5ex}
         #3}
         {\smallskip
         \captionof{figure}{\mycaption}
         \label{\mylabel}
     \end{minipage}
     &\hspace{.3ex} \\
     \hline
     \end{tabular}
     \end{center}    
    }
\else
    \newenvironment{nfbox}[3]{
    \renewcommand \mycaption {#1}
    \renewcommand \mylabel {#2}
    \begin{center}\small
    \begin{tabular}{|ll|}
    \hline
    \hspace{.3ex}
    \begin{minipage}{.94\linewidth}
         \vspace{0.5ex}
         #3}
         {\smallskip
         \captionsetup{type=figure}
     \end{minipage}
     &\hspace{-0ex} \\
     \hline
     \end{tabular}
     \captionof{figure}{\mycaption}
     \label{\mylabel}
     \end{center}    
    }
\fi

\newcommand{\Gen}{\ensuremath{\textsf{Gen}}\xspace}
\newcommand{\Sign}{\ensuremath{\textsf{Sign}}\xspace}
\newcommand{\Verify}{\ensuremath{\textsf{Verify}}\xspace}

\newcommand{\Share}{\ensuremath{\textsf{Share}}\xspace}
\newcommand{\Recon}{\ensuremath{\textsf{Recon}}\xspace}

\newcommand{\maj}{\MathAlgX{majority}}

\newcommand{\vect}[1]{{ \boldsymbol #1}}
\newcommand{\vx}{\vect{x}}
\newcommand{\vy}{\vect{y}}
\newcommand{\vr}{\vect{r}}
\newcommand{\vrho}{\vect{\rho}}
\newcommand{\vs}{\vect{s}}
\newcommand{\vS}{\vect{S}}

\newcommand{\Adv}{{\ensuremath{\cal A}}\xspace}
\newcommand{\Env}{{\ensuremath{\cal Z}}\xspace}
\newcommand{\Sim}{{\ensuremath{\cal S}}\xspace}

\newcommand{\Party}{{\ensuremath{P}}\xspace}
\newcommand{\DParty}{{\tilde{\Party}}\xspace}
\newcommand{\QParty}{{\ensuremath{Q}}\xspace}
\newcommand{\DQParty}{{\tilde{\QParty}}\xspace}
\newcommand{\secParam}{{\ensuremath{\kappa}}\xspace}
\newcommand{\aux}{{\ensuremath{z}}\xspace}

\newcommand{\AdvneHonest}{\Adv^{\textsf{honest-}\is}_n\xspace}
\newcommand{\AdvneCorrupt}{\Adv^{\textsf{corrupt-}\is}_n\xspace}
\newcommand{\graph}{G}
\newcommand{\Gend}{\graph_{\textsf{end}}\xspace}
\newcommand{\GendC}{\Gend^{\corrupt}\xspace}
\newcommand{\Gred}{\graph_{\textsf{red}}\xspace}
\newcommand{\Gblue}{\graph_{\textsf{blue}}\xspace}
\newcommand{\Greal}{\graph_{\textsf{real}}\xspace}

\newcommand{\GphaseII}{\graph_{\textsf{phaseII}}\xspace}
\newcommand{\GphaseIIH}{\GphaseII^{\honest}\xspace}
\newcommand{\GphaseIIC}{\GphaseII^{\corrupt}\xspace}
\newcommand{\GphaseIII}{\graph_{\textsf{phaseIII}}\xspace}
\newcommand{\GphaseIIIH}{\GphaseIII^{\honest}\xspace}
\newcommand{\GphaseIIIC}{\GphaseIII^{\corrupt}\xspace}
\newcommand{\round}{\rho\xspace}
\newcommand{\roundphaseII}{\round_{\textsf{phaseII}}\xspace}
\newcommand{\roundphaseIIH}{\roundphaseII^{\honest}\xspace}
\newcommand{\roundphaseIIC}{\roundphaseII^{\corrupt}\xspace}
\newcommand{\roundphaseIII}{\round_{\textsf{phaseIII}}\xspace}
\newcommand{\roundphaseIIIH}{\roundphaseIII^{\honest}\xspace}
\newcommand{\roundphaseIIIC}{\roundphaseIII^{\corrupt}\xspace}

\newcommand{\honest}{\textsf{honest}\xspace}
\newcommand{\corrupt}{\textsf{corrupt}\xspace}
\newcommand{\GredH}{\Gred^\honest\xspace}

\newcommand{\GblueH}{\Gblue^\honest\xspace}

\newcommand{\CS}{{\ensuremath{\mathcal{C}}}\xspace}
\newcommand{\E}{{\ensuremath{\mathcal{E}}}\xspace}

\newcommand{\IS}{{\ensuremath{\mathcal{I}}}\xspace}

\newcommand{\PS}{{\ensuremath{\mathcal{P}}}\xspace}

\newcommand{\committee}{{\ensuremath{\mathcal{C}}}\xspace}

\newcommand{\emptystr}{{\ensuremath{\epsilon}}\xspace}

\newcommand{\SMbox}[1]{\mbox{\scriptsize {\sc #1}}}
\newcommand{\VIEW}{\SMbox{VIEW}}
\newcommand{\REAL}{\SMbox{REAL}}
\newcommand{\IDEAL}{\SMbox{IDEAL}}
\newcommand{\HYBRID}{\SMbox{HYBRID}}
\newcommand{\HYB}{\SMbox{HYB}}

\newcommand{\inputCoins}{\textsc{InputsAndCoins}}
\newcommand{\redExec}{\textsc{RedExec}}
\newcommand{\blueExec}{\textsc{BlueExec}}
\newcommand{\finalCut}{\textsc{FinalCut}^\corrupt}
\newcommand{\Ind}{\textsc{Ind}}
\newcommand{\lastParty}{\textsc{LastParty}}

\newcommand{\bigbrack}{{\vphantom{2^{2^2}}}}
\mathchardef\mhyphen="2D

\newcommand{\ith}{$i$'th\xspace}
\newcommand{\jth}{$j$'th\xspace}

\newcommand{\prot}[1]{\pi^{\MathAlgX{#1}}}
\newcommand{\protne}{\prot{ne}}

\newcommand{\protane}{\prot{a\mhyphen ne}}

\newcommand{\func}[1]{f_{\MathAlgX{#1}}}
\newcommand{\felectshare}{\func{elect \mhyphen share}}
\newcommand{\felectsharefull}{\felectshare^{(t',n')}}

\newcommand{\freconcompute}{\func{recon\mhyphen compute}}
\newcommand{\freconcomputefull}{\freconcompute^{(\vk[1],\ldots,\vk[m])}}

\newcommand{\foutdist}{\func{out \mhyphen dist}}
\newcommand{\foutdistfull}{\foutdist^{\committee_1}}
\newcommand{\fitreconcompute}{\func{it\mhyphen recon\mhyphen compute}}
\newcommand{\faelectshare}{\func{a\mhyphen elect \mhyphen share}}
\newcommand{\faelectsharefull}{\faelectshare^{(t',n')}}
\newcommand{\fareconcompute}{\func{a\mhyphen recon\mhyphen compute}}
\newcommand{\fareconcomputefull}{\fareconcompute^{(\vk[1],\ldots,\vk[m])}}
\newcommand{\faoutdist}{\func{a\mhyphen out \mhyphen dist}}
\newcommand{\faoutdistfull}{\faoutdist^{(\vk[1],\ldots,\vk[m])}}
\newcommand{\fpki}{\func{pki}}

\newcommand{\fitpki}{\func{it\mhyphen pki}}
\newcommand{\fpsmt}{\func{psmt}}        
\newcommand{\lpsmt}{l_\MathAlgX{psmt}}  

\newcommand{\sigcals}{\ensuremath{\ell_S}\xspace}
\newcommand{\vercals}{\ensuremath{\ell_V}\xspace}
\newcommand{\sk}[1][\relax]{\ensuremath{\texttt{sk}_{#1}}\xspace}
\newcommand{\vk}[1][\relax]{\ensuremath{\texttt{vk}_{#1}}\xspace}
\newcommand{\vvk}[1][\relax]{\ensuremath{\vec{\vk}_{#1}}\xspace}
\newcommand{\outp}{\ensuremath{\rightarrow}}
\newcommand{\oSig}[1][\relax]{\ensuremath{\mathcal{O}^S_{#1}}\xspace}
\newcommand{\oVer}[1][\relax]{\ensuremath{\mathcal{O}^V_{#1}}\xspace}
\newcommand{\Field}{\ensuremath{F}\xspace}
\newcommand{\negl}{\textsf{negl}}

\newcommand{\extend}[3]{\textsf{extend}^{#2 \hookrightarrow #1}(#3)}

\newcommand{\comp}[1]{\bar{#1}}

\newcommand{\stepref}[1]{Step~\ref{#1}}

\newcommand{\edges}{\textsf{edges}}

\newcommand{\type}{\textsf{type}}

\newcommand{\full}{\textsf{full}}

\newcommand{\outgoing}{\textsf{out}}
\newcommand{\incoming}{\textsf{in}}

\newcommand{\instance}{\textsf{instance}}
\newcommand{\adaptive}{\textsf{adaptive}}

\newcommand{\phaseI}{Phase~I\xspace}
\newcommand{\phaseII}{Phase~II\xspace}
\newcommand{\phaseIII}{Phase~III\xspace}

\newcommand{\island}{U\xspace}
\newcommand{\Visland}{V\xspace}

\newcommand{\samp}{\mathsf{Samp}}

\newcommand{\mutualinfo}[1]{I\left(#1\right)}

\newcommand{\Damgard}{Damg{\aa}rd\xspace}

\title{Must the Communication Graph of MPC Protocols\\ be an Expander?\thanks{A preliminary version of this work appeared at \emph{CRYPTO 2018}~\cite{BCDH18}.}
}

\ifdefined\IsAnon
    \author{}
    \date{}
    \LLNCS{ \institute{}}
\else

    \ifdefined\IsLLNCS
        \author{Elette Boyle\inst{1}\thanks{Supported in part by ISF grant 1861/16, AFOSR Award FA9550-17-1-0069, and ERC Grant no. 307952.}
        \and Ran Cohen\inst{2}\thanks{Supported in part by Alfred P. Sloan Foundation Award 996698, ISF grant 1861/16, ERC starting grant 638121, NEU Cybersecurity and Privacy Institute, and NSF TWC-1664445.}$^\ddag$
        \and Deepesh Data\inst{3}\thanks{This work was done in part while visiting at the FACT Center at IDC Herzliya.}
        \and Pavel Hub{\'{a}}{\v{c}}ek\inst{4}\thanks{Supported by the project 17-09142S of GA \v{C}R, Charles University project UNCE/SCI/004, and Charles University project PRIMUS/17/SCI/9. This work was done under financial support of the Neuron Fund for the support of science.}$^\ddag$
        }
        \institute{IDC Herzliya\\ \email{elette.boyle@idc.ac.il}
        \and MIT and Northeastern University\\ \email{rancohen@ccs.neu.edu}
        \and UCLA\\ \email{deepeshdata@ucla.edu}
        \and Computer Science Institute, Charles University, Prague\\ \email{hubacek@iuuk.mff.cuni.cz}
        }
    \else
        \author{Elette Boyle\thanks{Reichman University and NTT Research. E-mail: \texttt{elette.boyle@runi.ac.il}. Supported in part by ISF grant 1861/16, AFOSR Award FA9550-17-1-0069, and ERC Grant no. 307952.}
        \and Ran Cohen\thanks{Reichman University. E-mail: \texttt{cohenran@runi.ac.il}. Some of the work was done while the author was at MIT and Northeastern University and supported in part by Alfred P. Sloan Foundation Award 996698, ISF grant 1861/16, ERC starting grant 638121, NEU Cybersecurity and Privacy Institute, and NSF TWC-1664445.}~\footnotemark[6]
        \and Deepesh Data\thanks{Meta Platforms, Inc. E-mail: \texttt{deepesh.data@gmail.com}.}~\footnotemark[6]
        \and Pavel Hub{\'{a}}{\v{c}}ek\thanks{Faculty of Mathematics and Physics, Charles University, Prague, Czech Republic. E-mail: \texttt{hubacek@iuuk.mff.cuni.cz}. Supported by the project 17-09142S of GA \v{C}R, Charles University project UNCE/SCI/004, and Charles University project PRIMUS/17/SCI/9. This work was done under financial support of the Neuron Fund for the support of science.}~\footnote{This work was done in part while visiting at the FACT Center at Reichman University (formerly IDC Herzliya).}
        }
    \fi
\fi

\begin{document}
\sloppy

\maketitle
\thispagestyle{empty}


\begin{abstract}
Secure multiparty computation (MPC) on incomplete communication networks has been studied within two primary models: (1) Where a partial network is fixed a priori, and thus corruptions can occur dependent on its structure, and (2) Where edges in the communication graph are determined dynamically as part of the protocol.
Whereas a rich literature has succeeded in mapping out the feasibility and limitations of graph structures supporting secure computation in the fixed-graph model (including strong classical lower bounds), these bounds do not apply in the latter dynamic-graph setting, which has recently seen exciting new results, but remains relatively unexplored.

In this work, we initiate a similar foundational study of MPC within the dynamic-graph model.
As a first step, we investigate the property of graph \emph{expansion}. All existing protocols (implicitly or explicitly) yield communication graphs which are expanders, but it is not clear whether this is inherent.
Our results consist of two types (for constant fraction of corruptions):

\begin{itemize}
    \item
    Upper bounds: We demonstrate secure protocols whose induced communication graphs are \emph{not} expander graphs, within a wide range of settings (computational, information theoretic, with low locality, even with low locality \emph{and} adaptive security), each assuming some form of input-independent setup.

    \item
    Lower bounds: In the plain model (no setup) with adaptive corruptions, we demonstrate that for certain functionalities, \emph{no} protocol can maintain a non-expanding communication graph against all adversarial strategies. Our lower bound relies only on protocol correctness (not privacy), and requires a surprisingly delicate argument.
\end{itemize}

More generally, we provide a formal framework for analyzing the evolving communication graph of MPC protocols, giving a starting point for studying the relation between secure computation and further, more general graph properties.
\end{abstract}

\ifdefined\IsLLNCS\else
\vfill
\fi


\Tableofcontents


\section{Introduction}

The field of secure multiparty computation (MPC), and more broadly fault-tolerant distributed computation, constitutes a deep and rich literature, yielding a vast assortment of protocols providing strong robustness and even seemingly paradoxical privacy guarantees. A central setting is that of $n$ parties who wish to jointly compute some function of their inputs while maintaining correctness (and possibly input privacy) in the face of adversarial behavior from a constant fraction of corruptions.

Since the original seminal results on secure multiparty computation~\cite{GMW87,BGW88,CCD88,RB89}, the vast majority of MPC solutions to date assume that every party can (and will) communicate with every other party. That is, the underlying point-to-point communication network forms a complete graph. Indeed, many MPC protocols begin directly with every party secret sharing his input across all other parties (or simply sending his input, in the case of tasks without privacy such as Byzantine agreement~\cite{PSL80,LSP82,DS83,FM97,GM93}).

There are two classes of exceptions to this rule, which consider MPC on incomplete communication graphs.

\textbf{Fixed-Graph Model.} The first corresponds to an area of work investigating achievable security guarantees in the setting of a \emph{fixed} partial communication network. In this model, communication is allowed only along edges of a fixed graph, known a priori, and hence where corruptions can take place as a function of its structure. This setting is commonly analyzed within the distributed computing community. In addition to positive results, this is the setting of many fundamental lower bounds: For example, to achieve Byzantine agreement deterministically against $t$ corruptions, the graph must be $(t +1)$-connected~\cite{Dolev82,FLM86}.\footnote{If no setup assumptions are assumed, the connectivity bound increases to $2t+1$.}
For graphs with lower connectivity, the best one can hope for is a form of ``almost-everywhere agreement,'' where some honest parties are not guaranteed to output correctly, as well as restricted notions of privacy~\cite{DkPPU88,GO08,CGO15,HLP11,HIJKR16}.
Note that because of this, one cannot hope to achieve protocols with standard security in this model with $o(n^2)$ communication, even for simple functionalities such as Byzantine agreement.

\textbf{Dynamic-Graph Model.} The second, more recent approach addresses a model where all parties have the \emph{ability} to initiate communication with one another, but make use of only a subset of these edges as determined dynamically during the protocol. We refer to this as the ``dynamic-graph model.'' When allowing for negligible error (in the number of parties), the above lower bounds do not apply, opening the door for dramatically different approaches and improvements in complexity. Indeed, distributed protocols have been shown for Byzantine agreement in this model with as low as $\tilde O(n)$ bits of communication~\cite{KSSV06,BGH13}, and secure MPC protocols have been constructed whose communication graphs have degree $o(n)$---and as low as $\polylog(n)$~\cite{DKMSZ17,BGT13,CCGGOZ15,BCP15}.\footnote{This metric is sometimes referred to as the communication \emph{locality} of the protocol~\cite{BGT13}.}
However, unlike the deep history of the model above, the current status is a sprinkling of positive results.
Little is known about what types of communication graphs must be generated from a secure MPC protocol execution.

Gaining a better understanding of this regime is motivated not only to address fundamental questions, but also to provide guiding principles for future protocol design. In this work, we take a foundational look at the dynamic-graph model, asking:

\begin{quote}
	\centering
	\emph{What properties of induced communication graphs\\ are necessary to support secure computation?}
\end{quote}	

\paragraph{On the necessity of graph expansion.}
Classical results tell us that the fully connected graph suffices for secure computation. Protocols achieving low locality indicate that a variety of significantly sparser graphs, with many low-weight cuts, can also be used~\cite{DKMSZ17,BGT13,CCGGOZ15,BCP15}. We thus consider a natural extension of connectivity to the setting of low degree.
Although the positive results in this setting take different approaches and result in different communication graph structures, we observe that in each case, the resulting sparse graph has high \emph{expansion}.

Roughly, a graph is an expander if every subset of its nodes that is not ``too large'' has a ``large'' boundary. Expander graphs have good mixing properties and in a sense ``mimic'' a fully connected graph. There are various ways of formalizing expansion; in this work we consider a version of \emph{edge} expansion, pertaining to the number of outgoing edges from any subset of nodes. We consider a variant of the expansion definition which is naturally monotonic: that is, expansion cannot decrease when extra edges are added (note that such monotonicity also holds for the capacity of the graph to support secure computation).

Indeed, expander graphs appear explicitly in some works~\cite{KSSV06,CCGGOZ15}, and implicitly in others (e.g., using random graphs~\cite{KS09}, pseudorandom graphs~\cite{BGT13,BCG21}, and averaging samplers~\cite{BGH13}, to convert from almost-everywhere to everywhere agreement). High connectivity and good mixing intuitively go hand-in-hand with robustness against corruptions, where adversarial entities may attempt to impede or misdirect information flow.

This raises the natural question: Is this merely an artifact of a convenient construction, or is high expansion \emph{inherent}? That is, we investigate the question:
Must the communication graph of a generic MPC protocol, tolerating a linear number of corruptions, be an expander graph?

\subsection{Our Results}\label{sec:intro:ourResult}

More explicitly, we consider the setting of secure multiparty computation with $n$ parties in the face of a linear number of active corruptions. As common in the honest-majority setting, we consider protocols that guarantee output delivery.
Communication is modeled via the dynamic-graph setting, where all parties have the ability to initiate communication with one another, and use a subset of edges as dictated by the protocol.
We focus on the synchronous setting, where the protocol proceeds in a round-by-round manner.

\medskip
Our contributions are of the following three kinds:
\paragraph{Formal definitional framework.}
As a first contribution, we provide a formal framework for analyzing and studying the evolving communication graph of MPC protocols. The framework abstracts and refines previous approaches concerning specific properties of protocols implicitly related to the graph structure, such as the degree~\cite{BGT13}. This gives a starting point for studying the relation between secure computation and further, more general, graph properties.

\paragraph{Upper bounds.}
We present secure protocols whose induced communication graphs are decidedly \emph{not} expander graphs, within a range of settings. This includes: with computational security, with information-theoretic security, with low locality, even with low locality \emph{and} adaptive security (in a hidden-channels model~\cite{CCGGOZ15}) --- but all with the common assumption of some form of input-independent \emph{setup} information (such as a \emph{public-key infrastructure}, PKI). The resulting communication graph has a low-weight cut, splitting the $n$ parties into two equal (linear) size sets with only poly-logarithmic edges connecting them.

\begin{theorem}[MPC with non-expanding communication graph, informal]\label{thm:intro:UB}
For any efficient functionality $f$ and any constant $\epsilon>0$, there exists a protocol in the PKI model, assuming digital signatures, securely realizing $f$ against $(1/4-\epsilon) \cdot n$ static corruptions, such that with overwhelming probability the induced communication graph is non-expanding.
\end{theorem}

\cref{thm:intro:UB} is stated in the computational setting with static corruptions; however, this approach extends to various other settings, albeit at the expense of a lower corruption threshold. (See \cref{sec:ne_mpc} for more details.)

\begin{theorem}[extensions of \cref{thm:intro:UB}, informal]\label{thm:intro:UB_extensions}
For any efficient functionality $f$, there exists a protocol securely realizing $f$, in the settings listed below, against a linear number of corruptions, such that with overwhelming probability the induced communication graph is non-expanding:
\begin{itemize}
    \item
    In the setting of \cref{thm:intro:UB} with poly-logarithmic locality.
    \item
    Unconditionally, in the information-theoretic PKI model (with or without low locality).
    \item
    Unconditionally, in the information-theoretic PKI model, facing adaptive adversaries.
    \item
    Under standard cryptographic assumptions, in the PKI model, facing adaptive adversaries, with poly-logarithmic locality.
\end{itemize}
\end{theorem}

As an interesting special case, since our protocols are over point-to-point channels and do not require a broadcast channel, these results yield the first Byzantine agreement protocols whose underlying communication graphs are not expanders.

The results in \cref{thm:intro:UB,thm:intro:UB_extensions} all follow from a central transformation converting existing secure protocols into ones with low expansion.
At a high level, the first $n/2$ parties will run a secure computation to elect two representative committees of poly-logarithmic size: one amongst themselves and the other from the other $n/2$ parties. These committees will form a ``communication bridge'' across the two halves (see \cref{fig:results}). The setup is used to certify the identities of the members of both committees to the receiving parties, either via a public-key infrastructure for digital signatures (in the computational setting) or correlated randomness for information-theoretic signatures~\cite{SHZI02,SASM10} (in the information-theoretic setting).

Interestingly, this committee-based approach can be extended to the adaptive setting (with setup), in the hidden-channels model considered by~\cite{CCGGOZ15}, where the adversary is not aware which communication channels are utilized between honest parties.\footnote{Sublinear locality is impossible in the adaptive setting if the adversary is aware of honest-to-honest communication, since it can simply isolate an honest party from the rest of the protocol.} Here, care must be taken to not reveal more information than necessary about the identities of committee members to protect them from being corrupted.

As a side contribution, we prove the first instantiation of a protocol with poly-logarithmic locality and information-theoretic security (with setup), by adjusting the protocol from~\cite{BGT13} to the information-theoretic setting.

\begin{theorem}[polylog-locality MPC with information-theoretic security, informal]\label{thm:intro:IT_BGT}
For any efficient functionality $f$ and any constant $\epsilon>0$, there exists a protocol with poly-logarithmic locality in the information-theoretic PKI model, securely realizing $f$ against computationally unbounded adversaries statically corrupting $(1/6-\epsilon) \cdot n$ parties.
\end{theorem}

\paragraph{Lower bounds.}
On the other hand, we show that in some settings a weak form of expansion \emph{is} a necessity.
In fact, we prove a stronger statement, that in these settings the graph must have high connectivity.\footnote{More concretely, the graph should be at least $\alpha(n)$-connected for every sublinear function $\alpha(n)\in o(n)$.}
Our lower bound is in the setting of adaptive corruptions, computational (or information-theoretic) security, and with common setup assumptions (but \emph{no} private-coin setup as PKI). Our proof relies only on correctness of the protocol and not on any privacy guarantees; namely, we consider the \emph{parallel broadcast} functionality (aka \emph{interactive consistency}~\cite{PSL80}), where every party distributes its input to all other parties. We construct an adversarial strategy in this setting such that no protocol can guarantee correctness against this adversary if its induced communication graph at the conclusion of the protocol has any cut with sublinear many crossing edges (referred to as a ``sublinear cut'' from now on).

\begin{theorem}[high connectivity is necessary for correct protocols, informal]\label{thm:intro:LB}
Let $t\in\Theta(n)$. Any $t(n)$-resilient protocol for parallel broadcast in the computational setting, even with access to a common reference string, tolerating an adaptive, malicious adversary \emph{cannot} maintain an induced communication graph with a sublinear cut.
\end{theorem}

\noindent
\cref{thm:intro:LB} in particular implies that the resulting communication graph must have a form of expansion.
We note that in a weaker communication model, a weaker form of consensus, namely Byzantine agreement, can be computed in a way that the underlying graph (while still an expander) has low-weight cuts~\cite{KS10}.
We elaborate on the differences between the two settings in the related work, \cref{sec:intro:relatedwork}.

It is indeed quite intuitive that if a sublinear cut exists in the communication graph of the protocol, and the adversary can adaptively corrupt a linear number of parties $t(n)$, then he could corrupt the parties on the cut and block information flow. The challenge, however, stems from the fact that the cut is not known a priori but is only revealed over time, and by the point at which the cut is identifiable, all necessary information may have already been transmitted across the cut. In fact, even the identity of the cut and visible properties of the communication graph itself can convey information to honest parties about input values without actual bits being communicated.
This results in a surprisingly intricate final attack, involving multiple indistinguishable adversaries, careful corruption strategies, and precise analysis of information flow. See below for more detail.

\subsection{Our Techniques}\label{sec:intro:technique}

We proceed to discuss the technical aspects of the lower bound result. We refer the reader to \cref{sec:upperbound_protocol} for an overview of the upper bound result.

\paragraph{Overview of the attack.}

Consider an execution of the parallel broadcast protocol over \emph{random} inputs. At a high level, our adversarial strategy, denoted $\AdvneHonest$, will select a party $\Party_\is$ at random and attempt to block its input from being conveyed to honest parties. We are only guaranteed that somewhere in the graph will remain a sublinear cut. Because the identity of the eventual cut is unknown, it cannot be attacked directly.
We take the following approach:
\begin{enumerate}
    \item \textbf{\phaseI.}
    Rather, our attack will first ``buy time'' by corrupting the neighbors of $\Party_\is$, and blocking information flow of its input $x_\is$ to the remaining parties.
    Note that this can only continue up to a certain point, since the degree of $\Party_\is$ will eventually surpass the corruption threshold (as we prove).
    But, the benefit of this delay is that in the meantime, the communication graph starts to fill in, which provides more information about the locations of the potential cuts.

    For this to be the case, it must be that the parties cannot identify that $\Party_\is$ is under attack (otherwise, the protocol may instruct many parties to quickly communicate to/from $\Party_\is$, forcing the adversary to run out of his ``corruption budget'' before the remaining graph fills in).
    The adversary thus needs to fool all honest parties and make each honest party believe that he participates in an honest execution of the protocol.
    This is done by maintaining two simulated executions: one pretending to be $\Party_\is$ running on a random input, and another pretending (to $\Party_\is$) to be all other parties running on random inputs.
    Note that for this attack strategy to work it is essential that the parties do not have pre-computed private-coin setup such as PKI.
	
    \item \textbf{\phaseII.}
    We show that with noticeable probability, by the time we run out of the \phaseI corruption threshold (which is a linear number of parties), \emph{all parties} in the protocol have high (linear) degree. In turn, we prove that the current communication graph can have at most a constant number of sublinear cuts.
	
    In the remainder of the protocol execution, the adversary will simultaneously attack all of these cuts. Namely, he will block information flow from $\Party_\is$ across any of these cuts by corrupting the appropriate ``bridge'' party, giving up on each cut one by one when a certain threshold of edges have already crossed it.
\end{enumerate}

\noindent
If the protocol is guaranteed to maintain a sublinear cut, then necessarily there will remain at least one cut for which all \phaseII communication across the cut has been blocked by the adversary.
Morally, parties on the side of this cut opposite $\Party_\is$ should not have learned $x_\is$, and thus the \emph{correctness} of the protocol should be violated.
Proving this, on the other hand, requires surmounting two notable challenges.
\begin{enumerate}
    \item
    We must prove that there still remains an uncorrupted party $\Party_\js$ on the opposite side of the cut. It is not hard to show that each side of the cut is of linear size, that $\Party_\is$ has a sublinear number of neighbors across the cut (all of which are corrupted), and that a sublinear number of parties get corrupted in \phaseII. Hence, there exists parties across the cut that are not neighbors of $\Party_\is$ and that are not corrupted in \phaseII. However, by the attack strategy, all of the neighbors of the \emph{virtual} $\Party_\is$ are corrupted in \phaseI as well, and this is also a linear size set, which is independent of the real neighbors of $\Party_\is$. Therefore, it is not clear that there will actually remain honest parties across the cut by the end of the protocol execution.
    \item
    More importantly, even though we are guaranteed that no bits of communication have been passed along any path from $\Party_\is$ to $\Party_\js$, this does not imply that no \emph{information} about $x_\is$ has been conveyed. For example, since the graph develops as a function of parties' inputs, it might be the case that this situation of $\Party_\js$ being blocked from $\Party_\is$, only occurs when $x_\is$ equals a certain value.
\end{enumerate}

We now discuss how these two challenges are addressed.

\paragraph{Guaranteeing honest parties across the cut.}
Unexpectedly, we cannot guarantee existence of honest parties across the cut.
Instead, we introduce a different adversarial strategy, which we prove \emph{must} have honest parties blocked across a cut from $\Party_\is$, and for which there exist honest parties who cannot distinguish which of the two attacks is taking place.
More explicitly, we consider the ``dual'' version of the original attack, denoted $\AdvneCorrupt$, where party $\Party_\is$ is \emph{corrupted} and instead pretends to be under attack as per $\AdvneHonest$ above.

Blocking honest parties from $x_\is$ in $\AdvneCorrupt$ does not contradict correctness explicitly on its own, as $\Party_\is$ is corrupted in this case. It is the combination of both of these attacks that will enable us to contradict correctness.
Namely, we prove that:
\begin{itemize}
    \item
    Under the attack $\AdvneCorrupt$, there exists a ``blocked cut'' $(S,\comp{S})$ with uncorrupted parties on both sides. By \emph{agreement}, all uncorrupted parties output the same value $y_\is$ as the $\is$'th coordinate of the output vector.
    \item
    The view of some of the uncorrupted parties under the attack $\AdvneCorrupt$ is identically distributed as that of uncorrupted parties in the original attack $\AdvneHonest$. Thus, their output distribution must be the same across the two attacks.
    \item
    Since under the attack $\AdvneHonest$, the party $\Party_\is$ is honest, by \emph{completeness}, all uncorrupted parties in $\AdvneHonest$ must output the \emph{correct} value $y_\is=x_\is$.
    \item
    Thus, uncorrupted parties in $\AdvneCorrupt$ (who have the same view) must output the correct value $x_\is$ as well.
\end{itemize}

\noindent
Altogether, this implies all honest parties in interaction with $\AdvneCorrupt$, in particular $\Party_\js$ who is blocked across the cut from $\Party_\is$, must output $y_\is=x_\is$.

\paragraph{Bounding information transmission about $x_\is$.}
The final step is to show that this cannot be the case, since an uncorrupted party $\Party_\js$ across the cut in $\AdvneCorrupt$ does not receive enough information about $x_\is$ to fully specify the input.
This demands delicate treatment of the specific attack strategy and analysis, as many ``side channel'' signals within the protocol can leak information on $x_\is$.
Corruption patterns in \phaseII, and their timing, can convey information ``across'' the isolated cut.
In fact, even the event of successfully reaching \phaseII may be correlated with the value of $x_\is$.

For example, say the cut at the conclusion of the protocol is $(S_1,\comp{S}_1)$ with $\is\in S_1$ and $\js\in \comp{S}_1$, but at the beginning of \phaseII there existed another cut $(S_2,\comp{S}_2)$, for which $S_1 \cap S_2 \neq \emptyset$, $S_1 \cap \comp{S}_2 \neq \emptyset$, $\comp{S}_1 \cap S_2 \neq \emptyset$, and $\comp{S}_1 \cap \comp{S}_2 \neq \emptyset$. Since any ``bridge'' party in $\comp{S}_2$ that receives a message from $S_2$, gets corrupted and discards the message, the view of honest parties in $\comp{S}_1$ might change as a result of the corruption related to the cut $(S_2,\comp{S}_2)$, which in turn could depend on $x_\is$.
See \cref{fig:manycuts} for an illustration of this situation.

Ultimately, we ensure that the final view of $\Party_\js$ in the protocol can be simulated given only ``\phaseI'' information, which is independent of $x_\is$, in addition to the identity of the final cut in the graph, which reveals only a constant amount of additional entropy.

\begin{figure}[htb]
	\begin{center}
		\includegraphics[scale=0.6]{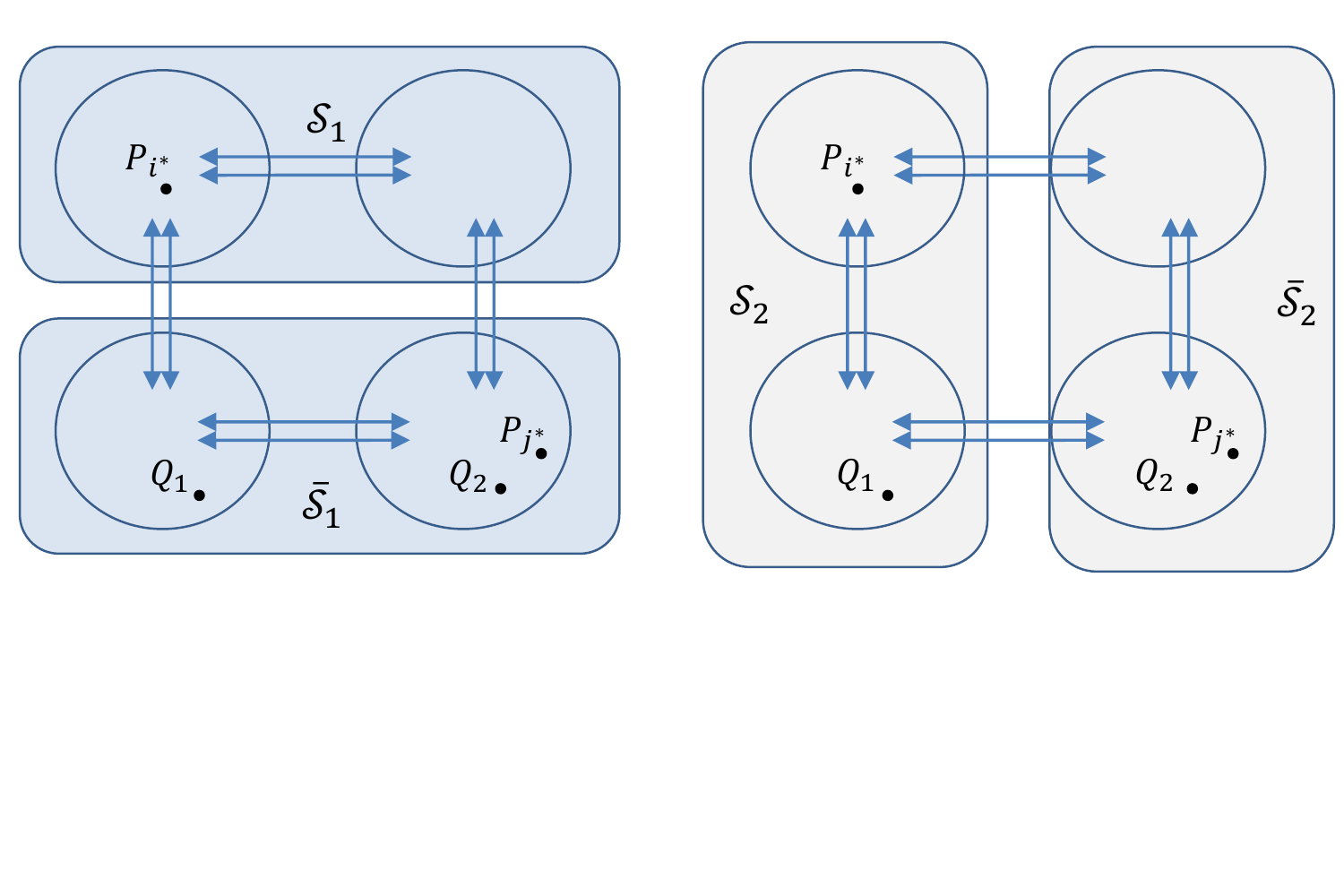}
	\end{center}
	\vspace{-15ex}
	\caption{At the end of \phaseI the communication graph is partitioned into 4 linear-size ``islands'' that are connected by sublinear many edges. On the left is the potential cut $(S_1,\comp{S}_1)$ and on the right the potential cut $(S_2,\comp{S}_2)$. If party $Q_1$ sends a message to party $Q_2$ then $Q_2$ gets corrupted and discards the message. This event can be identified by all parties in $\comp{S}_1$, and in particular by $\Party_\js$ by the end of the protocol.}
	\label{fig:manycuts}
\end{figure}

\paragraph{Additional subtleties.}
The actual attack and its analysis are even more delicate.
For example, it is important that the degree of the ``simulated $\Party_\is$,'' by the adversarial strategy $\AdvneHonest$, will reach the threshold faster than the real $\Party_\is$. In addition, in each of these cases, the threshold, and so the transition to the next phase, could possibly be reached in a middle of a round, requiring detailed treatment.

\subsection{Open Questions}

This work leaves open many interesting lines of future study.
\begin{itemize}
	\item
    Bridging the gap between upper and lower bounds. This equates to identifying the core properties that necessitate graph expansion versus not. Natural candidates suggested by our work are existence of setup information and adaptive corruptions in the hidden or visible (yet private) channels model.
	\item
    What other graph properties are necessary (or not) to support secure computation? Our new definitional framework may aid in this direction.
    \item
    Our work connects graph theory and secure protocols, giving rise to further questions and design principles. For example, can good constructions of expanders give rise to new communication-efficient MPC? On the other hand, can necessity of expansion (in certain settings) be used to argue new communication complexity lower bounds?
\end{itemize}

\subsection{Additional Related Work}\label{sec:intro:relatedwork}

Communication graphs induced by fault-tolerant protocols is a field that has been intensively studied in various aspects.

In the fixed-graph model, where the parties can communicate over a pre-determined partial graph, there have been many work for realizing secure message transmission~\cite{DDWY93,FY04,SA96,BF99,KGSR02,SNR04,BJLM06,BM05,ACH06,FFGV07,KurSuz09,SZ16}, Byzantine agreement~\cite{Dolev82,DkPPU88,BG93,Upfal92}, and secure multiparty computation~\cite{Beimel07,BJLM11,BJLM06,CGO15,CGO10,CGO12}.

In the setting of \emph{topology-hiding} secure computation (THC)~\cite{MOR15,HMTZ16,AM17,ALM17,BBMM18}, parties communicate over a partial graph, and the goal is to hide which pairs of honest parties are neighbors.
This is a stronger property than considered in this work, as we do not aim to hide the topology of the graph (in particular, the entire communication graph can be revealed by the conclusion of the protocol).
Intuitively, topology-hiding protocols can support non-expanding graphs since sublinear cuts should not be revealed during the protocol.
A followup work~\cite{BBCMM19} explored this connection, and showed that the classical definition of THC is in fact too strong to hide sublinear cuts in the graph, and demonstrated how to capture this property by a weaker definition called distributional-topology-hiding computation.
The separation between these security definitions is based on the distribution of graphs induced by the protocol in \cref{sec:communication_graph_expander}.

Another direction to study the connection between MPC and graph theory, explored in~\cite{HIK07,KRS16}, is to consider MPC protocols that are based on oblivious transfer (OT), and to analyze the graph structure that is induced by all pairwise OT channels that are used by the protocol.

King and Saia \cite{KS10} presented a Byzantine agreement protocol that is secure against adaptive corruptions and (while still being an expander) its communication graph has sublinear cuts. Compared to our lower bound (\cref{sec:LB_Expander}), both results do not assume any trusted setup and both consider adaptive corruptions. However, we highlight three aspects in which the setting in~\cite{KS10} is weaker. First, \cite{KS10} realize \emph{Byzantine agreement} which is a weaker functionality than parallel broadcast; indeed, the standard techniques of reducing broadcast to Byzantine agreement requires the sender to first send its input to all other parties who then run BA --- when every party acts as the sender this implies a complete communication graph. Second, \cite{KS10} assume hidden channels where the adversary is unaware of honest-to-honest communication. And third, \cite{KS10} consider \emph{atomic message delivery}, meaning that once a party has sent a message to the network, the adversary cannot change the content of the message even by corrupting the sender and before any honest party received it (for more details see~\cite{GKKZ11}, where atomic message delivery was used to overcome the lower bound of~\cite{HZ10}).

Abraham et al.\ \cite{ACDNPRS19} showed that when the adversary can remove messages of corrupted parties after the fact (\ie without assuming atomic message delivery), adaptively secure BA requires every honest party to communicate with $\Omega(t)$ neighbors; \ie every honest party has linear locality. We note that this does not imply our result since a linear degree does not rule out the existence of a sublinear cut.

\subsection*{Paper Organization}
Basic notations are presented in \cref{sec:Preliminaries}.
In \cref{sec:communication_graph_expander}, we provide our formalization of the communication graph induced by an MPC protocol and related properties.
In \cref{sec:ne_mpc}, we describe our upper bound results, constructing protocols with non-expanding graphs.
In \cref{sec:LB_Expander}, we prove our lower bound.
We defer general preliminaries and further details to the appendix.


\section{Preliminaries}\label{sec:Preliminaries}

\paragraph{Notations.}
In the following we introduce some necessary notation and terminology.
For $n,n_1,n_2\in\mathbb{N}$, let $[n]=\sset{1,\cdots,n}$ and $[n_1,n_2]=\sset{n_1,\cdots,n_2}$. We denote by $\secParam$ the security parameter. Let $\poly$ denote the set of all positive polynomials and let PPT denote a probabilistic algorithm that runs in \emph{strictly} polynomial time. A function $\nu \colon \mathbb{N} \rightarrow \mathbb{R}$ is \textit{negligible} if $\nu(\secParam)<1/p(\secParam)$ for every $p\in\poly$ and sufficiently large $\secParam$.
Given a random variable $X$, we write $x\gets X$ to indicate that $x$ is selected according to $X$.
The statistical distance between two random variables $X$ and $Y$ over a finite set $\Uni$, denoted $\SD(X,Y)$, is defined as $\frac12 \cdot \sum_{u\in \Uni}\size{\pr{X = u}- \pr{Y = u}}$.

Two distribution ensembles $X=\set{X(a,\secParam)}_{a\in\zs,\secParam\in\N}$ and $Y=\set{Y(a,\secParam)}_{a\in\zs,\secParam\in\N}$ are \textsf{computationally indistinguishable} (denoted $X\compindist Y$) if for every non-uniform polynomial-time distinguisher $\Adv$ there exists a function $\nu(\secParam) = \negl(\secParam)$, such that for every $a\in\zs$ and every $\secParam$,
\[
\abs{\pr{\Adv(X(a,\secParam),1^\secParam)=1} - \pr{\Adv(Y(a,\secParam),1^\secParam)=1}}\leq \nu(\secParam).
\]
The distribution ensembles $X$ and $Y$ are \textsf{statistically close} (denoted $X\statclose Y$) if for every $a\in\zs$ and every $\secParam$ it holds that $\SD(X(a,\secParam),Y(a,\secParam))\leq \nu(\secParam)$.

\paragraph{Graph-theoretic notations.}
Let $G=(V,E)$ be an undirected graph of size $n$, \ie $\ssize{V}=n$.
Given a set $S\subseteq V$, we denote its complement set by $\comp{S}$, \ie $\comp{S}=V\setminus S$.
Given two disjoint subsets $U_1,U_2\subseteq V$ define the set of all the edges in $G$ for which one end point is in $U_1$ and the other end point is in $U_2$ as
\[
\edges_G(U_1,U_2)\assign \set{(u_1,u_2):u_1\in U_1, u_2\in U_2, \text{ and } (u_1,u_2)\in E}.
\]
We denote by $\ssize{\edges_G(U_1,U_2)}$ the total number of edges going across $U_1$ and $U_2$.
For simplicity, we denote $\edges_G(S)=\edges_G(S,\comp{S})$.
A \textsf{cut} in the graph $G$ is a partition of the vertices $V$ into two non-empty, disjoint sets $\sset{S,\comp{S}}$.
The \textsf{weight} of a cut $\sset{S,\comp{S}}$ is defined to be $\ssize{\edges_G(S)}$.
An $\alpha$-cut is a cut $\sset{S,\comp{S}}$ whose weight is smaller than $\alpha$, \ie such that $\ssize{\edges_G(S)}\leq\alpha$.

Given a graph $G=(V,E)$ and a node $i\in V$, denote by $G\setminus\sset{i}=(V',E')$ the graph obtained by removing node $i$ and all its edges, \ie $V'=V\setminus\sset{i}$ and $E'=E\setminus\sset{(i,j) \mid j\in V'}$.

\paragraph{MPC Model.}
We consider multiparty protocols in the stand-alone, synchronous model, and require security with guaranteed output delivery.
We elaborate on the model in \cref{sec::Def:model}, and refer the reader to \cite{Canetti00,Goldreich04} for a precise definition of the model.
Throughout the paper we assume malicious adversaries that can deviate from the protocol in an arbitrary manner. We will consider both \emph{static} corruptions, where the set of corrupted parties is fixed at the onset of the protocol, and \emph{adaptive} corruptions, where the adversary can dynamically corrupt parties during the protocol execution, In addition, we will consider both PPT adversaries and computationally unbounded adversaries

Recall that in the synchronous model protocols proceed in rounds, where every round consists of a \emph{send phase} followed by a \emph{receive phase}.
The adversary is assumed to be \emph{rushing}, meaning that he can determine the messages for corrupted parties \emph{after} seeing the messages sent by the honest parties. We assume a complete network of point-to-point channels (broadcast is not assumed), where every party has the ability to send a message to every other party.
We will normally consider \emph{secure} (private) channels where the adversary learns that a message has been sent between two honest parties, but not its content. If a public-key encryption is assumed, this assumption can be relaxed to \emph{authenticated} channels, where the adversary can learn the content of all messages (but not change them). For our upper bound in the adaptive setting (\cref{sec:ne_mpc_adaptive}) we consider \emph{hidden} channels (as introduced in~\cite{CCGGOZ15}), where the adversary does not even know whether two honest parties have communicated or not.

\section{Communication Graphs Induced by MPC Protocols}\label{sec:communication_graph_expander}

In this section, we present formal definitions of properties induced by the communication graph of interactive protocols. These definitions are inspired by previous works in distributed computing~\cite{KSSV06,KKKSS08,KS10,KLST11} and multiparty computation~\cite{BGT13,CCGGOZ15,BCP15} that constructed interactive protocols with \emph{low locality}.

\subsection{Ensembles of Protocols and Functionalities}

In order to capture certain asymptotic properties of the communication graphs of generic $n$-party protocols, such as edge expansion and locality, it is useful to consider a family of protocols that are parametrized by the number of parties $n$. This is implicit in many distributed protocols and in generic multiparty protocols, for example~\cite{PSL80,LSP82,DS83,GMW87,BGW88}. We note that for many large-scale protocols, \eg protocols with low locality~\cite{KSSV06,KKKSS08,KS10,KLST11,BGT13,BCP15}, the security guarantees increase with the number of parties, and in fact, the number of parties is assumed to be polynomially related to the security parameter.

\begin{definition}[protocol ensemble]
Let $f=\sset{f_n}_{n\in\N}$ be an ensemble of functionalities, where $f_n$ is an $n$-party functionality, let $\pi=\sset{\pi_n}_{n\in\N}$ be an ensemble of protocols, and let $\CS=\sset{\CS_n}_{n\in\N}$ be an ensemble of classes of adversaries (\eg $\CS_n$ is the class of PPT $t(n)$-adversaries). We say that $\pi$ securely computes $f$ tolerating adversaries in $\CS$ if for every $n$ that is polynomially related to the security parameter $\secParam$, it holds that $\pi_n$ securely computes $f_n$ tolerating adversaries in $\CS_n$.
\end{definition}

In \cref{sec:ne_mpc}, we will consider several classes of adversaries. We use the following notation for clarity and brevity.
\begin{definition}
Let $f=\sset{f_n}_{n\in\N}$ be an ensemble of functionalities and let $\pi=\sset{\pi_n}_{n\in\N}$ be an ensemble of protocols.
We say that $\pi$ securely computes $f$ tolerating adversaries of the form $\type$ (\eg static PPT $t(n)$-adversaries, adaptive $t(n)$-adversaries, etc.), if $\pi$ securely computes $f$ tolerating adversaries in $\CS=\sset{\CS_n}_{n\in\N}$, where for every $n$, the set $\CS_n$ is the class of adversaries of the form $\type$.
\end{definition}

\subsection{The Communication Graph of a Protocol's Execution}

Intuitively, the communication graph induced by a protocol should include an edge $(i,j)$ precisely if parties $\Party_i$ and $\Party_j$ exchange messages during the protocol execution.
For instance, consider the property of \emph{locality}, corresponding to the maximum degree of the communication graph.
When considering malicious adversaries that can deviate from the protocol using an arbitrary strategy, it is important to consider only messages that are sent by honest parties and messages that are received by honest parties. Otherwise, every corrupted party can send a message to every other corrupted party, yielding a subgraph with degree $\Theta(n)$. We note that restricting the analysis to only consider honest parties is quite common in the analysis of protocols.

Another issue that must be taken under consideration is flooding by the adversary. Indeed, there is no way to prevent the adversary from sending messages from all corrupted parties to all honest parties; however, we wish to only count those message which are actually processed by honest parties. To model this, the \emph{receive} phase of every communication round\footnote{Recall that in the synchronous model, every communication round is composed of a \emph{send} phase and a \emph{receive} phase, see \cref{sec::Def:model}.} is composed of two sub-phases:
\begin{enumerate}
	\item \emph{The filtering sub-phase:}
    Each party inspects the list of messages received in the previous round, according to specific filtering rules defined by the protocol, and discards the messages that do not pass the filter.
	The resulting list of messages is appended to the local transcript of the protocol.
    \item \emph{The processing sub-phase:}
    Based on its local transcript, each party computes the next-message function and obtains the list of messages to be sent in the current round along with the list of recipients, and sends them to the relevant parties.
\end{enumerate}

\noindent
In practice, the filtering procedure should be ``lightweight,'' such as verifying validity of a signature. However, we assume only an abstraction and defer the actual choice of filtering procedure (as well as corresponding discussion) to specific protocol specifications.
We note that the above two-phase processing of rounds is implicit in protocols from the literature that achieve low locality~\cite{KSSV06,KKKSS08,KS10,KLST11,BGT13,CCGGOZ15,BCP15}. It is also implicit when analyzing the communication complexity of general protocols, where malicious parties can send long messages to honest parties, and honest parties filter out invalid messages before processing them.

We now turn to define the communication graph of a protocol's execution, by which we mean the deterministic instance of the protocol defined by fixing the adversary and all input values and random coins of the parties and the adversarial strategy. We consider protocols that are defined in the correlated-randomness model (\eg for establishing PKI). This is without loss of generality since by defining the ``empty distribution,'' where every party is given an empty string, we can model also protocols in the plain model.
Initially, we focus on the \emph{static} setting, where the set of corrupted parties is determined at the onset of the protocol.
In \cref{sec:def_adaptive}, we discuss the \emph{adaptive} setting.

\begin{definition}[protocol execution instance]\label{def:protocol_instance}
For $n \in \N$, let $\pi_n$ be an $n$-party protocol, let $\secParam$ be the security parameter, let $\vx=(x_1,\ldots,x_n)$ be an input vector for the parties, let $\vrho=(\rho_1,\ldots,\rho_n)$ be correlated randomness for the parties, let $\Adv$ be an adversary, let $\aux$ be the auxiliary information of $\Adv$, let $\IS\subseteq[n]$ be the set of indices of corrupted parties controlled by $\Adv$, and let $\vr=(r_1,\ldots,r_n,r_\Adv)$ be the vector of random coins for the parties and for the adversary.

Denote by $\instance(\pi_n)=(\pi_n,\Adv,\IS,\secParam,\vx,\vrho,\aux,\vr)$ the list of parameters that deterministically define an \textsf{execution instance of the protocol $\pi_n$}.
\end{definition}

Note that $\instance(\pi_n)$ fully specifies the entire views and transcript of the protocol execution, including all messages sent to/from honest parties.

\begin{definition}[communication graph of protocol execution]\label{def:comm_graph}
For $n \in \N$, let $\instance(\pi_n)=(\pi_n,\Adv,\IS,\secParam,\vx,\vrho,\aux,\vr)$ be an execution instance of the protocol $\pi_n$. We now define the following communication graphs induced by this execution instance. Each graph is defined over the set of $n$ vertices $[n]$.
\begin{itemize}
    \item\emph{Outgoing communication graph.}
    The directed graph $G_\outgoing(\instance(\pi_n))=([n],E_\outgoing)$ captures all the communication lines that are used by honest parties to send messages. That is,
    \[
    E_\outgoing(\instance(\pi_n)) = \set{(i,j) \mid \Party_i \text{ is honest and sent a message to } \Party_j}.
    \]

    \item\emph{Incoming communication graph.}
    The directed graph $G_\incoming(\instance(\pi_n))=([n],E_\incoming)$ captures all the communication lines in which honest parties received messages that were processed (\ie excluding messages that were filtered out). That is,
    \[
    E_\incoming(\instance(\pi_n)) = \set{(i,j) \mid \Party_j \text{ is honest and processed a message received from } \Party_i}.
    \]

    \item\emph{Full communication graph.}
    The undirected graph $G_\full(\instance(\pi_n))=([n],E_\full)$ captures all the communication lines in which honest parties received messages that were processed, or used by honest parties to send messages.
    That is,
    \[
    E_\full(\instance(\pi_n)) = \set{(i,j) \mid (i,j)\in E_\outgoing \text{ or } (i,j)\in E_\incoming}.
    \]
\end{itemize}
\end{definition}
We will sometimes consider ensembles of protocol instances (for $n\in\N$) and the corresponding ensembles of graphs they induce.

Looking ahead, in subsequent sections we will consider the full communication graph $G_\full$.
Apart from making the presentation clear, the graphs $G_\outgoing$ and $G_\incoming$ are used for defining $G_\full$ above, and the locality of a protocol in \cref{def:locality_of_protocol}.
Note that $G_\outgoing$ and $G_\incoming$ are interesting in their own right, and can be used for a fine-grained analysis of the communication graph of protocols in various settings, \eg when transmitting messages is costly but receiving messages is cheap (or vice versa).
We leave it open as an interesting problem to study various graph properties exhibited by these two graphs.

\subsection{Locality of a Protocol}

We now present a definition of communication locality, aligning with that of~\cite{BGT13}, with respect to the terminology introduced above.
\begin{definition}[locality of a protocol instance]\label{def:locality_of_protocol}
Let $\instance(\pi_n) = (\pi_n, \secParam, \vx, \vrho,\Adv, \aux, \IS\subseteq[n], \vr)$ be an execution instance as in \cref{def:comm_graph}.
For every honest party $\Party_i$ we define the \textsf{locality of party $\Party_i$} to be the number of parties from which $\Party_i$ received and processed messages, or sent messages to; that is,
\[
\ell_i(\instance(\pi_n)) = \size{\set{j \mid (i,j) \in G_\outgoing} \cup \set{j \mid (j,i) \in G_\incoming}}.
\]
The locality of $\instance(\pi_n)$ is defined as the maximum locality of an honest party, \ie
\[
\ell(\instance(\pi_n)) = \max_{i\in[n]\setminus\IS}\set{\ell_i(\instance(\pi_n))}.
\]
\end{definition}

We proceed by defining locality as a property of a protocol ensemble.
The protocol ensemble is parametrized by the number of parties $n$.
To align with standard notions of security where asymptotic measurements are with respect to the security parameter $\secParam$, we consider the situation where the growth of $n$ and $\secParam$ are polynomially related.

\begin{definition}[locality of a protocol]\label{def:locality_protocol}
Let $\pi=\sset{\pi_n}_{n\in\N}$ be a family of protocols in the correlated-randomness model with distribution $D_{\pi}=\sset{D_{\pi_n}}_{n\in\N}$, and let $\CS=\sset{\CS_n}_{n\in\N}$ be a family of adversary classes. We say that $\pi$ has locality $\ell(n)$ facing adversaries in $\CS$ if for every $n$ that is polynomially related to $\secParam$ it holds that for every input vector $\vx=(x_1,\ldots,x_n)$, every auxiliary information $\aux$, every adversary $\Adv\in\CS_n$ running with $\aux$, and every set of corrupted parties $\IS\subseteq[n]$, it holds that
\[
\pr{\ell(\pi_n,\Adv,\IS,\secParam,\vx,\aux) > \ell(n)} \leq \negl(\secParam),
\]
where $\ell(\pi_n,\Adv,\IS,\secParam,\vx,\aux)$ is the random variable corresponding to $\ell(\pi_n,\Adv,\IS,\secParam,\vx,\vrho,\aux,\vr)$ when $\vrho$ is distributed according to $D_{\pi_n}$ and $\vr$ is uniformly distributed.
\end{definition}

The following proposition follows from the sequential composition theorem of Canetti \cite{Canetti00}.
\begin{proposition}[composition of locality]\label{prop:Composition}
Let $f=\sset{f_n}_{n\in\N}$ and $g=\sset{g_n}_{n\in\N}$ be ensembles of $n$-party functionalities.
\begin{itemize}
    \item
    Let $\varphi=\sset{\varphi_n}_{n\in\N}$ be a protocol ensemble that securely computes $f$ with locality $\ell_\varphi$ tolerating adversaries in $\CS=\sset{\CS_n}_{n\in\N}$.
    \item
    Let $\pi=\sset{\pi_n}_{n\in\N}$ be a protocol that securely computes $g$ with locality $\ell_\pi$ in the $f$-hybrid model, tolerating adversaries in $\CS$, using $q=q(n)$ calls to the ideal functionality.
\end{itemize}
Then protocol $\pi^{f\mapsto\varphi}$, that is obtained from $\sset{\pi_n}$ by replacing all ideal calls to $f_n$ with the protocol $\varphi_n$, is a protocol ensemble that securely computes $g$ in the real model, tolerating adversaries in $\CS$, with locality at most $\ell_\pi + q\cdot\ell_\varphi$.
\end{proposition}

\subsection{Edge Expansion of a Protocol}

The measure of complexity we study for the communication graph of interactive protocols will be that of \emph{edge expansion} (see discussion below). We refer the reader to~\cite{HLW06,DW10} for more background on expanders.
We consider a definition of edge expansion which satisfies a natural monotonic property, where adding more edges cannot decrease the graph's  measure of expansion (see discussion in \cref{sec:expander_discussion}).

\begin{definition}{(edge expansion of a graph)}\label{def:graph_edge_expansion}
Given an undirected graph $G=(V,E)$, the \textsf{edge expansion ratio of $G$}, denoted $h(G)$, is defined as
\begin{equation}\label{eq:definition_h(G)}
h(G)=\min_{\sset{S\subseteq V: \ssize{S}\le \frac{\ssize{V}}{2}}} \frac{\ssize{\edges(S)}}{\ssize{S}} \ ,
\end{equation}
where $\edges(S)$ denotes the set of edges between $S$ and its complement $\comp{S}=V\setminus S$.
\end{definition}

\begin{definition}{(family of expander graphs)}\label{def:graph_expander}
A sequence $\sset{G_n}_{n\in\N}$ of graphs is a \textsf{family of expander graphs} if there exists a constant $\epsilon>0$ such that $ h(G_n) \geq \epsilon$ for all $n$.
\end{definition}

\noindent
We now consider the natural extension of graph expansion to the setting of protocol-induced communication graph.

\begin{definition}{(bounds on edge expansion of a protocol)}\label{def:protocol_edge_expansion}
Let $\pi=\sset{\pi_n}_{n\in\N}$, $D_{\pi}=\sset{D_{\pi_n}}_{n\in\N}$, and $\CS=\sset{\CS_n}_{n\in\N}$ be as in \cref{def:locality_protocol}.
\begin{itemize}
    \item
    A function $f(n)$ is a \textsf{lower bound of the edge expansion of $\pi$} facing adversaries in $\CS$, denoted $f(n) \leq h_{\pi,D_{\pi},\CS}(n)$, if for every $n$ that is polynomially related to $\secParam$, for every $\vx=(x_1,\ldots,x_n)$, every auxiliary information $\aux$, every $\Adv\in\CS_n$ running with~$\aux$, and every $\IS\subseteq[n]$, it holds that
    \[
    \pr{h(G_\full(\pi_n,\Adv,\IS,\secParam,\vx,\aux)) \leq f(n)} \leq \negl(\secParam),
    \]
    where $G_\full(\pi_n,\Adv,\IS,\secParam,\vx,\aux)$ is the random variable $G_\full(\pi_n,\Adv,\IS,\secParam,\vx,\vrho,\aux,\vr)$, when $\vrho$ is distributed according to $D_{\pi_n}$ and $\vr$ is uniformly distributed.

    \item
    A function $f(n)$ is an \textsf{upper bound of the edge expansion of $\pi$} facing adversaries in $\CS$, denoted $f(n) \geq h_{\pi,D_{\pi},\CS}(n)$,
    if there exists a polynomial relation between $n$ and $\secParam$ such that for infinitely many $n$ it holds that for every $\vx=(x_1,\ldots,x_n)$, every auxiliary information $\aux$, every $\Adv\in\CS_n$ running with~$\aux$, and every $\IS\subseteq[n]$, it holds that
    \[
    \pr{h(G_\full(\pi_n,\Adv,\IS,\secParam,\vx,\aux)) \geq f(n)} \leq \negl(\secParam).
    \]
\end{itemize}
\end{definition}

\begin{definition}[expander protocol]\label{def:protocol_expander}
Let $\pi=\sset{\pi_n}_{n\in\N}$, $D_{\pi}=\sset{D_{\pi_n}}_{n\in\N}$, and $\CS=\sset{\CS_n}_{n\in\N}$ be as in \cref{def:locality_protocol}. We say that \textsf{the communication graph of $\pi$ is an expander}, facing adversaries in~$\CS$, if there exists a constant function $\epsilon(n)>0$ such that $\epsilon(n)\leq h_{\pi,D_{\pi},\CS}(n)$.
\end{definition}

We note that most (if not all) secure protocols in the literature are expanders according to \cref{def:protocol_expander}, both in the realm of distributed computing~\cite{DS83,FM97,GM93,KSSV06,KKKSS08,KLST11,KS10} and in the realm of MPC~\cite{GMW87,BGW88,BGT13,CCGGOZ15,BCP15}.
Proving that a protocol is not an expander according to this definition requires showing an adversary for which the edge expansion is sub-constant. Looking ahead, both in our constructions of protocols that are not expanders (\cref{sec:ne_mpc}) and in our lower bound, showing that non-expander protocols can be attacked (\cref{sec:LB_Expander}), we use a stronger definition, that requires that the edge expansion is sub-constant facing \emph{all} adversaries, see \cref{def:protocol_non_expander} below. While it makes our positive results stronger, we leave it as an interesting open question to attack protocols that do not satisfy \cref{def:protocol_expander}.

\begin{definition}[strongly non-expander protocol]\label{def:protocol_non_expander}
Let $\pi=\sset{\pi_n}_{n\in\N}$, $D_{\pi}=\sset{D_{\pi_n}}_{n\in\N}$, and $\CS=\sset{\CS_n}_{n\in\N}$ be as in \cref{def:locality_protocol}. We say that \textsf{the communication graph of $\pi$ is strongly not an expander}, facing adversaries in $\CS$, if there exists a sub-constant function $\alpha(n)\in o(1)$ such that
$\alpha(n) \geq h_{\pi,D_{\pi},\CS}(n)$.
\end{definition}

We next prove a useful observation that will come into play in \cref{sec:LB_Expander}, stating that if the communication graph of $\pi$ is strongly not an expander, then there must exist a sublinear cut in the graph.
\begin{lemma}\label{lem:sublinear_cut}
Let $\pi=\sset{\pi_n}_{n\in\N}$ be a family of protocols in the correlated-randomness model with distribution $D_{\pi}=\sset{D_{\pi_n}}_{n\in\N}$, and let $\CS=\sset{\CS_n}_{n\in\N}$ be such that $\CS_n$ is the class of adversaries corrupting at most $\beta\cdot n$ parties, for a constant $0<\beta<1$.

Assuming the communication graph of $\pi$ is strongly not an expander facing adversaries in $\CS$, there exists a sublinear function $\alpha(n)\in o(n)$ such that for infinitely many $n$'s the full communication graph of $\pi_n$ has an $\alpha(n)$-cut with overwhelming probability.
\end{lemma}

\begin{proof}
Since the full communication graph of $\pi$ is strongly not an expander, there exists a sub-constant function $\alpha'(n)\in o(1)$ such that there exists a polynomial relation between $n$ and $\secParam$ such that for infinitely many $n$'s it holds that for every input $\vx=(x_1,\ldots,x_n)$ and every adversary $\Adv\in\CS_n$ and every set of corrupted parties $\IS$,
\[
\pr{h(G_\full(\pi_n,\Adv,\IS,\secParam,\vx,\aux))> \alpha'(n)} \leq \negl(\secParam).
\]
This means that for these $n$'s, with overwhelming probability there exists a subset $S_n\subseteq[n]$ of size at most $n/2$ for which
\[
\frac{\ssize{\edges(S_n)}}{\ssize{S_n}} \leq \alpha'(n).
\]
Since $\ssize{S_n}\leq n/2$ it holds that
\[
\size{\edges(S_n)}\leq \alpha'(n)  \cdot \frac{n}{2}.
\]
We define $\alpha(n) = \alpha'(n)  \cdot \frac{n}{2}$, and the claim thus holds for $\alpha(n)\in o(n)$.
\QED
\end{proof}

\subsection{Discussion on the Definition of Expander Protocols}\label{sec:expander_discussion}

In addition to edge expansion, there are two other commonly studied notions of expansion in graphs: \emph{spectral expansion} and \emph{vertex expansion} (see \cite{HLW06} for a survey on expander graphs).
Spectral expansion is a linear-algebraic definition of expansion in regular graphs, equal to the difference of first and second largest eigenvalues of the graph's adjacency matrix.
Vertex (like edge) expansion is combinatorial definitions of expansion, but considers the number of distinct nodes neighboring subsets of the graph as opposed to the number of outgoing edges.

The three metrics are loosely coupled, and analyzing any of them for the communication graphs in our setting seems a natural choice.
We choose to study edge expansion, as it is particularly amenable to our new techniques.
We further consider an \emph{unscaled} version of edge expansion, as opposed to the more stringent scaled version wherein the expansion ratio $h(G)$ of \cref{eq:definition_h(G)} is defined with an additional scale factor of $(\deg G)^{-1}$.
The unscaled variant (while weaker) satisfies the more natural property that it is \emph{monotonic} with respect to the number of edges. That is, by adding more edges to the graph, the value of $h(G)$ cannot decrease; whereas, when dividing $h(G)$ by the degree, adding more edges to the graph might actually decrease the value of $h(G)$, and the graph which was an expander earlier may end up being non-expander. This monotonicity appears more natural in the setting of communication graphs, where adding more edges cannot harm the ability to successfully execute a protocol.

We remark that while our definitional focus is on this form of unscaled edge expansion, our upper bounds (\ie protocols with non-expanding communication graphs) apply to all aforementioned notions. Considering and extending our lower bound to alternative notions of expansion (spectral/vertex/scaled) is left as an interesting open problem.

\subsection{The Adaptive Setting}\label{sec:def_adaptive}

In the adaptive setting, the definitions of locality of protocols and of protocols with communication graph that forms an expander follow the same spirit as in the static case, however, require a few technical modifications.

Recall that we follow the adaptive model from Canetti \cite{Canetti00},\footnote{We follow the modular composition framework~\cite{Canetti00} for the sake of clarity and simplicity. We note that the definitions can be adjusted to the UC framework~\cite{Canetti01}.} where an environment machine interacts with the adversary/simulator. In particular, the adversary does not receive auxiliary information at the onset of the protocol; rather the environment acts as an ``interactive auxiliary-information provider'' and hands the adversary auxiliary information about parties that get corrupted dynamically. In addition, the set of corrupted parties is not defined at the beginning, but generated dynamically during the protocol based on corruption request issued by the adversary, and also after the completion of the protocol, during the post-execution corruption (PEC) phase, based on corruption requests issued by the environment.

Therefore, the required changes to the definitions are two-fold:
\begin{enumerate}
    \item
    The parameters for defining an instance of a protocol execution are: the $n$-party protocol $\pi_n$ the security parameter $\secParam$, the input vector for the parties $\vx=(x_1,\ldots,x_n)$, the correlated randomness for the parties $\vrho=(\rho_1,\ldots,\rho_n)$, the environment $\Env$, the adversary $\Adv$, the auxiliary input for the environment $\aux$, and the random coins for the parties, the adversary, and the environment $\vr=(r_1,\ldots,r_n,r_\Adv,r_\Env)$.

    We denote by $\instance_\adaptive(\pi_n)=(\pi_n,\Env,\Adv,\secParam,\vx,\vrho,\aux,\vr)$.
    \item
    The second difference considers the timing where the communication graph is set.
    \begin{itemize}
        \item
        After the parties generate their output, and before the PEC phase begins.
        \item
        At the end of the PEC phase, when the environment outputs its decision bit.
    \end{itemize}
    Since the communication graph is fixed at the end of the protocol (before the PEC phase begins), the difference lies in the identity of the corrupted parties. More precisely, an edge that appears in the graph before the PEC phase might not appear after the PEC phase, in case both parties became corrupt. For this reason, we consider the communication graph after the parties generate their outputs and before the PEC phase begins.
\end{enumerate}

All definitions as presented for static case translate to the adaptive setting with the two adjustments presented above.

\section{MPC with Non-Expanding Communication Graph}\label{sec:ne_mpc}

In this section, we show that in various standard settings, the communication graph of an MPC protocol is \emph{not} required to be an expander graph, even when the communication locality is poly-logarithmic. In \cref{sec:ne_mpc_static_computational}, we focus on static corruptions and computational security.
In \cref{sec:ne_mpc_static_it}, we extend the construction to the information-theoretic setting, and in \cref{sec:ne_mpc_adaptive} to the adaptive-corruption setting.

\subsection{Computational Security with Static Corruptions}\label{sec:ne_mpc_static_computational}

We start by considering the computational setting with static corruptions.
\begin{theorem}[restating \cref{thm:intro:UB} and Item~1 of \cref{thm:intro:UB_extensions}]\label{thm:mpc_no_expander}
Let $f=\sset{f_n}_{n\in\N}$ be an ensemble of functionalities, let $\delta>0$, and assume that one-way functions exist.
Then, the following holds in the PKI-hybrid model with secure channels:
\begin{enumerate}
    \item
    Let $\beta< 1/4-\delta$ and let $t(n)=\beta \cdot n$.
    Then, $f$ can be securely computed by a protocol ensemble~$\pi$ tolerating static PPT $t(n)$-adversaries such that the communication graph of $\pi$ is strongly not an expander.
    \item
    Let $\beta< 1/6-\delta$ and let $t(n)=\beta \cdot n$.
    Then, $f$ can be securely computed by a protocol ensemble $\pi$ tolerating static PPT $t(n)$-adversaries such that (1) the communication graph of $\pi$ is strongly not an expander, and (2) the locality of $\pi$ is poly-logarithmic in $n$.
    \item
    Let $\beta< 1/4-\delta$, let $t(n)=\beta \cdot n$, and assume in addition the \emph{secret-key infrastructure (SKI)} model\footnote{In the SKI model every pair of parties has a secret random string that is unknown to other parties.} and the existence of public-key encryption schemes.
    Then, $f$ can be securely computed by a protocol ensemble $\pi$ tolerating static PPT $t(n)$-adversaries such that (1) the communication graph of $\pi$ is strongly not an expander, and (2) the locality of $\pi$ is poly-logarithmic in $n$.\footnote{This item hold in the authenticated-channels model, since we assume PKE.}
\end{enumerate}
\end{theorem}

\begin{proof}
The theorem follows from \cref{lem:mpc_no_expander} (below) by instantiating the hybrid functionalities using existing MPC protocols from the literature.
\begin{itemize}
    \item
    The first part follows using honest-majority MPC protocols that exist assuming one-way functions in the secure-channels model, \eg the protocol of Beaver et al.\ \cite{BMR90} or of \Damgard and Ishai \cite{DI05}.\footnote{Generic honest-majority MPC protocols require a broadcast channel or some form of trusted setup assumptions~\cite{CL17,CHOR18}. The PKI assumption in \cref{thm:mpc_no_expander} is sufficient in the computational setting. Looking ahead, in the information-theoretic setting (\cref{sec:ne_mpc_static_it}) additional adjustments are required.}
    \item
    The second part follows using the low-locality MPC protocol of Boyle et al.\ \cite{BGT13} that exists assuming one-way functions in the PKI model with secure channels and tolerates $t= (1/3-\delta)n$ static corruptions.\footnote{In~\cite{BGT13} public-key encryption is also assumed, but as we show in \cref{sec:it_bgt}, this assumption can be removed.}
    \item
    The third part follows using the low-locality MPC protocol of Chandran et al.\ \cite{CCGGOZ15} that exists assuming public-key encryption in the PKI and SKI model with authenticated channels and tolerates $t<n/2$ static corruptions.
\qedhere
\end{itemize}
\end{proof}

\subsubsection{Ideal Functionalities used in the Construction\SUBSUBSEC}

The proof of \cref{thm:mpc_no_expander} relies on \cref{lem:mpc_no_expander} (below). We start by defining the notations and the ideal functionalities that will be used in the protocol considered in \cref{lem:mpc_no_expander}.

\paragraph{Signature notations.}
Given a signature scheme $(\Gen,\Sign,\Verify)$ and $m$ pairs of signing and verification keys $(\sk[i],\vk[i])\gets\Gen(1^\secParam)$ for $i\in[m]$, we use the following notations for signing and verifying with multiple keys:
\begin{itemize}
    \item
    Given a message $\mu$ we denote by $\Sign_{\sk[1],\ldots,\sk[m]}(\mu)$ the vector of $m$ signatures $\sigma=(\sigma_1,\ldots,\sigma_m)$, where $\sigma_i\gets\Sign_{\sk[i]}(\mu)$.
    \item
    Given a message $\mu$ and a signature $\sigma=(\sigma_1,\ldots,\sigma_m)$, we denote by $\Verify_{\vk[1],\ldots,\vk[m]}(\mu,\sigma)$ the verification algorithm that for every $i\in[m]$ computes $b_i\gets\Verify_{\vk[i]}(\mu,\sigma_i)$, and accepts the signature $\sigma$ if and only if $\sum_{i=1}^m b_i \geq m-t$, \ie even if up to $t$ signatures are invalid.
\end{itemize}
We note that it is possible to use multi-signatures or aggregated signatures~\cite{MOR01,BGLS03,LMRS04,LOSSW13} in order to obtain better communication complexity, however, we use the notation above both for simplicity and as a step toward the information-theoretic construction in the following section.

\paragraph{The Elect-and-Share functionality.}

In the Elect-and-Share $m$-party functionality, $\felectsharefull$, every party $\Party_i$ has a pair of inputs $(x_i,\sk[i])$, where $x_i\in\zs$ is the ``actual input'' and $\sk[i]$ is a private signing key. The functionality starts by electing two random subsets $\committee_1,\committee_2\subseteq[m]$ of size $n'$, and signing each subset using all signing keys. In addition, every input value $x_i$ is secret shared using a $(t',n')$ error-correcting secret-sharing scheme (see \cref{def:ECSS}).
Every party receives as output the subset $\committee_1$, whereas a party $\Party_i$, for $i\in\committee_1$, receives an additional output consisting of a signature on $\committee_1$, the signed subset $\committee_2$, along with one share for each one of the $m$ input values.
The formal description of the functionality can be found in~\cref{fig:felectshare}.

\begin{nfbox}{The Elect-and-Share functionality}{fig:felectshare}
\begin{center}
    \textbf{The functionality} $\felectsharefull$
\end{center}
\vspace{-0.3cm}
The $m$-party functionality $\felectsharefull$ is parametrized by a signature scheme $(\Gen,\Sign,\Verify)$ and a $(t',n')$ ECSS scheme $(\Share,\Recon)$, and proceeds with parties $\PS_1=\sset{\Party_1,\ldots,\Party_m}$ as follows.
\begin{enumerate}
    \item
    Every party $\Party_i$ sends a pair of values $(x_i,\sk[i])$ as its input, where $x_i$ is the actual input value and $\sk[i]$ is a signing key. (If $\Party_i$ didn't send a valid input, set it to the default value, \eg zero.)
    \item
    Sample uniformly at random two subsets (committees) $\committee_1,\committee_2\subseteq[m]$ of size $n'$.
    \item
    Sign each subset as $\sigma_1=\Sign_{\sk[1],\ldots,\sk[m]}(\committee_1)$ and $\sigma_2=\Sign_{\sk[1],\ldots,\sk[m]}(\committee_2)$.
    \item
    For every $i\in[m]$, secret share $x_i$ as $(s_i^1,\ldots,s_i^{n'})\gets \Share(x_i)$.
    \item
    For every $j\in[n']$, set $\vs_j=(s_1^j,\ldots,s_m^j)$.
    \item
    Denote $\committee_1=\sset{i_{(1,1)},\ldots,i_{(1,n')}}$.
    For every $i\notin\committee_1$ set the output of $\Party_i$ to be $\committee_1$.
    For every $i=i_{(1,j)}\in\committee_1$ (for some $j\in[n']$), set the output of $\Party_i$ to be $(\committee_1,\sigma_1,\committee_2,\sigma_2,\vs_j)$.
\end{enumerate}
\end{nfbox}

\paragraph{The Reconstruct-and-Compute functionality.}

The Reconstruct-and-Compute functionality, $\freconcomputefull$, is an $m$-party functionality. Denote the party-set by $\sset{\Party_{m+1}, \ldots, \Party_{2m}}$.\footnote{We use the notation $\sset{\Party_{m+1}, \ldots, \Party_{2m}}$ instead of the more standard $\sset{\Party_1, \ldots, \Party_m}$ for consistency with the protocol $\protne_n$ for $n=2m$ (described in \cref{fig:protne}), where parties $\sset{\Party_{m+1}, \ldots, \Party_{2m}}$ invoke the functionality.} Every party $\Party_{m+i}$ has an input value $x_{m+i}\in\zs$, and a potential additional input value consisting of a signed subset $\committee_2\subseteq[m]$ and a vector of $m$ shares. The functionality starts by verifying the signatures, where every invalid input is ignored. The signed inputs should define a single subset $\committee_2\subseteq[m]$ (otherwise the functionality aborts), and the functionality uses the additional inputs of parties $\Party_{m+i}$, for every $i\in\committee_2$, in order to reconstruct the $m$-tuple $(x_1,\ldots,x_m)$. Finally, the functionality computes $y=f(x_1,\ldots,x_{2m})$ and hands $y$ as the output for every party.
The formal description of the functionality can be found in~\cref{fig:freconcompute}.

\paragraph{The Output-Distribution functionality.}
The $m$-party Output-Distribution functionality is parametrized by a subset $\committee_1\subseteq[m]$. Every party $\Party_i$, with $i\in\committee_1$, hands in a value, and the functionality distributes the majority of these inputs to all the parties.
The formal description of the functionality can be found in~\cref{fig:foutdist}.

\begin{nfbox}{The Output-Distribution functionality}{fig:foutdist}
\begin{center}
    \textbf{The functionality} $\foutdistfull$
\end{center}
\vspace{-0.3cm}
The $m$-party functionality $\foutdistfull$ is parametrized by a subset $\committee_1\subseteq[m]$, and proceeds with parties $\PS_1=\sset{\Party_1,\ldots,\Party_m}$ as follows.
\begin{enumerate}
    \item
    Every party $\Party_i$, with $i\in\committee_1$, gives input value $y_i$, a party $\Party_i$ with $i\notin\committee_1$ gives the empty input~$\emptystr$.
    \item
    Denote $y=\maj\sset{y_i \mid i\in\committee_1}$ (choose arbitrarily if the majority is not unique).
    \item
    Output $y$ to every $\Party_i\in\PS_1$.
\end{enumerate}
\end{nfbox}

\begin{nfbox}{The Reconstruct-and-Compute functionality}{fig:freconcompute}
\begin{center}
    \textbf{The functionality} $\freconcompute$
\end{center}
\vspace{-0.3cm}
The $m$-party functionality $\freconcomputefull$ is parametrized by a signature scheme $(\Gen,\Sign,\Verify)$, a $(t',n')$ ECSS scheme $(\Share,\Recon)$, and a vector of verification keys $(\vk[1],\ldots,\vk[m])$, and proceeds with parties $\PS_2=\sset{\Party_{m+1},\ldots,\Party_{2m}}$ as follows.
\begin{enumerate}
    \item
    Every party $\Party_{m+i}$ sends a pair of values $(x_{m+i},z_{m+i})$ as its input, where $x_{m+i}$ is the actual input value and either $z_{m+i}=\emptystr$ or $z_{m+i}=(\committee_{m+i},\sigma_{m+i},\vs_{m+i})$.
    \item
    For every $\Party_{m+i}$ that provided $z_{m+i}\neq \emptystr$, verify that $\Verify_{\vk[1],\ldots,\vk[m]}(\committee_{m+i},\sigma_{m+i})=1$ (ignore invalid inputs).
    If there is no subset $\committee_2\subseteq[m]$ of size $n'$ with an accepting signature, or if there exists more than one such subset, then abort. Otherwise, denote $\committee_2=\sset{i_{(2,1)}\ldots,i_{(2,n')}}$.
    \item
    For every $i=i_{(2,j)}\in\committee_2$, let $\vs_{m+i_{(2,j)}}=(s_1^j,\ldots,s_m^j)$ be the input provided by $\Party_{m+i}$. (If $\Party_{m+i}$ provided invalid input, set $\vs_{m+i_{(2,j)}}$ to be the default value, \eg the zero vector.)
    \item
    For every $i\in[m]$, reconstruct $x_i=\Recon(s_i^1,\ldots,s_i^{n'})$.
    \item
    Compute $y=f(x_1\ldots, x_m,x_{m+1},\ldots,x_{2m})$.
    \item
    Output $y$ to every $\Party_{m+i}\in\PS_2$.
\end{enumerate}
\end{nfbox}

\subsubsection{Constructing Non-Expander Protocols\SUBSUBSEC}\label{sec:upperbound_protocol}

\paragraph{High-level overview of the protocol.} Having defined the ideal functionalities, we are ready to present the main lemma. We start by describing the underlying idea behind the non-expanding MPC protocol $\protne_n$ (\cref{fig:protne}).
At the onset of the protocol, the party-set is partitioned into two subsets of size $m=n/2$, a left subset and a right subset (see~\Cref{fig:results}). The left subset will invoke the Elect-and-Share functionality, that elects two subsets $\committee_1,\committee_2\subseteq[m]$ of size $n'=\log^2(n)$.\footnote{We note that any $n'\in\omega(\log{n})$ will do.} The parties in the left subset corresponding to $\committee_1$ and the parties in the right subset corresponding to $\committee_2$ will form a ``bridge.'' The parties in $\committee_1$ will receive shares of all inputs values of parties in the left subset, and transfer them to $\committee_2$. Next, the right subset of parties will invoke the Reconstruct-and-Compute functionality, where each party hands its input value, and parties in $\committee_2$ additionally provide the shares they received from $\committee_1$. The functionality reconstructs the left-subset's inputs, computes the function $f$ and hands the output to the right subset. Finally, $\committee_2$ will transfer the output value to $\committee_1$, and the left subset will invoke the Output-Distribution functionality in order to distribute the output value to all the parties.

\begin{figure}[htb]
	\begin{center}
    \includegraphics[scale=0.6]{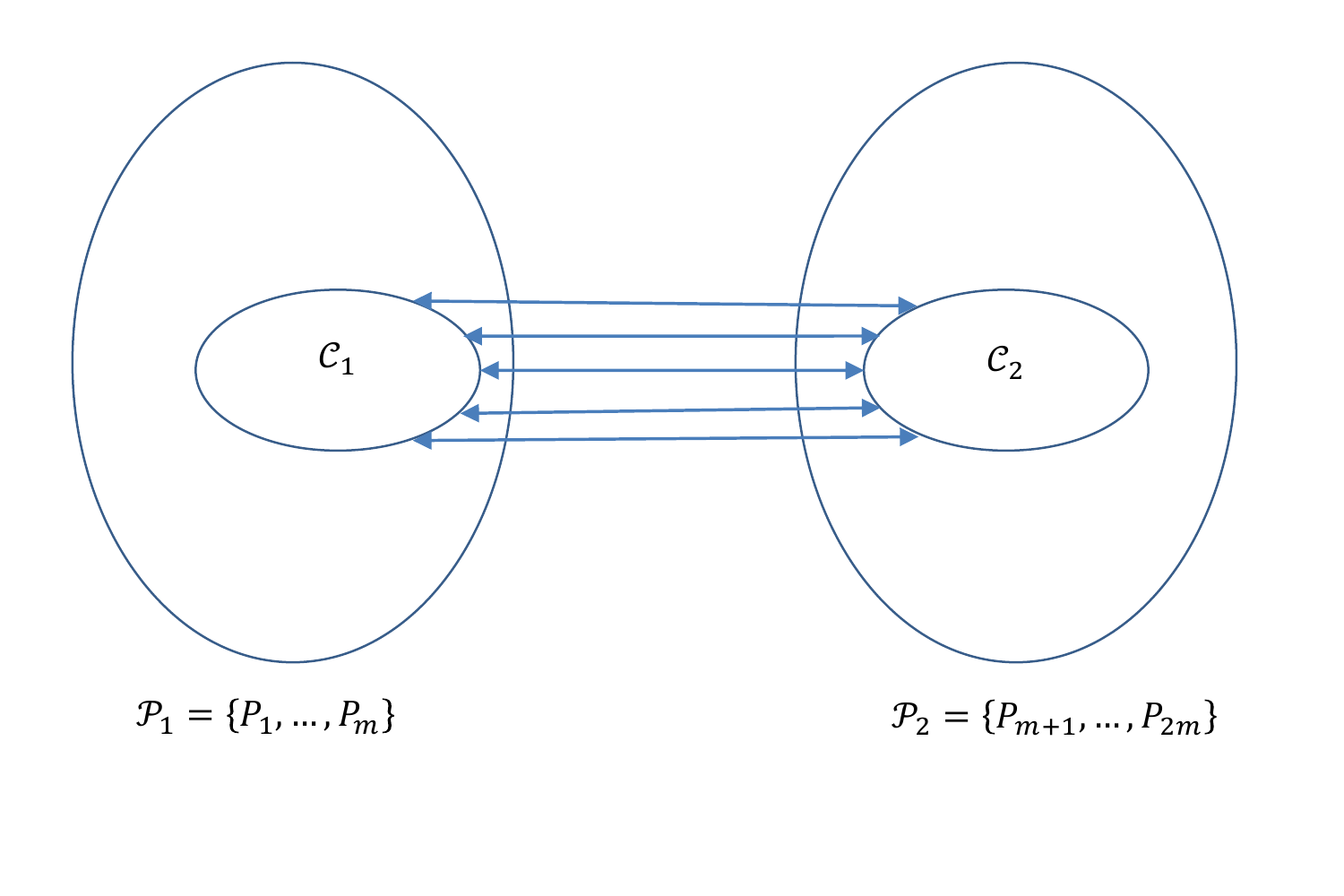}
	\end{center}
	\vspace{-10ex}
	\caption{The non-expanding subsets in the protocol $\protne$. The sets $\committee_1$ and $\committee_2$ are of poly-logarithmic size and the sets $\PS_1$ and $\PS_2$ are of linear size. The number of edges between $\PS_1$ and $\PS_2$ is poly-logarithmic.}
	\label{fig:results}
\end{figure}

\begin{lemma}\label{lem:mpc_no_expander}
Let $f=\sset{f_n}_{n\in2\N}$,\footnote{For simplicity, we consider even $n$'s. Extending the statement to any $n$ is straightforward, however, adds more details.} where $f_n$ is an $n$-party functionality for $n=2m$, let $\delta>0$, and assume that one-way functions exist.
Then, in the PKI-hybrid model with secure channels, where a trusted party additionally computes the $m$-party functionality-ensembles $(\felectshare,\freconcompute,\foutdist)$ tolerating $\gamma \cdot m$ corruptions, there exists a protocol ensemble $\pi$ that securely computes $f$ tolerating static PPT $\beta n$-adversaries, for $\beta<\min(1/4-\delta,\gamma/2)$, with the following guarantees:
\begin{enumerate}
    \item
    The communication graph of $\pi$ is strongly not an expander.
    \item
    Denote by $f_1,f_2,f_3$ the functionality-ensembles $\felectshare,\freconcompute,\foutdist$ (resp.).
    If protocol-ensembles $\rho_1,\rho_2,\rho_3$ securely compute $f_1,f_2,f_3$ (resp.) with locality $\ell_\rho=\ell_\rho(m)$, then $\pi^{f_i\rightarrow \rho_i}$ (where every call to $f_i$ is replaced by an execution of $\rho_i$) has locality $\ell=2\cdot\ell_\rho + \log^2(n)$.
\end{enumerate}
\end{lemma}

\begin{proof}
For $m\in\N$ and $n=2m$, we construct the $n$-party protocol $\protne_n$ (see \cref{fig:protne}) in the $(\felectshare,\freconcompute,\foutdist)$-hybrid model. The parameters for the protocol are $n'=\log^2(n)$ and $t'=(1/2-\delta)\cdot n'$. We start by proving in \cref{prop:security} that the protocol $\protne_n$ securely computes $f_n$. Next, in \cref{prop:not_expanding} we prove that the communication graph of $\protne$ is strongly not an expander. Finally, in \cref{prop:low_locality} we prove that by instantiating the functionalities $(\felectshare,\freconcompute,\foutdist)$ using low-locality protocols, the resulting protocol has low locality.

\begin{proposition}\label{prop:security}
For sufficiently large $n$, the protocol $\protne_n$ securely computes the function $f_n$, tolerating static PPT $\beta n$-adversaries, in the $(\felectshare,\freconcompute,\foutdist)$-hybrid model.
\end{proposition}

\ifdefined\IsProofInAppendix
The proof of \cref{prop:security} can be found in \cref{sec:prop:security_proof}.
\else
\begin{proof}
Let $\Adv$ be an adversary attacking the execution of $\protne_n$, and let $\IS\subseteq[n]$ be a subset of size at most $t=\beta n$. We construct the following adversary $\Sim$ for the ideal model computing $f$. On inputs $\set{x_i}_{i\in\IS}$ and auxiliary input $\aux$, the simulator $\Sim$ starts by emulating $\Adv$ on these inputs. Initially, $\Sim$ generates a pair of signature keys $(\sk[i],\vk[i])\gets\Gen(1^\secParam)$ for every $i\in[n]$, and hands $\sset{\sk[i]}_{i\in\IS}$ and $(\vk[1],\ldots,\vk[n])$ to $\Adv$. Next, $\Sim$ plays towards $\Adv$ the roles of the honest parties and the ideal functionalities $\felectshare$, $\freconcompute$, and $\foutdist$. For simplicity, assume that all input values are elements in $\zo^\secParam$.

To simulate Step~\ref{step:elect_share}, the simulator emulates $\felectsharefull$ towards $\Adv$ as follows. $\Sim$ receives from $\Adv$ input values $\sset{(x'_i,\sk[i]')}_{i\in\IS\cap[m]}$ (replace invalid inputs with default). Next, $\Sim$ samples uniformly at random two subsets $\committee_1=\sset{i_{(1,1)},\ldots,i_{(1,n')}}\subseteq[m]$ and $\committee_2=\sset{i_{(2,1)},\ldots,i_{(2,n')}}\subseteq[m]$, sets $\sk[i]'=\sk[i]$ for every $i\in[m]\setminus\IS$, and signs the subsets as $\sigma_1=\Sign_{\sk[1]',\ldots,\sk[m]'}(\committee_1)$ and $\sigma_2=\Sign_{\sk[1]',\ldots,\sk[m]'}(\committee_2)$. Finally, $\Sim$ generates secret shares of zero for every $i\in[m]$ as $(s_i^1,\ldots, s_i^{n'})\gets\Share(0^\secParam)$, and hands $\Adv$ the output $\committee_1$ for every $\Party_i$ with $i\in([m]\cap\IS)\setminus \committee_1$, and $(\committee_1,\sigma_1,\committee_2,\sigma_2,\vs_j)$, with $\vs_j=(s_1^j,\ldots,s_m^j)$, for every $i=i_{(1,j)}\in\committee_1\cap\IS$.

To simulate Step~\ref{step:first_bridge}, the simulator sends $(\committee_1,\sigma_1,\committee_2,\sigma_2)$ on behalf of every honest party $\Party_i$ with $i\in\committee_1\setminus \IS$ to every corrupted party $\Party_{m+i_{(2,j)}}$ with $i_{(2,j)}\in\committee_2$ and $m+i_{(2,j)}\in\IS$, and receives such values from the adversary ($\Sim$ ignores any invalid messages from $\Adv$). In addition, for every $j\in[n']$ such that $i_{(1,j)}\notin\IS$ and $m+i_{(2,j)}\in\IS$, the simulator $\Sim$ sends to $\Adv$ on behalf of the honest party $\Party_{i_{(1,j)}}$ the vector $\vs_j$, intended to $\Party_{m+i_{(2,j)}}$.

To simulate Step~\ref{step:recon_compute}, the simulator emulates $\freconcompute^{\vk[1],\ldots,\vk[m]}$ towards $\Adv$ as follows. $\Sim$ receives from $\Adv$ input values $\sset{(x'_{m+i},z_{m+i})}_{m+i\in\IS\cap[m+1,2m]}$, where either $z_{m+i}=\emptystr$ or $z_{m+i}=(\committee_{m+i},\sigma_{m+i},\vs_{m+i})$ (replace invalid inputs with default). Next, $\Sim$ sends the values $\sset{x'_i}_{i\in\IS}$ to the trusted party (for $i\in\IS\cap[m]$ the values $x'_i$ were obtained in the simulation of Step~\ref{step:elect_share}, and for $i\in\IS\cap[m+1,2m]$ in the simulation of Step~\ref{step:recon_compute}) and receives back the output value $y$. Finally, $\Sim$ hands $y$ to $\Adv$ for every corrupted party $\Party_{m+i}$, with $m+i\in\IS\cap[m+1,2m]$.

To simulate Step~\ref{step:second_bridge}, the simulator sends $y$ on behalf of every honest party $\Party_{m+i_{(2,j)}}$ with $i_{(2,j)}\in\committee_2$ and $m+i_{(2,j)}\notin\IS$ to every corrupted party $\Party_{i_{(1,j)}}$ with $i_{(1,j)}\in\committee_1\cap\IS$, and receives such values from the adversary.

To simulate Step~\ref{step:out_dist}, the simulator emulates the functionality $\foutdistfull$ by receiving a value $y_i$ from $\Adv$ for every $i\in\IS\cap\committee_1$, and sending $y$ to every corrupted party $\Party_i$ with $i\in\IS\cap[m]$. Finally, $\Sim$ outputs whatever $\Adv$ outputs and halts.

We prove computational indistinguishability between the execution of the protocol $\protne_n$ running with adversary $\Adv$ and the ideal computation of $f$ running with $\Sim$ via a series of hybrids experiments. The output of each experiment is the output of the honest parties and of the adversary.

\paragraph{The game $\HYB^1_{\pi, \IS, \Adv(\aux)}(\vx,\secParam)$.}
This game is defined to be the execution of the protocol $\protne_n$ in the $(\felectshare,\freconcompute,\foutdist)$-hybrid model on inputs $\vx\in(\zs)^n$ and security parameter $\secParam$ with adversary $\Adv$ running on auxiliary information $\aux$ and controlling parties in $\IS$.

\paragraph{The game $\HYB^2_{\pi, \IS, \Adv(\aux)}(\vx,\secParam)$.}
In this game, we modify $\HYB^1_{\pi, \IS, \Adv(\aux)}(\vx,\secParam)$ as follows.
Instead of verifying the signatures of the input values $\committee_2$ and $\sigma_2$, the functionality $\freconcompute$ considers the subset $\committee_2$ that was sampled by $\felectshare$.

\begin{claim}\label{claim:hyb1}
$\sset{\HYB^1_{\pi, \IS, \Adv(\aux)}(\vx,\secParam)}_{\vx,\aux,\secParam} \compindist \sset{\HYB^2_{\pi, \IS, \Adv(\aux)}(\vx,\secParam)}_{\vx,\aux,\secParam}$.
\end{claim}
\begin{proof}
The claim follows from two observations. First, the signed subset $\committee_2$ that was computed by $\felectshare$ will be given as input to $\freconcompute$ by at least one party, as long as there exist honest parties in $\committee_1$ and $\committee_2$. This is guaranteed to occur except for a negligible probability following \cref{cor:elect}.

Second, by the security of the signature scheme, the functionality $\freconcompute$ will not receive a second signed subset. Indeed, if two subsets $\committee_2$ and $\committee'_2$ with accepting signatures $\sigma_2$ and $\sigma_2'$ (resp.) are given to $\freconcompute$ with a noticeable probability, one can construct a forger to the signature scheme.
\QED
\end{proof}

\paragraph{The game $\HYB^3_{\pi, \IS, \Adv(\aux)}(\vx,\secParam)$.}
In this game, we modify $\HYB^2_{\pi, \IS, \Adv(\aux)}(\vx,\secParam)$ as follows.
instead of computing the function $f$ based on the inputs of $\PS_2$ provided to $\freconcompute$ and the reconstructed input values of $\PS_1$ based on the shares provided by the parties in $\committee_2$, the computation of $f$ is performed as follows. The input values $\sset{x'_i}_{i\in\IS\cap[m]}$ provided by the adversary to $\felectshare$ and the input values $\sset{x'_i}_{i\in\IS\cap[m+1,2m]}$ provided by the adversary to $\freconcompute$ are sent to an external party that computes $y=f(x'_1,\ldots,x'_n)$, where $x'_i=x_i$ for every $i\in[n]\setminus\IS$.

\begin{claim}\label{claim:hyb2}
$\sset{\HYB^2_{\pi, \IS, \Adv(\aux)}(\vx,\secParam)}_{\vx,\aux,\secParam} \statclose \sset{\HYB^3_{\pi, \IS, \Adv(\aux)}(\vx,\secParam)}_{\vx,\aux,\secParam}$.
\end{claim}
\begin{proof}
The claim will follow as long as $\ssize{(\committee_1\cup\sset{m+i\mid i\in\committee_2})\cap\IS}\leq t'$.
In order to prove this, we will show separately that $\ssize{\committee_1\cap\IS}< t'/2$ and that $\ssize{\sset{m+i\mid i\in\committee_2}\cap\IS}< t'/2$, except for a negligible probability.

Since $\beta<1/4-\delta$, for some fixed $\delta>0$, and there are at most $\beta\cdot n$ corruptions, and $n=2m$ parties, it holds that $\ssize{\IS}<\beta\cdot2m<(1/2-2\delta)\cdot m$, therefore, $\ssize{[m]\cap\IS}<(1/2-2\delta)\cdot m$, and similarly, $\ssize{[m+1,2m]\cap\IS}<(1/2-2\delta)\cdot m$. Now, since $n'=\omega(\log(n))=\omega(\log(m))$, it follows from \cref{cor:elect} that $\ssize{\committee_1\cap\IS}<(1/2-\delta) \cdot n'$ and $\ssize{\sset{m+i\mid i\in\committee_2}\cap\IS}<(1/2-\delta) \cdot n'$, except for a negligible probability.
\QED
\end{proof}

\paragraph{The game $\HYB^4_{\pi, \IS, \Adv(\aux)}(\vx,\secParam)$.}
In this game, we modify $\HYB^3_{\pi, \IS, \Adv(\aux)}(\vx,\secParam)$ as follows.
instead of computing shares of the input values $\sset{x_i}_{i\in[m]}$, the functionality $\felectshare$ computes shares of $0^\secParam$.

\begin{claim}\label{claim:hyb3}
$\sset{\HYB^4_{\pi, \IS, \Adv(\aux)}(\vx,\secParam)}_{\vx,\aux,\secParam} \statclose \sset{\HYB^3_{\pi, \IS, \Adv(\aux)}(\vx,\secParam)}_{\vx,\aux,\secParam}$.
\end{claim}
\begin{proof}
The claim follows from the privacy of the ECSS scheme and since $\ssize{\committee_1\cap\IS}\leq t'$ and $\ssize{\sset{m+i\mid i\in\committee_2}\cap\IS}\leq t'$, except for a negligible probability.
\QED
\end{proof}

The proof now follows since $\HYB^4_{\pi, \IS, \Adv(\aux)}(\vx,\secParam)$ exactly describes the simulation done by $\Sim$, and in particular, does not depend on the input values of honest parties. This concludes the proof of \cref{prop:security}.
\QED
\end{proof}

\fi

\begin{proposition}\label{prop:not_expanding}
The communication graph of the Protocol~$\protne$ is strongly not an expander, facing static PPT $\beta n$-adversaries.
\end{proposition}
\begin{proof}
For $n=2m$, consider the set $\PS_1=\sset{\Party_1,\ldots,\Party_m}$ and its complement $\PS_2=\PS\setminus \PS_1$.
For any input vector and for every static PPT $\beta n$-adversary it holds that with overwhelming probability that
$\ssize{\PS_1}=n/2$ and $\edges(\PS_1,\PS_2)=\ssize{\committee_1}\times\ssize{\committee_2}=\log^2(n)\cdot\log^2(n)$. Therefore, considering the function
\[
f(n)=\frac{2\log^4(n)}{n},
\]
it holds that $f(n)\in o(1)$ and $f(n)$ is an upper bound of the edge expansion of $\protne$ (see \cref{def:protocol_edge_expansion}).
We conclude that the communication graph of $\protne$ is strongly not an expander.
\QED
\end{proof}

\begin{nfbox}{Non-expanding MPC in the $(\felectshare,\freconcompute,\foutdist)$-hybrid model}{fig:protne}
\begin{center}
    \textbf{Protocol} $\protne_n$
\end{center}
\begin{itemize}
    \item\textbf{Hybrid Model:}
    The protocol is defined in the $(\felectshare,\freconcompute,\foutdist)$-hybrid model.
    \item\textbf{Common Input:}
    A $(t',n')$ ECSS scheme $(\Share,\Recon)$, a signature scheme $(\Gen,\Sign,\Verify)$, and a partition of the party-set $\PS=\sset{\Party_1,\ldots,\Party_n}$ into $\PS_1=\sset{\Party_1,\ldots,\Party_m}$ and $\PS_2=\PS\setminus \PS_1$.
    \item\textbf{PKI:}
    Every party $\Party_i$, for $i\in[n]$, has signature keys $(\sk[i],\vk[i])$; the signing key $\sk[i]$ is private, whereas the vector of verification keys $(\vk[1],\ldots,\vk[n])$ is public and known to all parties.
    \item\textbf{Private Input:}
    Every party $\Party_i$, for $i\in[n]$, has private input $x_i\in\zs$.
    \item\textbf{The Protocol:}
\end{itemize}
\begin{enumerate}
    \item\label{step:elect_share}
    The parties in $\PS_1$ invoke $\felectsharefull$, where every $\Party_i\in\PS_1$ sends input $(x_i,\sk[i])$, and receives back output consisting of a committee $\committee_1=\sset{i_{(1,1)},\ldots,i_{(1,n')}}\subseteq[m]$. Every party $\Party_i$ with $i=i_{(1,j)}\in\committee_1$, receives an additional output consisting of a signature $\sigma_1$ on $\committee_1$, a committee $\committee_2=\sset{i_{(2,1)},\ldots,i_{(2,n')}}\subseteq[m]$, a signature $\sigma_2$ on $\committee_2$, and a vector $\vs_j=(s_1^j,\ldots,s_m^j)$.
    \item\label{step:first_bridge}
    For every $j\in[n']$, party $\Party_{i_{(1,j)}}$ sends $(\committee_1,\sigma_1,\committee_2,\sigma_2)$ to every party in $\committee_2$, and $\vs_j$ only to $\Party_{m+i_{(2,j)}}$.

    A party $\Party_{m+i}\in\PS_2$ that receives a message $(\committee_1,\sigma_1,\committee_2,\sigma_2)$ from $\Party_j\in\PS_1$ will discard the message in the following cases:
    \begin{enumerate}
        \item
        If $i\notin\committee_2$ or $j\notin\committee_1$.
        \item
        If $\Verify_{\vk[1],\ldots,\vk[m]}(\committee_1,\sigma_1)=0$ or $\Verify_{\vk[1],\ldots,\vk[m]}(\committee_2,\sigma_2)=0$.
    \end{enumerate}
    \item\label{step:recon_compute}
    The parties in $\PS_2$ invoke $\freconcomputefull$, where $\Party_{m+i}\in\PS_2$ sends input $(x_{m+i},z_{m+i})$ such that for $i\notin\committee_2$, set $z_{m+i}=\emptystr$, and for $i=i_{(2,j)}\in\committee_2$, set $z_{m+i}=(\committee_2,\sigma_2,\vs_j)$. Every party in $\PS_2$ receives back output $y$.
    \item\label{step:second_bridge}
    For every $j\in[n']$, party $\Party_{m+i_{(2,j)}}$ sends $y$ to party $\Party_{i_{(1,j)}}$. In addition, every party in $\PS_2$ outputs $y$ and halts.
    \item\label{step:out_dist}
    The parties in $\PS_1$ invoke $\foutdistfull$, where party $\Party_i$, with $i\in\committee_1$, has input $y$, and party $\Party_i$, with $i\notin\committee_1$ has the empty input $\emptystr$. Every party in $\PS_1$ receives output $y$, outputs it, and halts.
\end{enumerate}
\end{nfbox}

\begin{proposition}\label{prop:low_locality}
Let $\rho_1,\rho_2,\rho_3$, and $\pi^{f_i\rightarrow \rho_i}$ be the protocols defined in \cref{lem:mpc_no_expander}, and let $\ell_\rho=\ell_\rho(m)$ be the upper bound of the locality of $\rho_1,\rho_2,\rho_3$. Then $\pi^{f_i\rightarrow \rho_i}$ has locality $\ell=2\cdot\ell_\rho + \log^2(n)$.
\end{proposition}

\begin{proof}
Every party in $\PS_1$ communicates with $\ell_\rho$ parties when executing $\rho_1$, and with at most another $\ell_\rho$ parties when executing $\rho_3$. In addition, every party in $\committee_1$ communicates with all $n'=\log^2(n)$ parties in $\committee_2$. Similarly, every party in $\PS_2$ communicates with $\ell_\rho$ parties when executing $\rho_2$, and parties in $\committee_2$ communicates with all $n'$ parties in $\committee_1$. It follows that maximal number of parties that a party communicates with during the protocol is $2\cdot\ell_\rho+\log^2(n)$.
\QED
\end{proof}
This concludes the proof of \cref{lem:mpc_no_expander}.
\QED
\end{proof}

\subsection{Information-Theoretic Security}\label{sec:ne_mpc_static_it}

The protocol in \cref{sec:ne_mpc_static_computational} relies on digital signatures, hence, security is guaranteed only in the presence of computationally bounded adversaries. Next, we gain security facing all-powerful adversaries by using information-theoretic signatures (see \cref{sec:itsign}).

\begin{theorem}[restating Item~2 of \cref{thm:intro:UB_extensions}]\label{thm:mpc_no_expander_it}
Let $f=\sset{f_n}_{n\in\N}$ be an ensemble of functionalities and let $\delta>0$.
The following holds in the IT-PKI-hybrid model with secure channels:
\begin{enumerate}
    \item
    Let $\beta<1/4-\delta$ and let $t=\beta \cdot n$.
    Then, $f$ can be $t$-securely computed by a protocol ensemble~$\pi$ tolerating static $t(n)$-adversaries such that the communication graph of $\pi$ is strongly not an expander.
    \item
    Let $\beta<1/12-\delta$ and let $t=\beta \cdot n$.
    Then, $f$ can be $t$-securely computed by a protocol ensemble~$\pi$ tolerating static $t(n)$-adversaries such that (1) the communication graph of $\pi$ is strongly not an expander, and (2) the locality of $\pi$ is poly-logarithmic in $n$.
\end{enumerate}
\end{theorem}

\begin{proof}
The theorem follows from \cref{lem:mpc_no_expander_it} (below), which is an information-theoretic variant of \cref{lem:mpc_no_expander}, by instantiating the hybrid functionalities using appropriate MPC protocols.
\begin{itemize}
    \item
    The first part follows using honest-majority MPC protocols that exist in the secure-channels model, \eg the protocol of Rabin and Ben-Or \cite{RB89}.
    \item
    In~\Cref{sec:it_bgt} we prove information-theoretic variant of the low-locality MPC protocol of Boyle et al.\ \cite{BGT13} in the IT-PKI model with secure channels that tolerates $t<(1/6-\delta)\cdot n$ static corruptions.
    The second part follows from that protocol.
\qedhere
\end{itemize}
\end{proof}

The main differences between standard digital signatures and information-theoretic signatures are: (1) the verification key is not publicly known, but rather, must be kept hidden (meaning that each party $\Party_i$ has a different verification key $\vk[i]^j$ with respect to every party $\Party_j$), and (2) a bound on the number times that a secret signing key and a secret verification key are used must be a priori known. The latter does not form a problem since, indeed, the number of signatures that are generated in Protocol~$\protne_n$ (\cref{fig:protne}) by any of the signing keys is $2$, and likewise, each verification key is used to verify $2$ signatures. However, the former requires adjusting the functionality $\freconcompute$.

Instead of having the functionality be parametrized by the vector of verification keys $\vk[1],\ldots,\vk[m]$ (which, as mentioned above, will not be secure in the information-theoretic setting), each party $\Party_{m+i}$, with $i\in[m]$, has a vector of (secret) verification keys $\vk[m+i]^1,\ldots,\vk[m+i]^m$ corresponding to the parties in $\PS_1$.

In the adjusted functionality, denoted $\fitreconcompute$, we change the functionality $\freconcompute$ by having each party $\Party_{m+i}$ provide an additional input consisting of its verification keys $\vk[m+i]^1,\ldots,\vk[m+i]^m$ (and the functionality is no longer parametrized by any value). Now, on each input consisting of a subset $\committee_{m+i}$ and corresponding signature-vector $\sigma_{m+i}$, the functionality verifies the \jth signature of the set $\committee_{m+i}$ using the verification keys $\vk[m+1]^j,\ldots,\vk[2m]^j$. If at most $t$ verifications fail, the functionality accepts the committee, whereas if more than $t$ verifications fail, the functionality ignores the subset. Next, the functionality proceeds as in $\freconcompute$.

\paragraph{Signature notations.}
Given an information-theoretically $(\sigcals,\vercals)$-secure signature scheme $(\Gen,\Sign,\Verify)$, and $m$ tuples of signing and verification keys $(\sk[i],\vvk[i])\gets\Gen(1^\secParam,n,\sigcals)$, where for every $i\in[m]$ the verification keys are $\vvk[i]=(\vk[1]^i,\ldots,\vk[n]^i)$, we use the following notations for signing and verifying with multiple keys:
\begin{itemize}
    \item
    Given a message $\mu$ we denote by $\Sign_{\sk[1],\ldots,\sk[m]}(\mu)$ we consider the vector of $m$ signatures $\sigma=(\sigma_1,\ldots,\sigma_m)$, where $\sigma_i\gets\Sign_{\sk[i]}(\mu)$.
    \item
    Given a message $\mu$ and a signature $\sigma=(\sigma_1,\ldots,\sigma_m)$, we denote by $\Verify_{\vvk[m+1],\ldots,\vvk[2m]}(\mu,\sigma)$ the following verification algorithm:
    \begin{enumerate}
        \item
        For every $i\in[m]$ proceeds as follows:
        \begin{enumerate}
            \item
            For every $j\in[m]$ let $b_i^j\gets\Verify_{\vk[m+j]^i}(\mu,\sigma_i)$.
            \item
            Set $b_i=1$ if and only if $\sum_{j=1}^m b_i^j \geq m-t$, \ie even if up to $t$ verification keys reject the signature.
        \end{enumerate}
        \item
        Accepts $\sigma$ if and only if $\sum_{i=1}^m b_i \geq m-t$, \ie even if up to $t$ signatures are invalid.
    \end{enumerate}
\end{itemize}

\begin{lemma}\label{lem:mpc_no_expander_it}
Let $f=\sset{f_n}_{n\in2\N}$, where $f_n$ is an $n$-party functionality for $n=2m$, and let $\delta>0$.
Then, in the $(2,2)$-IT-PKI-hybrid model with secure channels, where a trusted party additionally computes the $m$-party functionality-ensembles $(\felectshare,\fitreconcompute,\foutdist)$ tolerating $\gamma \cdot m$ corruptions, there exists a protocol ensemble $\pi$ that securely computes $f$ tolerating static $\beta n$-adversaries, for $\beta<\min(1/4-\delta,\gamma/2)$, with the following guarantees:
\begin{enumerate}
    \item
    The communication graph of $\pi$ is strongly not an expander.
    \item
    Denote by $f_1,f_2,f_3$ the functionality-ensembles $\felectshare,\fitreconcompute,\foutdist$ (resp.).
    If protocol-ensembles $\rho_1,\rho_2,\rho_3$ securely compute $f_1,f_2,f_3$ (resp.) with locality $\ell_\rho=\ell_\rho(m)$, then $\pi^{f_i\rightarrow \rho_i}$ (where every call to $f_i$ is replaced by an execution of $\rho_i$) has locality $\ell=2\cdot\ell_\rho + \log^2(n)$.
\end{enumerate}
\end{lemma}

\begin{proof}[Proof sketch]
The proof of the lemma follows in similar lines as the proof of \cref{lem:mpc_no_expander}. We highlight the main differences.

\begin{itemize}
    \item\textbf{IT-PKI model.}
    The protocol is defined in the $(2,2)$-IT-PKI-hybrid model, rather than the PKI-hybrid model, meaning that each party receives a secret signing key along with a secret verification key for every other party. At most two values can signed verified by these keys.

    \item\textbf{Hybrid model.}
    The protocol is defined in the $\fitreconcompute$-hybrid model, rather than the $\freconcompute$-hybrid model, meaning that each party in $\PS_2$ sends its vector of secret verification keys as input to the functionality.

    \item\textbf{The simulation.}
    The simulator generates appropriate signing and verification keys for simulating the $(2,2)$-IT-PKI functionality, and receives the verification keys from the adversaries when emulating the functionality $\fitreconcompute$. No other changes are needed.

    \item\textbf{The hybrid games.}
    The first hybrid game and the second are statistically close when using $(2,2)$-IT-PKI, rather than computationally indistinguishable (see \cref{claim:hyb1}).
\end{itemize}

The rest of the proof follows the proof of \cref{lem:mpc_no_expander}.
\QED
\end{proof}


\subsection{Information-Theoretic MPC with Low Locality}\label{sec:it_bgt}

The protocol of Boyle, Goldwasser, and Tessaro \cite{BGT13} follows a framework common to other protocols achieving low locality (see \cite{KS09,KLST11,BGH13} and the references therein).
First, the parties compute almost-everywhere agreement, that is agreement among at least a $ 1- o(1) $ fraction of parties.
Next, the parties upgrade to full agreement via a transformation that preserves low locality.
The results in~\cite{BGT13} are in the computational setting where the main cryptographic tools that are being used are public-key encryption, digital signatures, and pseudorandom functions (PRF).
In this section, we show that the approach of~\cite{BGT13} can be adapted to the information-theoretic setting by removing the need of public-key encryption and by substituting other computational primitives by their information-theoretic analogues.
Namely, we will use information-theoretic signatures instead of digital signatures and samplers~\cite{Z97,G11} instead of PRF.

\paragraph{Overview of the BGT protocol.}
The protocol consists of two parts:
\begin{enumerate}
    \item\textbf{Establishing a polylog-degree communication tree.}
    This part of the protocol requires digital signatures established via a PKI and a PRF.
    \begin{itemize}
        \item
        Initially, the parties run the protocol of King et al.\ \cite{KSSV06} to reach \emph{almost everywhere} agreement on a random seed while establishing a polylog-degree communication tree, and maintaining polylog locality. This part holds information theoretically.
        \item
        Next, \emph{certified} almost everywhere agreement is obtained by having the parties sign the seed and distribute the signatures. Specifically, every party sends its signature on the seed up the tree to the supreme committee, which concatenates all signatures to form a certificate on the seed, and sends it down the tree to (almost all) the parties.
        \item
        Finally, to achieve \emph{full} agreement, every party that received sufficiently many signatures on the seed locally evaluates the PRF on the seed and its identity to get a polylog subset of parties, and sends the certified seed to each party in this set. A party that receives the certified seed can validate the seed is properly signed and that he is a valid recipient of the message.
    \end{itemize}

    Note that the PRF is used for its combinatorics properties and is not needed for security.
    \item\textbf{Computing the function.}
    Having established the communication tree, the supreme committee (\ie the parties assigned to the root) jointly generate keys to a threshold encryption scheme such that each committee member holds a share of the decryption key and the public key is known. Next, they distribute the public encryption key down the tree. Every party encrypts its input using the encryption key and sends it up the tree. Finally, the supreme committee runs a protocol to decrypt the ciphertexts and evaluate the function to obtain the output, which is distributed to all parties.
\end{enumerate}

We now turn to explain how to construct an information-theoretic analogue for this protocol.

\paragraph{Establishing the communication tree information theoretically.}
This part follows almost immediately from \cite{BGT13}. The digital signatures and the PKI are replaced by information-theoretic signatures and an IT-PKI, where every party signs the $\secParam$-bit seed (\ie one signature operation per party) and has to verify $n$ signatures. As mentioned above, the PRF is used only for its combinatorial properties (mapping each party to a polylog set of neighbors) and not for other security purposes, and so it can be replaced by a sampler with good parameters (this approach was adopted by~\cite{BGH13} to construct the first BA protocol for with polylog communication complexity).
We provide more information on the samplers that are employed in \cref{sec:samplers}.

\paragraph{Computing the function information theoretically.}
Once the communication tree is established, each party must send his input to the supreme-committee members in a way that allows them to compute the function. We replace the public-key encryption used in \cite{BGT13} by secret sharing. To understand this step, we will first explain the structure of the communication tree.

For any $n\in \N$, the communication tree from \cite{KSSV06} is a graph $G=G(n)$ in which every node is labeled by subsets of $[n]$ that satisfies the following properties:
\begin{itemize}
	\item
    $G$ is a tree of height $\ls \in O(\log n/\log\log n)$. Each node from level $\ell > 0$ has $\log n$ nodes from level $\ell - 1$ as its children.
	\item
    Each node of $G$ is labeled by a subset $\polylog(n)$ parties.
	\item
    Each party is assigned to $\polylog(n)$ nodes at each level.
\end{itemize}
King et al.\ \cite{KSSV06} showed that for any $t=\beta n$ static corruptions, all but a $3/\log n$ fraction of the leaf nodes have a good path up to the root node (\ie a path on which each committee contains a majority of honest parties). As observed in \cite{BGT13}, this implies that for a $1-o(1)$ fraction of the parties, majority of the leaf nodes that they are assigned to are good. In addition, each leaf node is connected to $\polylog(n)$ parties as determined by the sampler.

The high-level idea now is to let each party $\Party_i$, with associated leaf-nodes $v_1,\ldots,v_k$, secret share its input as $(s_{i,1},\ldots, s_{i,k})\gets\Share(x_i)$ and send the share $s_{i,j}$ to $v_j$. Each leaf-node will send the received shares up the tree until to supreme committee (the parties associated with the root) receive all shares, reconstruct all inputs, and compute the function. Clearly, this idea does not provide any privacy, since the adversary may have corrupted parties in many leaves, thus recover the honest parties' inputs values. To overcome this problem, instead of sending $s_{i,j}$ in the clear, each $\Party_i$ will secret share each share as $(s^1_{i,j},\ldots, s^m_{i,j})\gets\Share(s_{i,j})$, where $m$ is the size of a committee associated to a leaf node, and send $s^h_{i,j}$ to the $h$'th party in $v_j$. Stated differently, each leaf-node will hold the shares in a secret-shared form.

The next part of the protocol proceeds recursively. For every node $v$ in level $\ell$ and every child node $u$ of $v$ in level $\ell-1$, the parties associated with $u$ and with $v$ will run a secure protocol for the following functionality: For each of the shared values held by the parties associated with $u$, they enter the secret shares as input; the functionality reconstructs the value, reshares it, and outputs the new shares to the parties associated with $v$.

In order to implement this functionality using BGW, we require that the union of the parties associated with $u$ and with $v$ will have a 2/3 majority. Such a majority is guaranteed with overwhelming probability if the total fraction of corruptions is $1/6-\epsilon$, for an arbitrary small constant~$\epsilon$.
We thus proved the following theorem.

\begin{theorem}[restating \cref{thm:intro:IT_BGT}]\label{thm:IT_BGT}
For any efficient functionality $f$ and any constant $\epsilon>0$, there exists a protocol with poly-logarithmic locality in the information-theoretic PKI model, securely realizing $f$ against computationally unbounded adversaries statically corrupting $(1/6-\epsilon) \cdot n$ parties.
\end{theorem}

\subsection{Adaptive Corruptions}\label{sec:ne_mpc_adaptive}

In this section, we focus on the adaptive setting, where the adversary can corrupt parties dynamically, based on information gathered during the course of the protocol.

Adjusting \cref{lem:mpc_no_expander} to the adaptive setting is not straightforward, since once the subsets $\committee_1$ and $\committee_2$ are known to the adversary, he can completely corrupt them. A first attempt to get around this obstacle is not to reveal the entire subsets in the output of the Elect-and-Share functionality, but rather, let each party in $\committee_1$ learn the identity of a single party in $\committee_2$ with which he will communicate. This way, if a party in $\committee_1$ (resp.\ $\committee_2$) gets corrupted, only one additional party in $\committee_2$ (resp.\ $\committee_1$) is revealed to the adversary. This solution comes with the price of tolerating a smaller fraction of corrupted parties, namely, $(1/8 - \delta)$ fraction.

This solution, however, is still problematic if the adversary can monitor the communication lines, even when they are completely private (as in the secure-channels setting). The reason is that once the adversary sees the communication that is sent between $\committee_1$ and $\committee_2$ he can completely corrupt both subsets. This is inherent when the communication lines are visible to the adversary; therefore, we turn to the hidden-channels setting that was used by \cite{CCGGOZ15}, where the adversary does not learn whether a message is sent between two honest parties (see \cref{sec::Def:model}).

\begin{theorem}[restating Items~3 and~4 of \cref{thm:intro:UB_extensions}]\label{thm:mpc_no_expander_adaptive}
Let $f=\sset{f_n}_{n\in\N}$ be an ensemble of functionalities, let $\delta>0$, let $\beta<1/8-\delta$, and let $t=\beta \cdot n$.
The following holds in the hidden-channels model:
\begin{enumerate}
    \item
    Assuming the existence of one-way functions, $f$ can be securely computed by a protocol ensemble~$\pi$ in the PKI model tolerating adaptive PPT $t(n)$-adversaries such that the communication graph of $\pi$ is strongly not an expander.
    \item
    Assuming the existence of trapdoor permutations with a reverse domain sampler, $f$ can be securely computed by a protocol ensemble $\pi$ in the PKI and SKI model tolerating adaptive PPT $t(n)$-adversaries such that (1) the communication graph of $\pi$ is strongly not an expander, and (2) the locality of $\pi$ is poly-logarithmic in $n$.\footnote{We note that the adaptively secure protocols in~\cite{CCGGOZ15} are  proven in a model with atomic simultaneous multi-send operations~\cite{HZ10,GKKZ11} and secure erasures.}
    \item
    $f$ can be securely computed by a protocol ensemble~$\pi$ in the IT-PKI model tolerating adaptive $t(n)$-adversaries such that the communication graph of $\pi$ is strongly not an expander.
\end{enumerate}
\end{theorem}

\begin{proof}
The theorem follows from \cref{lem:mpc_no_expander_adaptive} (below), which is an adaptively secure variant of \cref{lem:mpc_no_expander}, by instantiating the hybrid functionalities using MPC protocols from the literature.
\begin{itemize}
    \item
    The first part follows using an adaptively secure honest-majority MPC protocol in the secure-channels model, \eg Cramer et al.\ \cite{CDDHR99} or \Damgard and Ishai \cite{DI05}.
    \item
    The second part follows using the adaptively secure protocol of Chandran et al.\ \cite{CCGGOZ15}.
    \item
    The third part follows using information-theoretic signatures via the same adjustments that were employed in \cref{sec:ne_mpc_static_it}, and using the protocol of Cramer et al.\ \cite{CDDHR99}.
\qedhere
\end{itemize}
\end{proof}

Hiding the subsets $\committee_1$ and $\committee_2$ from the parties requires adjusting the ideal functionalities that are used in \cref{sec:ne_mpc_static_computational}. We now describe the adjusted functionalities.

\paragraph{The Adaptive-Elect-and-Share functionality.}

The Adaptive-Elect-and-Share $m$-party functionality, $\faelectsharefull$, is defined in a similar way as the Elect-and-Share functionality (\cref{fig:felectshare}) with the following difference. Instead of outputting the set $\committee_1$ to all parties and the set $\committee_2$ to parties in $\committee_1$, the functionality outputs for every party in $\committee_1$ an index of a single party in $\committee_2$ (and signs the values). Parties outside of $\committee_1$ receive no output.
The formal description of the functionality can be found in~\cref{fig:faelectshare}.

\begin{nfbox}{The Adaptive-Elect-and-Share functionality}{fig:faelectshare}
\begin{center}
    \textbf{The functionality} $\faelectsharefull$
\end{center}
\vspace{-0.3cm}
The $m$-party functionality $\faelectsharefull$ is parametrized by a signature scheme $(\Gen,\Sign,\Verify)$ and a $(t',n')$ ECSS scheme $(\Share,\Recon)$, and proceeds with parties $\PS_1=\sset{\Party_1,\ldots,\Party_m}$ as follows.
\begin{enumerate}
    \item
    Every party $\Party_i$ sends a pair of values $(x_i,\sk[i])$ as its input, where $x_i$ is the actual input value and $\sk[i]$ is a signing key. (If $\Party_i$ didn't send a valid input, set it to the default value, \eg zero.)
    \item
    Sample uniformly at random two subsets (committees) $\committee_1,\committee_2\subseteq[m]$ of size $n'$.
    Denote $\committee_1=\sset{i_{(1,1)},\ldots,i_{(1,n')}}$ and $\committee_2=\sset{i_{(2,1)},\ldots,i_{(2,n')}}$.
    \item
    For every $j\in[n']$, sign $\sigma_{1,j}=\Sign_{\sk[1],\ldots,\sk[m]}(i_{(1,j)})$ and $\sigma_{2,j}=\Sign_{\sk[1],\ldots,\sk[m]}(m+i_{(2,j)})$.
    \item
    For every $i\in[m]$, secret share $x_i$ as $(s_i^1,\ldots,s_i^{n'})\gets \Share(x_i)$.
    \item
    For every $j\in[n']$, set $\vs_j=(s_1^j,\ldots,s_m^j)$.
    \item
    For every $i=i_{(1,j)}\in\committee_1$ (for some $j\in[n']$), set the output of $\Party_i$ to be $(i_{2,j},\sigma_{1,j},\sigma_{2,j},\vs_j)$. (Parties outside of $\committee_1$ receive the empty output $\emptystr$)
\end{enumerate}
\end{nfbox}

\paragraph{The Adaptive-Reconstruct-and-Compute functionality.}

The Adaptive-Reconstruct-and-Compute functionality, $\fareconcomputefull$, is defined in a similar way as the Reconstruct-and-Compute functionality (\cref{fig:freconcompute}) with the following difference. Instead of having the potential additional input value consist of a signed subset $\committee_2\subseteq[m]$, it consists of a signed index. The functionality verifies that if a party provided an additional input, then it has a valid signature of its own index, and derives the committee $\committee_2$ from the indices with a valid signature. The formal description of the functionality can be found in~\cref{fig:fareconcompute}.

\begin{nfbox}{The Adaptive-Reconstruct-and-Compute functionality}{fig:fareconcompute}
\begin{center}
    \textbf{The functionality} $\fareconcompute$
\end{center}
\vspace{-0.3cm}
The $m$-party functionality $\fareconcomputefull$ is parametrized by a signature scheme $(\Gen,\Sign,\Verify)$, a $(t',n')$ ECSS scheme $(\Share,\Recon)$, and a vector of verification keys $(\vk[1],\ldots,\vk[m])$, and proceeds with parties $\PS_2=\sset{\Party_{m+1},\ldots,\Party_{2m}}$ as follows.
\begin{enumerate}
    \item
    Every party $\Party_{m+i}$ sends a pair of values $(x_{m+i},z_{m+i})$ as its input, where $x_{m+i}$ is the actual input value and either $z_{m+i}=\emptystr$ or $z_{m+i}=(\sigma_{m+i},\vs_{m+i})$.
    \item
    For every $\Party_{m+i}$ with $z_{m+i}\neq \emptystr$, check if $\Verify_{\vk[1],\ldots,\vk[m]}(m+i,\sigma_{m+i})=1$ (ignore invalid inputs).
    If exactly $n'$ indices are properly signed, denote as $\committee_2=\sset{i_{(2,1)}\ldots,i_{(2,n')}}$; otherwise, abort.
    \item
    For every $i=i_{(2,j)}\in\committee_2$, let $\vs_{m+i_{(2,j)}}=(s_1^j,\ldots,s_m^j)$ be the input provided by $\Party_{m+i}$. (If $\Party_{m+i}$ provided invalid input, set $\vs_{m+i_{(2,j)}}$ to be the default value, \eg the zero vector.)
    \item
    For every $i\in[m]$, reconstruct $x_i=\Recon(s_i^1,\ldots,s_i^{n'})$.
    \item
    Compute $y=f(x_1\ldots, x_m,x_{m+1},\ldots,x_{2m})$.
    \item
    Output $y$ to every $\Party_{m+i}\in\PS_2$.
\end{enumerate}
\end{nfbox}

\paragraph{The Adaptive-Output-Distribution functionality.}
The Adaptive-Output-Distribution functionality is defined in a similar way as the Output-Distribution functionality (\cref{fig:foutdist}) with the following difference. Instead of being parametrized by a subset $\committee_1$ that specifies the input providers, the functionality is parametrized by the verification keys $\vk[1],\ldots,\vk[m]$, every party that provides an input value must also provide a signature of its own index. The formal description of the functionality can be found in~\cref{fig:faoutdist}.

\begin{nfbox}{The Adaptive-Output-Distribution functionality}{fig:faoutdist}
\begin{center}
    \textbf{The functionality} $\faoutdist$
\end{center}
\vspace{-0.3cm}
The $m$-party functionality $\faoutdistfull$ is parametrized by a signature scheme $(\Gen,\Sign,\Verify)$ and a vector of verification keys $(\vk[1],\ldots,\vk[m])$, and proceeds with parties $\PS_1=\sset{\Party_1,\ldots,\Party_m}$ as follows.
\begin{enumerate}
    \item
    Every party $\Party_i$ gives either the empty input~$\emptystr$ of an input value of the form $(\sigma_i, y_i)$.
    \item
    For every party $\Party_i$ that provided a non-empty input, verify that $\Verify_{\vk[1],\ldots,\vk[m]}(i,\sigma_i)=1$ (ignore invalid inputs). Denote by $\committee_1$ the set of indices of parties that gave valid inputs.
    \item
    Denote $y=\maj\sset{y_i \mid i\in\committee_1}$ (choose arbitrarily if the majority is not unique).
    \item
    Output $y$ to every $\Party_i\in\PS_1$.
\end{enumerate}
\end{nfbox}

\begin{lemma}\label{lem:mpc_no_expander_adaptive}
Let $f=\sset{f_n}_{n\in2\N}$, where $f_n$ is an $n$-party functionality for $n=2m$, and let $\delta>0$.
Then, in the PKI-hybrid model with secure channels, where a trusted party additionally computes the $m$-party functionality-ensembles $(\faelectshare,\fareconcompute,\faoutdist)$ tolerating $\gamma \cdot m$ corruptions, there exists a protocol ensemble $\pi$ that securely computes $f$, with computational security, tolerating adaptive PPT $\beta n$=adversaries, for $\beta<\min(1/8-\delta,\gamma/2)$ with the following guarantees:
\begin{enumerate}
    \item
    The communication graph induced by $\pi$ is strongly not an expander.
    \item
    Denote by $f_1,f_2,f_3$ the functionalities $\faelectshare,\fareconcompute,\faoutdist$ (resp.).
    If protocols $\rho_1,\rho_2,\rho_3$ securely compute $f_1,f_2,f_3$ (resp.) with locality $\ell_\rho=\ell_\rho(m)$, then $\pi^{f_i\rightarrow \rho_i}$ (where every call to $f_i$ is replaced by an execution of $\rho_i$) has locality $\ell=2\cdot\ell_\rho + 1$.
\end{enumerate}
\end{lemma}

\ifdefined\IsProofInAppendix
The proof of \cref{lem:mpc_no_expander_adaptive} can be found in \cref{sec:lem:mpc_no_expander_adaptive_proof}.
\else

\begin{proof}
For $m\in\N$ and $n=2m$, we construct the $n$-party protocol $\protane_n$ (see \cref{fig:protane}) in the $(\faelectshare,\fareconcompute,\faoutdist)$-hybrid model. As in the proof of \cref{{lem:mpc_no_expander}}, the parameters for the protocol are $n'=\log^2(n)$ and $t'=(1/2-\delta)\cdot n'$. We start by proving in \cref{prop:security_adaptive} that the protocol $\protane_n$ securely computes $f_n$. Next, in \cref{prop:not_expanding_adaptive} we prove that the communication graph of $\protane$ is strongly not an expander. Finally, in \cref{prop:low_locality_adaptive} we prove that by instantiating the functionalities $(\faelectshare,\fareconcompute,\faoutdist)$ using low-locality protocols, the resulting protocol has low locality.

\begin{proposition}\label{prop:security_adaptive}
For sufficiently large $n$, the protocol $\protane_n$ securely computes the function $f_n$, tolerating adaptive, PPT $\beta n$ corruptions, in the $(\faelectshare,\fareconcompute,\faoutdist)$-hybrid model.
\end{proposition}
\begin{proof}
Let $\Adv$ be an adversary attacking the execution of $\protane_n$ and let $\Env$ be an environment. We construct an ideal-process adversary $\Sim$, interacting with the environment $\Env$, the ideal functionality $f_n$, and with ideal (dummy) parties $\DParty_1, \ldots, \DParty_n$.
The simulator $\Sim$ constructs virtual parties $\Party_1, \ldots, \Party_n$, and runs the adversary $\Adv$.
Note that the protocol $\protane$ is deterministic and the only randomness arrives from the ideal functionalities. Therefore, upon a corruption request, the simulator is only required to provide the party's input, interface with the ideal functionalities, and possibly the output. Denote by $\IS$ the set of corrupted parties (note that this set is dynamic: initially it is set to $\emptyset$ and it grows whenever an honest party gets corrupted).

\begin{nfbox}{Non-expanding MPC in the $(\faelectshare,\fareconcompute,\faoutdist)$-hybrid model}{fig:protane}
\begin{center}
    \textbf{Protocol} $\protane_n$
\end{center}
\begin{itemize}
    \item\textbf{Hybrid Model:}
    The protocol is defined in $(\faelectshare,\fareconcompute,\faoutdist)$-hybrid model.
    \item\textbf{Common Input:}
    A $(t',n')$ ECSS scheme $(\Share,\Recon)$, a signature scheme $(\Gen,\Sign,\Verify)$, and a partition of the party-set $\PS=\sset{\Party_1,\ldots,\Party_n}$ into $\PS_1=\sset{\Party_1,\ldots,\Party_m}$ and $\PS_2=\PS\setminus \PS_1$.
    \item\textbf{PKI:}
    Every party $\Party_i$, for $i\in[n]$, has signature keys $(\sk[i],\vk[i])$; the signing key $\sk[i]$ is private, whereas the vector of verification keys $(\vk[1],\ldots,\vk[n])$ is public and known to all parties.
    \item\textbf{Private Input:}
    Every party $\Party_i$, for $i\in[n]$, has private input $x_i\in\zs$.
    \item\textbf{The Protocol:}
\end{itemize}
\begin{enumerate}
    \item\label{step:a_elect_share}
    The parties in $\PS_1$ invoke $\faelectsharefull$, where every $\Party_i\in\PS_1$ sends input $(x_i,\sk[i])$, and receives back either the empty output, or an output consisting of an index $i_{(2,j)}\in[m]$, a signature $\sigma_{1,j}$ of its own index (denoted $i_{1,j}$), a signature $\sigma_{2,j}$ of the index $m+i_{(2,j)}$, and a vector $\vs_j=(s_1^j,\ldots,s_m^j)$.
    \item\label{step:a_first_bridge}
    Denote $\committee_1=\sset{i_{1,1},\ldots,i_{1,n'}}$.
    Every $\Party_i$ with $i=i_{(1,j)}\in\committee_1$ sends $(\sigma_{1,j},\sigma_{2,j},\vs_j)$ to $\Party_{m+i_{(2,j)}}$.

    A party $\Party_{i_2}\in\PS_2$ that receives a message $(\sigma_{1,j},\sigma_{2,j},\vs_j)$ from $\Party_{i_1}\in\PS_1$ will discard the message in the following cases:
    \begin{enumerate}
        \item
        If $\Verify_{\vk[1],\ldots,\vk[m]}(i_1,\sigma_{1,j})=0$.
        \item
        If $\Verify_{\vk[1],\ldots,\vk[m]}(i_2,\sigma_{2,j})=0$.
    \end{enumerate}
    \item\label{step:a_recon_compute}
    Denote the set of parties that received a valid message in Step~\ref{step:a_first_bridge} by $\committee_2$.
    The parties in $\PS_2$ invoke $\fareconcomputefull$, where $\Party_{m+i}\in\PS_2$ sends input $(x_{m+i},z_{m+i})$ such that for $m+i\notin\committee_2$, set $z_{m+i}=\emptystr$, and for $m+i=m+i_{(2,j)}\in\committee_2$, set $z_{m+i}=(\sigma_{(2,j)},\vs_j)$. Every party in $\PS_2$ receives the output $y$.
    \item\label{step:a_second_bridge}
    For every $j\in[n']$, party $\Party_{m+i_{(2,j)}}$ sends $y$ to party $\Party_{i_{(1,j)}}$. In addition, every party in $\PS_2$ outputs $y$ and halts.
    \item\label{step:a_out_dist}
    The parties in $\PS_1$ invoke $\faoutdist^{(\vk[1],\ldots,\vk[m])}$, where party $\Party_i$, with $i\in\committee_1$, has input $(\sigma_{(1,j)},y)$, and party $\Party_i$, with $i\notin\committee_1$ has the empty input $\emptystr$. Every party in $\PS_1$ receives output $y$, outputs it, and halts.
\end{enumerate}
\end{nfbox}

\paragraph{Simulating communication with the environment:}
In order to simulate the communication with $\Env$, every input value that $\Sim$ receives from $\Env$ is written on $\Adv$'s input tape.
Likewise, every output value written by $\Adv$ on its output tape is copied to $\Sim$'s own output tape.

\paragraph{Simulating the PKI:}
The simulator $\Sim$ generates $(\sk[i],\vk[i])\gets\Gen(1^\secParam)$ for every $i\in[n]$.

\paragraph{Simulating the protocol:}

To simulate Step~\ref{step:a_elect_share}, the simulator emulates $\faelectsharefull$ towards $\Adv$ as follows. $\Sim$ receives from $\Adv$ input values $\sset{(x'_i,\sk[i]')}_{i\in\IS\cap[m]}$ (replace invalid inputs with default). Next, $\Sim$ samples uniformly at random two subsets $\committee_1=\sset{i_{(1,1)},\ldots,i_{(1,n')}}\subseteq[m]$ and $\committee_2=\sset{i_{(2,1)},\ldots,i_{(2,n')}}\subseteq[m]$, and sets $\sk[i]'=\sk[i]$ for every $i\in[m]\setminus\IS$. For every $j\in[n']$, sign $\sigma_{(1,j)}=\Sign_{\sk[1],\ldots,\sk[m]}(i_{(1,j)})$ and $\sigma_{(2,j)}=\Sign_{\sk[1],\ldots,\sk[m]}(m+i_{(2,j)})$. Finally, generate secret shares of zero for every $i\in[m]$ as $(s_i^1,\ldots, s_i^{n'})\gets\Share(0^\secParam)$, and hand $\Adv$ the output $(i_{(2,j)},\sigma_{i_{(1,j)}},\sigma_{i_{(2,j)}},\vs_j)$, with $\vs_j=(s_1^j\ldots,s_m^j)$, for every $i=i_{(1,j)}\in\committee_1\cap\IS$.

To simulate Step~\ref{step:a_first_bridge}, for every $j\in[n']$ such that $i_{(1,j)}\notin\IS$ and $m+i_{(2,j)}\in\IS$, the simulator $\Sim$ sends to $\Adv$ on behalf of the honest party $\Party_{i_{(1,j)}}$ the value $(\sigma_{i_{(1,j)}},\sigma_{i_{(2,j)}},\vs_j)$, intended to $\Party_{m+i_{(2,j)}}$. In addition, $\Sim$ receives such values from the adversary ($\Sim$ ignores any invalid messages from $\Adv$).

To simulate Step~\ref{step:a_recon_compute}, the simulator emulates $\fareconcompute^{\vk[1],\ldots,\vk[m]}$ towards $\Adv$ as follows. $\Sim$ receives from $\Adv$ input values $\sset{(x'_{m+i},z_{m+i})}_{m+i\in\IS\cap[m+1,2m]}$, where either $z_{m+i}=\emptystr$ or $z_{m+i}=(\sigma_{i_{(2,j)}},\vs_j)$ (replace invalid inputs with default). Next, $\Sim$ sends the values $\sset{x'_i}_{i\in\IS}$ to the trusted party (for $i\in\IS\cap[m]$ the values $x'_i$ were obtained in the simulation of Step~\ref{step:a_elect_share}, and for $i\in\IS\cap[m+1,2m]$ in the simulation of Step~\ref{step:a_recon_compute}) and receives back the output value $y$. Finally, $\Sim$ hands $y$ to $\Adv$ for every corrupted party $\Party_{m+i}$, with $m+i\in\IS\cap[m+1,2m]$.

To simulate Step~\ref{step:a_second_bridge}, the simulator sends $y$ on behalf of every honest party $\Party_{m+i_{(2,j)}}$ with $i_{(2,j)}\in\committee_2$ and $m+i_{(2,j)}\notin\IS$ to every corrupted party $\Party_{i_{(1,j)}}$ with $i_{(1,j)}\in\committee_1\cap\IS$, and receives such values from the adversary.

To simulate Step~\ref{step:a_out_dist}, the simulator emulates the functionality $\faoutdistfull$ by receiving either a value $(\sigma_i,y_i)$ or an empty string $\emptystr$ from $\Adv$ for every $i\in\IS\cap[m]$, and sending $y$ to every corrupted party $\Party_i$ with $i\in\IS\cap[m]$.

\paragraph{Simulating corruption requests by $\Adv$:}
Whenever the adversary $\Adv$ requests to corrupt an honest party $\Party_i$, the simulator $\Sim$ corrupts $\DParty_i$, learn its input $x_i$  and continues as follows to compute the internal state of $\Party_i$, based on the timing of the corruption request:
\begin{itemize}
    \item \textbf{In Step~\ref{step:a_elect_share}, before calling $\faelectshare$:}
    The simulator $\Sim$ sets the contents of $\Party_i$'s input tape to be $x_i$. The secret signing key is set to be $\sk[i]$ and the verification keys are set to be $\vk[1],\ldots,\vk[n]$.
    \item \textbf{In Step~\ref{step:a_elect_share}, after calling $\faelectshare$:}
    In addition to the above, if $\Party_i\in\PS_1$, the simulator $\Sim$ sets the input of $\Party_i$ to $\faelectshare$ to be $(x_i,\sk[i])$. If $i=i_{(1,j)}\in\committee_1$, set the output from $\faelectshare$ to be $(i_{(2,j)},\sigma_{i_{(1,j)}},\sigma_{i_{(2,j)}},\vs_j)$ as computed in the simulation.
    \item \textbf{In Step~\ref{step:a_first_bridge}:}
    In addition to the above, if $i=i_{(1,j)}\in\committee_1$, the simulator sets the outgoing message of $\Party_i$ to $\Party_{m+i_{(2,j)}}$ to be $(\sigma_{i_{(1,j)}},\sigma_{i_{(2,j)}},\vs_j)$. If $i=m+i_{(2,j)}$ with $i_{(2,j)}\in\committee_1$, the simulator set the incoming message of $\Party_i$ from $\Party_{i_{(1,j)}}$ as follows: if $\Party_{i_{(1,j)}}$ is honest set it to be $(\sigma_{i_{(1,j)}},\sigma_{i_{(2,j)}},\vs_j)$; otherwise, set it according to the values sent by $\Adv$.
    \item \textbf{In Step~\ref{step:a_recon_compute}, before calling $\fareconcompute$:}
    In addition to the above, if $\Party_i\in\PS_2$, the simulator $\Sim$ sets the input of $\Party_i$ to $\fareconcompute$ to be $(x_i,z_i)$, where if $i=m+i_{(2,j)}$ with $i_{(2,j)}\in\committee_2$, set $z_i=(\sigma_{i_{(2,j)}},\tilde{\vs}_j)$ where $\tilde{\vs}_j=(\tilde{s}^1_j,\ldots,\tilde{s}^{n'}_j)$, such that if party $\Party_{i_{(1,k)}}$ was corrupted during Step~\ref{step:a_first_bridge} then $\tilde{s}^k_j$ is set according to the values sent by $\Adv$, and if $\Party_{i_{(1,k)}}$ was honest then set $\tilde{s}^k_j=s^k_j$. Otherwise, set $z_i=\emptystr$.
    \item \textbf{In Step~\ref{step:a_recon_compute}, after calling $\fareconcompute$:}
    In addition to the above, if $\Party_i\in\PS_2$, the simulator $\Sim$ sets the output of $\Party_i$ from $\fareconcompute$ to be $y$.
    \item \textbf{In Step~\ref{step:a_second_bridge}:}
    In addition to the above, if $i=m+i_{(2,j)}$ with $i_{(2,j)}\in\committee_2$, the simulator sets the outgoing message of $\Party_i$ to $\Party_{i_{(1,j)}}$ to be $y$. If $i=i_{(1,j)}\in\committee_1$, the simulator sets the incoming message of $\Party_i$ from $\Party_{m+i_{(2,j)}}$ as follows: if $\Party_{m+i_{(2,j)}}$ is honest set it to be $y$; otherwise, set it according to the values sent by $\Adv$.
    \item \textbf{In Step~\ref{step:a_out_dist}, before calling $\faoutdist$:}
    In addition to the above, if $i=i_{(1,j)}\in\committee_1$, the simulator $\Sim$ sets the input of $\Party_i$ to $\faoutdist$ to be $(\sigma_{i_{(1,j)}},y)$.
    \item \textbf{In Step~\ref{step:a_out_dist}, after calling $\faoutdist$:}
    In addition to the above, if $i=[m]$, the simulator $\Sim$ sets the output of $\Party_i$ from $\faoutdist$ to be $y$.
\end{itemize}
Next, $\Sim$ sends the internal state of $\Party_i$ to $\Adv$.

\paragraph{Simulating post-execution corruption requests by $\Env$:}
Whenever the environment $\Env$ requests to corrupt an honest party $\Party_i$ in the post-execution corruption phase, the simulator $\Sim$ proceeds to compute the internal state of $\Party_i$ as in a corruption request from $\Adv$ \emph{in Step~\ref{step:a_out_dist}, after calling $\faoutdist$}, and sends it to $\Env$.

\paragraph{Proving security.}
We prove computational indistinguishability between the execution of the protocol $\protane_n$ running with adversary $\Adv$ and the ideal computation of $f_n$ running with $\Sim$ via a series of hybrids experiments. The output of each experiment is the output of the honest parties and of the adversary.

\paragraph{The game $\HYB^1_{\pi, \Adv, \Env}(\vx,\aux,\secParam)$.}
This game is defined to be the execution of the protocol $\protane_n$ in the $(\faelectshare,\fareconcompute,\faoutdist)$-hybrid model on inputs $\vx\in(\zs)^n$ and security parameter $\secParam$ with adversary $\Adv$ and environment $\Env$ running on auxiliary information $\aux$.

\paragraph{The game $\HYB^2_{\pi, \Adv, \Env}(\vx,\aux,\secParam)$.}
In this game, we modify $\HYB^1_{\pi, \Adv, \Env}(\vx,\aux,\secParam)$ as follows.
Instead of verifying the signatures of the index of every party that provides an additional input, the functionality $\fareconcompute$ takes the shares $\vs_j$ from parties $\Party_{m+i_{(2,j)}}$, with the $i_{(2,j)}\in\committee_2$ that was sampled by $\faelectshare$.

\begin{claim}\label{claim:hyb1_adaptive}
$\sset{\HYB^1_{\pi, \Adv, \Env}(\vx,\aux,\secParam)}_{\vx,\aux,\secParam} \compindist \sset{\HYB^2_{\pi, \Adv, \Env}(\vx,\aux,\secParam)}_{\vx,\aux,\secParam}$.
\end{claim}
\begin{proof}
The claim follows from the following observations. First, For a fixed $j\in[n']$, if both parties $\Party_{i_{(1,j)}}$ and $\Party_{m+i_{(2,j)}}$ are honest, then $\Party_{m+i_{(2,j)}}$ will provide $\fareconcompute$ with the correct signature of its index.

Second, with overwhelming probability, for at least $\ceil{n'/2}$ of the $j$'s it holds that both $\Party_{i_{(1,j)}}$ and $\Party_{m+i_{(2,j)}}$ are honest. This follows from the strong honest-majority assumption. In more detail, since the communication between honest parties is hidden from the adversary, he can identify that an honest party $\Party_i$ is in the committee $\committee_1$, \ie $i=i_{(1,j)}$, only if the corresponding party $\Party_{m+i_{(2,j)}}$ in $\committee_2$ is corrupted and receives a message from $\Party_{i_{(1,j)}}$ (and vice versa). Assume towards a contradiction that for $\ceil{n'/2}$ of the $j$'s at least one of $\Party_{i_{(1,j)}}$ or $\Party_{m+i_{(2,j)}}$ is corrupted. This means that the fraction of corrupted parties in one of the committees $\committee_1$ and $\committee_2$ is $\beta'\geq 1/4$. However, since by assumption there are at most $(1/8 - \delta)\cdot n = (1/4 - 2\delta)\cdot m$ corrupted parties at any point during the protocol's execution (and after its completion), it holds that in $\PS_1$ and in $\PS_2$ the fraction of corrupted parties is at most $(1/4 - 2\delta)$. Now, since $n'=\omega(\log{n})$ it follows from \cref{cor:elect} with overwhelming probability that if all the corruption took place at the onset of the protocol, then the fraction of corrupted parties in the committees is at most $(1/4 - \delta)$ (by setting $\epsilon=\delta$).

It is now remains to show that other than identifying ``matching parties'' in the committees (\ie the pair of parties $\Party_{i_{(1,j)}}$ and $\Party_{m+i_{(2,j)}}$) in Steps~\ref{step:a_first_bridge} and ~\ref{step:a_second_bridge}, the adversary does not gain any advantage in increasing the fraction of corrupted parties in the committees by dynamically corrupting parties. This follows since the communication is hidden from the adversary, and its view in the protocol (except for Steps~\ref{step:a_first_bridge} and~\ref{step:a_second_bridge}) is independent of the committees. Therefore, we derive a contradiction.

Finally, by the security of the signature scheme, the functionality $\fareconcompute$ will not receive input with signed index from parties outside of $\committee_2$.
\QED
\end{proof}

\paragraph{The game $\HYB^3_{\pi, \Adv, \Env}(\vx,\aux,\secParam)$.}
In this game, we modify $\HYB^2_{\pi, \Adv, \Env}(\vx,\aux,\secParam)$ as follows.
Instead of computing the function $f$ using the inputs of $\PS_2$ as provided to $\fareconcompute$, and the input values of $\PS_1$ as reconstructed from the shares provided by the parties in $\committee_2$, the computation of $f$ is performed as follows.
Let $\IS$ be the set of corrupted parties in Step~\ref{step:a_recon_compute} when calling $\fareconcompute$.
The input values $\sset{x'_i}_{i\in\IS\cap[m]}$ provided by the adversary to $\faelectshare$ and the input values $\sset{x'_i}_{i\in\IS\cap[m+1,2m]}$ provided by the adversary to $\fareconcompute$ are sent to an external party that computes $y=f(x'_1,\ldots,x'_n)$, where $x'_i=x_i$ for every $i\in[n]\setminus\IS$.

\begin{claim}\label{claim:hyb2_adaptive}
$\sset{\HYB^2_{\pi, \Adv, \Env}(\vx,\aux,\secParam)}_{\vx,\aux,\secParam} \statclose \sset{\HYB^3_{\pi, \Adv, \Env}(\vx,\aux,\secParam)}_{\vx,\aux,\secParam}$.
\end{claim}
\begin{proof}
The claim will follow as long as $\ssize{(\committee_1\cup\sset{m+i\mid i\in\committee_2})\cap\IS}\leq t'$.
This follows from the proof of \cref{claim:hyb1_adaptive}.
\QED
\end{proof}

\paragraph{The game $\HYB^4_{\pi, \Adv, \Env}(\vx,\aux,\secParam)$.}
In this game, we modify $\HYB^3_{\pi, \Adv, \Env}(\vx,\aux,\secParam)$ as follows.
instead of computing shares of the input values $\sset{x_i}_{i\in[m]}$, the functionality $\faelectshare$ computes shares of $0^\secParam$.

\begin{claim}\label{claim:hyb3_adaptive}
$\sset{\HYB^4_{\pi, \Adv, \Env}(\vx,\aux,\secParam)}_{\vx,\aux,\secParam} \statclose \sset{\HYB^3_{\pi, \Adv, \Env}(\vx,\aux,\secParam)}_{\vx,\aux,\secParam}$.
\end{claim}
\begin{proof}
The claim follows from the privacy of the ECSS scheme and since $\ssize{\committee_1\cap\IS}\leq t'$ and $\ssize{\sset{m+i\mid i\in\committee_2}\cap\IS}\leq t'$ at any point during the protocol, except for a negligible probability.
\QED
\end{proof}

The proof now follows since $\HYB^4_{\pi, \IS, \Adv(\aux)}(\vx,\secParam)$ exactly describes the simulation done by $\Sim$, and in particular, does not depend on the input values of honest parties. This concludes the proof of \cref{prop:security_adaptive}.
\QED
\end{proof}

\begin{proposition}\label{prop:not_expanding_adaptive}
The communication graph of $\protne$ is strongly not an expander.
\end{proposition}
\begin{proof}
The proof follows in a similar way as the proof of \cref{prop:not_expanding}.
\QED
\end{proof}

\begin{proposition}\label{prop:low_locality_adaptive}
Let $\rho_1,\rho_2,\rho_3$, and $\pi^{f_i\rightarrow \rho_i}$ be the protocols defined in \cref{lem:mpc_no_expander_adaptive}, and let $\ell_\rho=\ell_\rho(m)$ be the upper bound of the locality of $\rho_1,\rho_2,\rho_3$. Then $\pi^{f_i\rightarrow \rho_i}$ has locality $\ell=2\cdot\ell_\rho + 1$.
\end{proposition}
\begin{proof}
Every party in $\PS_1$ communicates with $\ell_\rho$ parties when executing $\rho_1$, and with at most another $\ell_\rho$ parties when executing $\rho_3$. In addition, every party in $\committee_1$ communicates with exactly one party in $\committee_2$. Similarly, every party in $\PS_2$ communicates with $\ell_\rho$ parties when executing $\rho_2$, and parties in $\committee_2$ communicates with exactly one party in $\committee_1$. It follows that maximal number of parties that a party communicates with during the protocol is $2\cdot\ell_\rho+1$.
\QED
\end{proof}

This concludes the proof of \cref{lem:mpc_no_expander_adaptive}.
\QED
\end{proof}

\fi

\section{Expansion is Necessary for Correct Computation}\label{sec:LB_Expander}

In this section, we show that in certain natural settings there exist functionalities such that the final communication graph of any MPC protocol that securely computes them \emph{must} be an expander. In fact, we prove a stronger statement, that removing a sublinear number of edges from such graphs will not disconnect them.
We consider the plain model, in which parties do not have any trusted setup assumptions,\footnote{The lower bound immediately extends to a setting where the parties have access to a common reference string; we consider the plain setting for simplicity.} a PPT adaptive adversary, and focus on \textsf{parallel multi-valued broadcast} (also known as interactive consistency~\cite{PSL80}), where every party has an input value, and all honest parties agree on a common output vector, such that if $\Party_i$ is honest then the \ith coordinate equals $\Party_i$'s input.
In particular, our proof does not rely on any privacy guarantees of the protocol, merely its correctness.

For simplicity, and without loss of generality, we assume the security parameter is the number of parties $n$.

\begin{definition}[parallel broadcast]
A protocol ensemble $\pi=\sset{\pi_n}_{n\in\N}$ is a \textsf{$t(n)$-resilient, parallel broadcast protocol with respect to input space $\sset{\zn}_{n\in\N}$}, if there exists a negligible function $\mu(n)$, such that for every $n\in\N$, every party $\Party_i$ in $\pi_n$ has input $x_i\in \zn$ and outputs a vector of $n$ values $\vy_i=(y^i_1,\ldots, y^i_n)$ such that the following is satisfied, except for probability $\mu(n)$. Facing any adaptive, malicious PPT adversary that dynamically corrupts and controls a subset of parties $\sset{\Party_j}_{j\in\IS}$, with $\IS\subseteq[n]$ of size $\ssize{\IS}\leq t(n)$, it holds that:
\begin{itemize}
    \item
    \textbf{Agreement.} There exists a vector $\vy=(y_1,\ldots, y_n)$ such that for every party $\Party_i$ that is honest at the conclusion of the protocol it holds that $\vy_i=\vy$.
    \item
    \textbf{Validity.} For every party $\Party_i$ that is honest at the conclusion of the protocol it holds that the \ith coordinate of the common output equals his input value, \ie $y_i=x_i$.
\end{itemize}
\end{definition}

Recall that a connected graph is \textsf{$k$-edge-connected} if it remains connected whenever fewer than $k$ edges are removed.
We are now ready to state the main result of this section.
We note that as opposed to \cref{sec:ne_mpc_adaptive}, where we considered adaptive corruptions in the hidden-channels model, this section considers the \emph{parallel secure message transmission (SMT)} model, formally defined in \cref{sec:lb_comm}, where the adversary is aware of communication between honest parties, but not of the message content.

\begin{theorem}[restating \cref{thm:intro:LB}]\label{thm:lower_bound}
Let $\beta>0$ be a fixed constant, let $t(n)=\beta\cdot n$, and let $\pi=\sset{\pi_n}_{n\in\N}$ be a $t(n)$-resilient, parallel broadcast protocol with respect to input space $\sset{\zn}_{n\in\N}$, in the parallel SMT hybrid model (in the computational setting, tolerating an adaptive, malicious PPT adversary).
Then, the communication graph of $\pi$ must be $\alpha(n)$-edge-connected, for every $\alpha(n)\in o(n)$.
\end{theorem}

From \cref{thm:lower_bound} and \cref{lem:sublinear_cut} (stating that if the communication graph of $\pi$ is strongly not an expander then there must exist a sublinear cut in the graph) we get the following corollary.
\begin{corollary}\label{cor:lower_bound}
Consider the setting of \cref{thm:lower_bound}.
If the communication graph of $\pi$ is strongly not an expander (as per \cref{def:protocol_non_expander}), then $\pi$ is not a $t(n)$-resilient parallel broadcast protocol.
\end{corollary}

The remainder of this section goes toward proving \cref{thm:lower_bound}.
We start by presenting the communication model in \cref{sec:lb_comm}. In \cref{sec:GT_lemma}, we prove a graph-theoretic theorem that will be used in the core of our proof and may be of independent interest. Then, in \cref{sec:lower_bound_proof} we present the proof of \cref{thm:lower_bound}.

\subsection{The Communication Model}\label{sec:lb_comm}
We consider secure communication channels, where the adversary can see that a message has been sent but not its content (in contrast to the hidden-communication model, used in \cref{sec:ne_mpc_adaptive}, where the communication between honest parties was hidden from the eyes of the adversary).
A standard assumption when considering adaptive corruptions is that in addition to being notified that an honest party sent a message, the adversary can corrupt the sender \emph{before} the receiver obtained the message, learn the content of the message, and replace it with another message of its choice that will be delivered to the receiver. Although the original modular composition framework~\cite{Canetti00} does not give the adversary such power, this ability became standard after the introduction of the \emph{secure message transmission (SMT)} functionality in the UC framework~\cite{Canetti01}. As we consider \emph{synchronous} protocols, we use the \textsf{parallel SMT} functionality that was formalized in~\cite{CCGZ16,CCGZ17}.\footnote{We note that by considering secure channels, that hide the content of the messages from the adversary, we obtain a stronger lower bound than, for example, authenticated channels.}

\begin{definition}[parallel SMT]
The \textsf{parallel secure message transmission functionality} $\fpsmt$ is a two-phase functionality.
For every $i,j\in[n]$, the functionality initializes a value $x^i_j$ to be the empty string $\emptystr$ (the value $x^i_j$ represents the message to be sent from $\Party_i$ to $\Party_j$).
\begin{itemize}
    \item
    \textbf{The input phase.}
    Every party $\Party_i$ sends a vector of $n$ messages $(v^i_1,\ldots,v^i_n)$. The functionality sets $x^i_j=v^i_j$, and provides the adversary with leakage information on the input values.
    As we consider rushing adversaries, who can determine the messages to be sent by the corrupted parties \emph{after} receiving the messages sent by the honest parties, the leakage function should leak the messages that are to be delivered from honest parties to corrupted parties.
    Therefore, the leakage function is
    \[
    \lpsmt\left((x^1_1,\ldots,x^1_n),\ldots,(x^n_1,\ldots,x^n_n)\right)
    =
    \left((y^1_1,\ldots,y^1_n),\ldots,(y^n_1,\ldots,y^n_n)\right),
    \]
    where $y^i_j=\ssize{x^i_j}$ in case $\Party_j$ is honest and $y^i_j=x^i_j$ in case $\Party_j$ is corrupted.

    We consider adaptive corruptions, and so, the adversary can corrupt an honest party during the input phase based on this leakage information, and send a new input on behalf of the corrupted party (note that the messages are not delivered yet to the honest parties).
    \item
    \textbf{The output phase.}
    In the second phase, the messages are delivered to the parties, \ie party $\Party_i$ receives the vector of messages $(x^1_i,\ldots,x^n_i)$.
\end{itemize}
\end{definition}
\noindent
In addition, we assume that the parties do not have any trusted-setup assumption.

\subsection{A Graph-Theoretic Theorem}\label{sec:GT_lemma}

Our lower-bound proof is based on the following graph-theoretic theorem, which we believe may be of independent interest.
We show that every graph in which every node has a linear degree, can be partitioned into a constant number of linear-size sets that are pairwise connected by sublinear many edges. These subsets are ``minimal cuts'' in the sense that every sublinear cut in the graph is a union of some of these subsets.
The proof of the theorem is deferred to \cref{sec:GT_lemma_cont}.

\begin{definition}[$(\alpha,d)$-partition]\label{def:partition}
Let $G=(V,E)$ be a graph of size $n$. An \textsf{$(\alpha,d)$-partition} of $G$ is a partition $\Gamma=(\island_1,\ldots,\island_\ell)$ of $V$ that satisfies the following properties:
\begin{enumerate}
    \item
    For every $i\in[\ell]$ it holds that $\ssize{\island_i}\geq d$.
    \item
    For every $i\neq j$, there are at most $\alpha$ edges between $\island_i$ and $\island_j$, \ie $\ssize{\edges_G(\island_i,\island_j)}\leq \alpha$.
    \item
    For every $S\subseteq V$ such that $\sset{S,\comp{S}}$ is an $\alpha$-cut, \ie $\ssize{\edges_{G}(S)}\leq\alpha$, it holds that there exists a subset $J\subsetneq[\ell]$ for which $S=\bigcup_{j\in J} U_j$ and $\comp{S}=\bigcup_{j\in [\ell]\setminus J} U_j$.
\end{enumerate}
\end{definition}

In \cref{lem:number_of_cuts} we first show that if every node in the graph has a linear degree $d(n)$, and $\alpha(n)$ is sublinear, then for sufficiently large $n$ there exists an $(\alpha(n),d(n))$-partition of the graph, and moreover, the partition can be found in polynomial time.

\begin{theorem}\label{lem:number_of_cuts}
Let $c>1$ be a constant integer, let $\alpha(n)\in o(n)$ be a fixed sublinear function in $n$, and let $\sset{G_n}_{n\in\N}$ be a family of graphs, where $G_n=([n],E_n)$ is defined on $n$ vertices, and every vertex of $G_n$ has degree at least $\frac{n}{c}-1$.
Then, for sufficiently large $n$ it holds that:
\begin{enumerate}
    \item
    There exists an $(\alpha(n),n/c)$-partition of $G_n$, denoted $\Gamma$; it holds that $\size{\Gamma} \leq c$.
    \item
    An $(\alpha(n),n/c)$-partition $\Gamma$ of $G_n$ can be found in (deterministic) polynomial time.
\end{enumerate}
\end{theorem}

\noindent
Note that if for every $n$ there exists an $\alpha(n)$-cut in $G_n$, then it immediately follows that $\size{\Gamma}>1$, \ie the partition is not the trivial partition of the set of all nodes.

\subsection{Proof of the Main Theorem (\texorpdfstring{\cref{thm:lower_bound}}{Lg})}\label{sec:lower_bound_proof}

\paragraph{High-level overview of the attack.}
For $n\in\N$, consider an execution of the alleged parallel broadcast protocol $\pi_n$ over uniformly distributed $n$-bit input values for the parties $(x_1,\ldots,x_n)\in_R (\zn)^n$.
We define two ensembles of adversarial strategies $\sset{\AdvneHonest}_{n\in\N}$ and $\sset{\AdvneCorrupt}_{n\in\N}$ (described in full in \cref{sec:LB_adv}).

The adversary $\AdvneCorrupt$ corrupts a random party $\Party_\is$, and simulates an honest execution on a random input $\tilde{x}_\is$ until $\Party_\is$ has degree $\beta n/4$. Next, $\AdvneCorrupt$ switches the internal state of $\Party_\is$ with a view that is consistent with an honest execution over the initial input $x_\is$, where all other parties have random inputs.
The adversary $\AdvneCorrupt$ continues by computing an $(\alpha(n),n/c)$-partition $\sset{U_1,\ldots,U_\ell}$ of the communication graph, (where $c$ is a constant depending only on $\beta$ -- this is possible due to \cref{lem:number_of_cuts}), and blocking every message that is sent between every pair of $U_i$'s. In \cref{lem:lb:entropy}, we show that there exist honest parties that at the conclusion of the protocol have received a bounded amount of information on the initial input value $x_\is$.

The second adversary, $\AdvneHonest$, is used for showing that under the previous attack, every honest party will eventually output the initial input value $x_\is$ (\cref{lem:lb:correct_output}). This is done by having $\AdvneHonest$ corrupt all the neighbors of $\Party_\is$, while keeping $\Party_\is$ honest, and simulate the previous attack to the remaining honest parties.

We show that there exist honest parties whose view is identically distributed under both attacks, and since they output $x_\is$ in the latter, they must also output $x_\is$ in the former.
By combining both of these lemmata, we then derive a contradiction.

\begin{proof}[Proof of \cref{thm:lower_bound}]
First, since we consider the plain model, without any trusted setup assumptions, known lower bounds~\cite{PSL80,LSP82,FLM86} state that parallel broadcast cannot be computed for $t(n)\geq n/3$; therefore, we can focus on $0<\beta<1/3$, \ie the case where $t(n)=\beta \cdot n <n/3$.

Assume toward a contradiction that $\pi$ is $t(n)$-resilient parallel broadcast protocol in the above setting, and that there exists a sublinear function $\alpha(n)\in o(n)$ such that the communication graph of $\pi$ is not $\alpha(n)$-edge-connected, \ie for sufficiently large $n$ there exists a cut $\sset{S_n,\comp{S}_n}$ of weight at most $\alpha(n)$.

\paragraph{Notations.}
We start by defining a few notations.
For a fixed $n$,\footnote{For clarity, we denote the random variables without the notation $n$.} consider the following independently distributed random variables
\[
\inputCoins=\left(X_1,\ldots,X_n,R_1,\ldots,R_n,\tilde{X}_1,\ldots,\tilde{X}_n,\tilde{R}_1,\ldots,\tilde{R}_n,\Is\right),
\]
where for every $i\in[n]$, each $X_i$ and $\tilde{X}_i$ take values uniformly at random in the input space~$\zn$, each $R_i$ and $\tilde{R}_i$ take values uniformly at random in $\zs$, and $\Is$ takes values uniformly at random in $[n]$.
During the proof, $(X_i,R_i)$ represent the pair of input and private randomness of party $\Party_i$, whereas $(\tilde{X}_1,\ldots,\tilde{X}_n,\tilde{R}_1,\ldots,\tilde{R}_n,\Is)$ correspond to the random coins of the adversary (used in simulating the two executions toward the honest parties).
Unless stated otherwise, all probabilities are taken over these random variables.

Let $\redExec$ be a random variable defined as
\[
\redExec=\left(X_{-\Is},\tilde{X}_\Is,R_{-\Is},\tilde{R}_\Is\right).
\]
That is, $\redExec$ contains $X_i$ and $R_i$ for $i\in[n]\setminus\sset{\Is}$, along with $\tilde{X}_\Is$ and $\tilde{R}_\Is$.
We denote by the ``red execution'' an \emph{honest} protocol execution when the inputs and private randomness of the parties are $(X_{-\Is},\tilde{X}_\Is,R_{-\Is},\tilde{R}_\Is)$.
We denote by the ``blue execution'' an \emph{honest} protocol execution when the inputs and private randomness of the parties are $(\tilde{X}_{-\Is},X_\Is,\tilde{R}_{-\Is},R_\Is)$.
Note that such a sample fully determines the view and transcript of all parties in an \emph{honest} simulated execution of $\pi_n$.

Let $\finalCut$ be a random variable defined over $2^{[n]}\cup\sset{\bot}$.
The distribution of $\finalCut$ is defined by running protocol $\pi_n$ until its conclusion with adversary $\AdvneCorrupt$ (defined in \cref{sec:LB_adv}) on inputs and coins sampled according to \inputCoins. If at the conclusion of the protocol there is no $\alpha(n)$-cut in the graph, then set the value of $\finalCut$ to be $\bot$; otherwise, set the value to be the identity of the smallest $\alpha(n)$-cut $\sset{S,\comp{S}}$ in the communication graph according to some canonical ordering on the $\alpha(n)$-cuts.
We will prove that conditioned on the value of $\redExec$, the $\finalCut$ can only take one of a constant number of values depending only on $\beta$ (and not on $n$).

Let $\E_1$ denote the event that $\Party_\Is$ is the last among all the parties to reach degree $\beta n/4$ in both the red and the blue honest executions of the protocol. More precisely, the event that $\Party_\Is$ reaches degree $\beta n/4$ in both executions, and if it has reached this degree in round $\round$ in the red (blue) execution, then all parties in the red (blue) execution have degree at least $\beta n/4$ in round $\round$.

Let $\E_2$ denote the event that the degree of $\Party_\Is$ reaches $\beta n/4$ in the red execution before, or at the same round as, in the blue execution.
Note that $\E_1$ and $\E_2$ are events with respect to two honest executions of the protocol (the red execution and the blue execution) that are defined according to $\inputCoins$. In both adversarial strategies that will be used in the proof, the corrupted parties will operate in a way that indeed induces the red and blue executions, respectively, and so, the events $\E_1$ and $\E_2$ are well defined in an execution of the protocol with those adversarial strategies.

In \cref{sec:LB_adv}, we formally describe two adversarial strategies, $\AdvneHonest$ and $\AdvneCorrupt$ (see \cref{fig:LBattack_honest_one,fig:LBattack_honest_two} and \cref{fig:LBattack_corrupt_one,fig:LBattack_corrupt_two}, respectively).
We denote by $Y^\corrupt_\Is$, respectively $Y^\honest_\Is$, the random variable that corresponds to the $\Is$'th coordinate of the common output of honest parties, when running the protocol over random inputs with adversarial strategy $\AdvneCorrupt$, respectively $\AdvneHonest$.

\begin{nfbox}{(\phaseI of Adversary $\AdvneHonest$)}{fig:LBattack_honest_one}

\begin{center}
    \textbf{Adversary $\AdvneHonest$}
\end{center}
The adversary $\AdvneHonest$ attacks an $n$-party protocol $\pi_n=(\Party_1,\ldots,\Party_n)$. The adversary is parametrized by $\beta_1=\beta/2$, by $c=\ceil{4/\beta}$, and by a sublinear function $\alpha(n)\in o(n)$, and proceeds as follows:
\begin{description}
    \item[\phaseI:] (ensuring high degree in the ``red execution'')
    \begin{enumerate}
        \item\label{step:index}
        Choose uniformly at random $\is \gets [n]$.
        \item\label{step:virtual_protocol}
        Emulate in its head:
        \begin{itemize}
            \item
            A virtual (blue) execution of $\pi_n$ with virtual parties $(\DQParty_1,\ldots,\DQParty_n)$, where every $\DQParty_i$, for $i\in[n]\setminus\sset{\is}$, is initialized with a uniformly distributed input $\tilde{x}_i\in_R \zn$ and random coins $\tilde{r}_i$. (Party $\DQParty_\is$ is simulated based on the actions of party $\Party_\is$.)
            \item
            A virtual party $\DParty_\is$ with a uniformly distributed input $\tilde{x}_\is\in_R \zn$ and random coins~$\tilde{r}_\is$. (Party $\DParty_\is$ is used to simulate a (red) execution with real $\Party_i$'s.)
        \end{itemize}
        \item\label{step:corrupt_one}
        In every round $\round$, receive the leakage from $\fpsmt$ containing the messages for corrupted parties, and which honest party sends a message to another honest party. Next:
        \begin{enumerate}
            \item
            Proceed with round $\round$ in the \emph{red} execution as follows. For every $j\in[n]\setminus\sset{\is}$ (in lexicographic order) check the following:
            \begin{enumerate}
                \item
                If party $\Party_j$ sends a message to $\Party_\is$ then corrupt $\Party_j$, learn the message $\mu$, and discard the message. Next, simulate party $\DParty_\is$ as receiving message $\mu$ from~$\Party_j$.
                \item
                If virtual party $\DParty_\is$ should send a message $\mu$ to party $\Party_j$, then corrupt $\Party_j$ and instruct him to play as an honest party that received the message $\mu$ from $\Party_\is$.
                \item
                Let $\Gred(\round,j)$ be the communication graph of the red execution at round $\round$, except for messages from $\DParty_\is$ to $\Party_{j'}$, and from $\DParty_{j'}$ to $\Party_\is$, for $j'>j$.
                If
                \[
                \deg_{\Gred(\round,j)}(\is) \geq(\beta_1/2)\cdot n,
                \]
                then let $\GphaseIIH=\Gred(\round,j)$, let $\roundphaseIIH=\round$, and proceed to \phaseII.
            \end{enumerate}
            \item
            Proceed with round $\round$ in the \emph{blue} execution as follows. For every $j\in[n]\setminus\sset{\is}$ (in lexicographic order) check the following:
            \begin{enumerate}
                \item
                If $\Party_\is$ sends a message to party $\Party_j$, corrupt $\Party_j$, learn the message $\mu$, and instruct $\Party_j$ to play as an honest party that does not receive messages from $\Party_\is$. In addition, simulate virtual party $\DQParty_\is$ sending the message $\mu$ to party $\DQParty_j$.
                \item
                If virtual party $\DQParty_j$ sends a message $\mu$ to virtual party $\DQParty_\is$ then corrupt $\Party_j$ and send the message $\mu$ to party $\Party_\is$.
                \item
                Let $\Gblue(\round,j)$ be the communication graph of the blue execution at round $\round$, except for messages from $\DQParty_\is$ to $\DQParty_{j'}$, and from $\DQParty_{j'}$ to $\DQParty_\is$, for $j'>j$.
                If
                \[
                \deg_{\Gblue(\round,j)}(\is) \geq(\beta_1/2)\cdot n,
                \]
                then output $\bot$ and halt (the attack fails).
            \end{enumerate}

            \item\label{step:lb:finish_phase_one}
            In case the protocol completes before \phaseII, then halt.
        \end{enumerate}
    \end{enumerate}

\end{description}
\end{nfbox}

\newpage
\paragraph{Proof structure.}
Our proof follows from two main steps.
In \cref{lem:lb:entropy}, stated and proven in \cref{sec:LB_lemma_entropy}, we show that in an execution of $\pi_n$ on random inputs with adversary $\AdvneCorrupt$, it holds that (1) $\pr{\E_1 \cap \E_2} \geq 1/2n^2 - \negl(n)$, and that (2) conditioned on the event $\E_1 \cap \E_2$, there exists an honest party $\Party_\js$ such that $X_\Is$, conditioned on $\E_1 \cap \E_2$ and on the view of $\Party_\js$ at the conclusion of the protocol, still retains at least $n/2$ bits of entropy. This means, in particular, that $\Party_\js$ will output the value $X_\Is$ only with negligible probability. Hence, by \emph{agreement}, the probability for any of the honest parties to output $X_\Is$ in an execution with $\AdvneCorrupt$ is negligible. In particular,
\[
\pr{Y^\corrupt_\Is=X_\Is\mid \E_1 \cap \E_2}=\negl(n).
\]

In \cref{lem:lb:correct_output}, stated and proven in \cref{sec:LB_lemma_correct_output}, we show that in an execution of $\pi_n$ on random inputs with adversary $\AdvneHonest$, it holds that (1) with overwhelming probability all honest parties output $X_\Is$ (this holds by correctness, since $\Party_\Is$ remains honest), \ie
\[
\pr{Y^\honest_\Is=X_\Is}\geq 1-\negl(n),
\]
and that (2) conditioned on the event $\E_1 \cap \E_2$, there exists an honest party whose view is identically distributed as in an execution with $\AdvneCorrupt$; therefore,
\[
\pr{Y^\corrupt_\Is=Y^\honest_\Is \mid \E_1 \cap \E_2}\geq 1-\negl(n).
\]

From the combination of the two lemmata, we derive a contradiction.
\QED
\end{proof}

\subsubsection{Defining Adversarial Strategies\SUBSUBSEC}\label{sec:LB_adv}

As discussed above, the main idea behind the proof is to construct two dual adversarial strategies that will show that on the one hand, the output of all honest parties must contain the initial value of a randomly chosen corrupted party, and on the other hand, there exist parties that only receive a bounded amount of information on this value during the course of the protocol.

We use the following notation for defining the adversarial strategies. Virtual parties that only exist in the head of the adversary are denoted with ``tilde''. In particular, for a random $\is\in[n]$, we denote by $\DParty_\is$ a virtual party that emulates the role of $\Party_\is$ playing with the real parties using a random input in the so-called ``red execution,'' and by $\sset{\DQParty_i}_{i\neq\is}$ virtual parties that emulate an execution over random inputs toward $\Party_\is$ in the so-called ``blue execution.''\footnote{Following the \emph{red pill blue pill} paradigm, in the adversarial strategy $\AdvneHonest$, the chosen party $\Party_\is$ is participating (without knowing it) in the \emph{blue} execution, which is a fake execution that does not happen in the real world. The real honest parties participate in the \emph{red} execution, where the adversary simulates $\Party_\is$ by running a virtual party.}

\paragraph{The adversary $\AdvneHonest$.}
At a high level, the adversary $\AdvneHonest$ (see \cref{fig:LBattack_honest_one,fig:LBattack_honest_two}) chooses a random $\is\in[n]$ and isolates the honest party $\Party_\is$. The adversary $\AdvneHonest$ consists of three phases. In \phaseI, $\AdvneHonest$ induces two honestly distributed executions.
\begin{itemize}
    \item
    The first (red) execution is set by simulating an honest execution of a virtual party $\DParty_\is$ over a random input $\tilde{x}_\is$ toward all other parties. The adversary corrupts any party that sends a message to $\Party_\is$, blocks its message, and simulates the virtual party $\DParty_\is$ receiving this message. Whenever the virtual party $\DParty_\is$ should send a message to some $\Party_j$, the adversary corrupts party $\Party_j$, and instructs him to proceed as if he received the intended message from $\DParty_\is$.
\end{itemize}

\begin{nfbox}{(Phases II and III of Adversary $\AdvneHonest$)}{fig:LBattack_honest_two}

\begin{center}
    \textbf{Adversary $\AdvneHonest$}
\end{center}
\begin{description}

    \item[\phaseII:](ensure that $\Party_\is$ remains isolated until his degree in the \emph{real} execution is high (\ie not in the blue execution))
    \begin{enumerate}
        \item\label{step:lb:honest_phaseI_check_degree}
        If for some $\Party_i$ with $i\neq \is$, it holds $\deg_{\GphaseIIH}(i)<(\beta_1/2)\cdot n$, output $\bot$ (the attack fails).
        \item
        Identify an $(\alpha, n/c)$-partition of $\GphaseIIH\setminus\sset{\is}$, denoted $\Gamma_1=\sset{\island_1, \ldots, \island_\ell}$ (see \cref{def:partition}).
        \item
        For every round $\round\geq \roundphaseIIH$, proceed with round $\round$ in the real execution as follows.
        \begin{enumerate}
            \item
            For every $j\in[n]\setminus\sset{\is}$ (in lexicographic order) check the following:
            \begin{enumerate}
                \item
                Let $\Greal(\round,j)$ be the communication graph of the real execution at round $\round$, except for messages from $\Party_\is$ to $\Party_{j'}$, and from $\Party_{j'}$ to $\Party_\is$, for $j'>j$.
                \item
                If $\Party_\is$ talked to $\Party_j$ in round $\round$, and $\deg_{\Greal(\round,j)}(\is)<(\beta_1/2)\cdot n$, then corrupt $\Party_j$ and instruct to play honestly as a party that does not receive messages from $\Party_\is$.
                \item
                If $\deg_{\Greal(\round,j)}(\is)\geq (\beta_1/2)\cdot n$, then let $\GphaseIIIH=\Greal(\round,j)$, let $\roundphaseIIIH=\round$, and proceed to \phaseIII.
            \end{enumerate}
            \item
            Let $\Greal(\round)$ be the communication graph of the real execution at round $\round$.
            \item
            For every $i,j\neq \is$, such that $\Party_i$ has sent a message to $\Party_j$ in round $\round$:
            \begin{enumerate}
                \item
                If $i\in \island_k$ and $j\in \island_{k'}$, for $k\neq {k'}$, and $\ssize{\edges_{\Greal(\round)}(\island_k,\island_{k'})} \leq \alpha(n)$, then corrupt $\Party_j$, and instruct $\Party_j$ to play as an honest party that does not send/receive messages to/from parties outside of $\island_k$ (from round $\round$ and onwards).
            \end{enumerate}

            \item\label{step:lb:honest_finish_phase_two}
            In case the protocol completes before \phaseIII, then halt.
        \end{enumerate}
    \end{enumerate}

    \item[\phaseIII:](isolate low-weight cuts until the protocol completes)
    \begin{enumerate}
        \item
        To add $\is$ to one of the sets in the partition, consider the minimum index of a set $\island_k$ for which $\Party_\is$ has more than $\alpha(n)$ neighbors as
        \[
        k_{\min}=\min\set{k : \size{\edges_{\GphaseIIIH}\left(\set{\is},\island_k\right)}>\alpha(n)}.
        \]
        Set $\Visland_{k_{\min}}=\island_{k_{\min}}\bigcup\sset{\is}$ and $\Visland_k=\island_k$ for all $k\neq k_{\min}$.
        Denote $\Gamma_2=\sset{\Visland_1, \ldots, \Visland_\ell}$.\footnote{Note that $\Gamma_2$ \emph{may not} be an $(\alpha,d)$-partition of the communication graph.}
        \item
        For every round $\round\geq \roundphaseIIIH$, proceed with round $\round$ in the real execution as follows.
        \begin{enumerate}
            \item
            Let $\Greal(\round)$ be the communication graph of the real execution at round $\round$.
            \item
            For every $i$ and every $j\neq \is$, such that $\Party_i$ has sent a message to $\Party_j$ in round $\round$:
            \begin{enumerate}
                \item
                If $i\in \Visland_k$ and $j\in \Visland_{k'}$, for $k\neq {k'}$, and $\ssize{\edges_{\Greal(\round)}(\Visland_k,\Visland_{k'})} \leq \alpha(n)$, then corrupt $\Party_j$, and instruct $\Party_j$ to play as an honest party that does not send/receive messages to/from parties outside of $\Visland_k$ (from round $\round$ and onwards).
            \end{enumerate}
        \end{enumerate}
    \end{enumerate}

\end{description}
\end{nfbox}

\begin{itemize}
    \item
    For the second (blue) execution, $\AdvneHonest$ emulates a virtual execution with virtual parties $(\DQParty_1,\ldots,\DQParty_n)\setminus\sset{\DQParty_\is}$ on random inputs toward the honest party $\Party_\is$. To do so, whenever $\Party_\is$ sends a message to $\Party_j$ in the real execution, the adversary corrupts $\Party_j$, instructing him to ignore this message, and simulates this message from $\Party_\is$ to $\DQParty_j$ in the virtual execution (that is running in the head of the adversary). Whenever a party $\DQParty_j$ sends a message to $\Party_\is$ in the virtual execution, the adversary corrupts the real party $\Party_j$ and instructs him to send this message to $\Party_\is$ in the real execution.
\end{itemize}

\noindent
Recall that according to \cref{def:comm_graph}, edges in which an honest party receives messages will be considered in the final communication graph if the receiver actually processed the messages. In particular, this means that the messages that are blocked during the red execution are not processed by $\Party_\is$ and will not be considered in the final graph, whereas the messages sent by corrupted parties to $\Party_\is$ in the blue execution will be considered as long as $\Party_\is$ will process them.

\begin{nfbox}{(\phaseI of Adversary $\AdvneCorrupt$)}{fig:LBattack_corrupt_one}
\begin{center}
    \textbf{Adversary $\AdvneCorrupt$}
\end{center}
The adversary $\AdvneCorrupt$ attacks an $n$-party protocol $\pi_n=(\Party_1,\ldots,\Party_n)$. The adversary is parametrized by $\beta_1=\beta/2$, by $c=\ceil{4/\beta}$, and by a sublinear function $\alpha(n) \in o(n)$, and proceeds as follows:
\begin{description}
    \item[\phaseI:] (ensuring high degree in the ``red execution'')
    \begin{enumerate}
        \item\label{step:lb:index}
        Choose uniformly at random $\is \gets [n]$ and corrupt party $\Party_\is$.
        \item\label{step:lb:virtual_protocol}
        Emulate in its head:
        \begin{enumerate}
            \item
            A virtual (blue) execution of $\pi_n$ with virtual parties $(\DQParty_1,\ldots,\DQParty_n)$, where every virtual party $\DQParty_i$, for $i\in[n]\setminus\sset{\is}$, is initialized with a uniformly distributed input $\tilde{x}_i\in_R \zn$ and random coins $\tilde{r}_i$. Virtual party $\DQParty_\is$ is initialized with input $x_\is$ and random coins~$r_\is$.
            \item
            A virtual party $\DParty_\is$ with a uniformly distributed input $\tilde{x}_\is\in_R \zn$ and random coins~$\tilde{r}_\is$.             (Party $\DParty_\is$ is used to simulate a (red) execution with real $\Party_i$'s.)
        \end{enumerate}
        \item\label{step:lb:corrupt_one}
        In every round $\round$, receive the leakage from $\fpsmt$ containing the messages for corrupted parties, and which honest party sends a message to another honest party. Next:
        \begin{enumerate}
            \item
            Proceed with round $\round$ in the \emph{red} execution $(\Party_1,\ldots,\DParty_\is,\ldots,\Party_n)$ as follows. For every $j\in[n]\setminus\sset{\is}$ (in lexicographic order) check the following:
            \begin{enumerate}
                \item
                If party $\Party_j$ sends a message $\mu$ to $\Party_\is$, then simulate $\DParty_\is$ as receiving message $\mu$ from~$\Party_j$.
                \item
                If virtual party $\DParty_\is$ generates a message $\mu$ for $\Party_j$, send $\mu$ to $\Party_j$ on behalf of $\Party_\is$.
                \item
                Let $\Gred(\round,j)$ be the communication graph of the red execution at round $\round$, except for message from $\DParty_\is$ to $\Party_{j'}$, and from $\DParty_{j'}$ to $\Party_\is$, for $j'>j$.
                If
                \[
                \deg_{\Gred(\round,j)}(\is) \geq(\beta_1/2)\cdot n,
                \]
                then let $\GphaseIIC=\Gred(\round,j)$, let $\roundphaseIIC=\round$, and proceed to \phaseII.
            \end{enumerate}

            \item\label{step:lb:corrupt_blue_execution}
            Proceed with round $\round$ in the \emph{blue} execution as follows. Generate all messages for round $\round$, and for every $j\in[n]\setminus\sset{\is}$ check the following:
            \begin{enumerate}
                \item
                Let $\Gblue(\round,j)$ be the communication graph of the blue execution at round $\round$, except for messages from $\DQParty_\is$ to $\DQParty_{j'}$, and from $\DQParty_{j'}$ to $\DQParty_\is$, for $j'>j$.
                If
                \[
                \deg_{\Gblue(\round,j)}(\is) \geq(\beta_1/2)\cdot n,
                \]
                then output $\bot$ and halt (the attack fails).
            \end{enumerate}

            \item\label{step:lb:corrupt_finish_phase_one}
            In case the protocol completes before \phaseII, then halt.
        \end{enumerate}
    \end{enumerate}
\end{description}
\end{nfbox}

\begin{nfbox}{(Phases II and III of Adversary $\AdvneCorrupt$)}{fig:LBattack_corrupt_two}
\begin{center}
    \textbf{Adversary $\AdvneCorrupt$}
\end{center}
\begin{description}

    \item[\phaseII:](hold back any information about $x_\is$ until the degree of $\Party_\is$ in the \emph{real} execution is high)
    \begin{enumerate}
        \item\label{step:lb:corrupt_phaseI_check_degree}
        If for some party $\Party_i$, for $i\neq \is$, it holds that $\deg_{\GphaseIIC}(i)<(\beta_1/2)\cdot n$, then output $\bot$ and halt (the attack fails).
        \item
        Identify an $(\alpha,(\beta_1/2)\cdot n)$-partition of $\GphaseIIC\setminus\sset{\is}$, denoted $\Gamma_1=\sset{\island_1, \ldots, \island_\ell}$.
        \item
        For every round $\round\geq \roundphaseIIC$, proceed with round $\round$ in the real execution as follows.
        \begin{enumerate}
            \item
            For every $j\in[n]\setminus\sset{\is}$ (in lexicographic order) check the following:
            \begin{enumerate}
                \item
                Let $\Greal(\round,j)$ be the communication graph of the real execution at round $\round$, except for message from $\Party_\is$ to $\Party_{j'}$, and from $\Party_{j'}$ to $\Party_\is$, for $j'>j$.
                \item
                If $\Party_j$ sends a message $\mu$ to $\Party_\is$ in round $\round$, and $\deg_{\Greal(\round,j)}(\is)<(\beta_1/2)\cdot n$, then simulates virtual party $\DQParty_\is$ receiving the message $\mu$ from $\DQParty_j$.
                \item
                If $\deg_{\Greal(\round,j)}(\is)\geq (\beta_1/2)\cdot n$, then let $\GphaseIIIC=\Greal(\round,j)$, let $\roundphaseIIIC=\round$, and proceed to \phaseIII.
            \end{enumerate}
            \item
            Let $\Greal(\round)$ be the communication graph of the real execution at round $\round$.
            \item
            For every $i,j\neq \is$, such that $\Party_i$ has sent a message to $\Party_j$ in round $\round$:
            \begin{enumerate}
                \item
                If $i\in \island_k$ and $j\in \island_{k'}$, for $k\neq {k'}$, and $\ssize{\edges_{\Greal(\round)}(\island_k,\island_{k'})} \leq \alpha(n)$, then corrupt $\Party_j$, and instruct $\Party_j$ to play as an honest party that does not send/receive messages to/from parties outside of $\island_k$ (from round $\round$ and onwards).
            \end{enumerate}

            \item\label{step:lb:corrupt_finish_phase_two}
            In case the protocol completes before \phaseIII, then halt.
        \end{enumerate}
    \end{enumerate}

    \item[\phaseIII:](isolate low-weight cuts until the protocol completes)
    \begin{enumerate}
        \item
        To add $\is$ to one of the sets in the partition, consider the minimum index of a set $\island_k$ for which $\Party_\is$ has more than $\alpha(n)$ neighbors as
        \[
        k_{\min}=\min\set{k : \size{\edges_{\GphaseIIIH}\left(\set{\is},\island_k\right)}>\alpha(n)}.
        \]
        Set $\Visland_{k_{\min}}=\island_{k_{\min}}\bigcup\sset{\is}$ and $\Visland_k=\island_k$ for all $k\neq k_{\min}$.
        Denote $\Gamma_2=\sset{\Visland_1, \ldots, \Visland_\ell}$.\footnote{Note that $\Gamma_2$ \emph{may not} be an $(\alpha,d)$-partition of the communication graph.}
        \item
        For every round $\round\geq \roundphaseIIIC$, proceed with round $\round$ in the real execution as follows.
        \begin{enumerate}
            \item
            Let $\Greal(\round)$ be the communication graph of the real execution at round $\round$.
            \item
            For every $i$ and every $j\neq \is$, such that $\Party_i$ has sent a message to $\Party_j$ in round $\round$:
            \begin{enumerate}
                \item
                If $i\in \Visland_k$ and $j\in \Visland_{k'}$, for $k\neq {k'}$, and $\ssize{\edges_{\Greal(\round)}(\Visland_k,\Visland_{k'})} \leq \alpha(n)$, then corrupt $\Party_j$, and instruct $\Party_j$ to play as an honest party that does not send/receive messages to/from parties outside of $\Visland_k$ (from round $\round$ and onwards).
            \end{enumerate}
        \end{enumerate}
    \end{enumerate}
\end{description}
\end{nfbox}

\phaseII begins when the degree of the virtual party $\DParty_\is$ in the \emph{red} execution is at least $\beta n/4$; if $\Party_\is$ reaches this threshold faster in the \emph{blue} execution, the attack fails.
\phaseIII begins when the degree of $\Party_\is$ in the \emph{real} execution is at least $\beta n/4$.

Ideally, \phaseI will continue until all parties in the real execution have a linear degree, and before the adversary will use half of his ``corruption budget'', \ie $\beta n/2$.
This would be the case if we were to consider a single honest execution of the protocol, since we show that there always exists a party that will be the last to reach the linear-degree threshold with a noticeable probability.
However, as the attack induces two \emph{independent} executions, in which the degree of the parties can grow at different rates, care must be taken. We ensure that even though $\Party_\is$ runs in the blue execution, by the time $\Party_\is$ will reach the threshold, all other parties (that participate in the red execution) will already have reached the threshold, and can be partitioned into ``minimal'' $\alpha(n)$-cuts, as follows.

The adversary allocates $\beta n/4$ corruptions for the red execution and $\beta n/4$ corruptions for the blue execution. We show that with a noticeable probability, once $\DParty_\is$ has degree $\beta n/4$ in the red execution, all other parties in the red execution also have high degree. Consider the communication graph of the red execution without the virtual party $\DParty_\is$ (\ie after removing the node $\is$ and its edges); by \cref{lem:number_of_cuts} there exists an $(\alpha(n),\beta n/4)$ partition of this graph into a constant number of linear-size subsets that are connected with sublinear many edges, denoted $\Gamma_1=\sset{\island_1,\ldots,\island_\ell}$ (in particular, this partition is independent of $x_\is$).
In \phaseII, the adversary continues blocking outgoing messages from $\Party_\is$ toward the real honest parties, until the degree of $\Party_\is$ in the real execution is $\beta n/4$. In addition, $\AdvneHonest$ blocks any message that is sent between two subsets in the partition, by corrupting the recipient and instructing him to ignore messages from outside of his subset.

In \phaseIII, which begins when $\Party_\is$ has high degree in the real execution, the adversary adds $\Party_\is$ to one of the subsets in the partition, in which $\Party_\is$ has many neighbors, and continues to block messages between different subsets in the partition until the conclusion of the protocol.

We note that special care must be taken in the transition between the phases, since such a transition can happen in a middle of a round, after processing some of the messages, but not all. Indeed, if the transition to the next phase will happen at the end of the round, the adversary may need to corrupt too many parties. For this reason, in Phases I and II, we analyze the messages to and from $\Party_\is$ one by one, and check whether the threshold has been met after each such message.

\paragraph{The adversary $\AdvneCorrupt$.}

The adversary $\AdvneCorrupt$ (see \cref{fig:LBattack_corrupt_one,fig:LBattack_corrupt_two}) corrupts the randomly chosen party $\Party_\is$, and emulates the operations of an honest $\Party_\is$ that is being attacked by $\AdvneHonest$.

In \phaseI, the adversary $\AdvneCorrupt$ induces two honestly distributed executions, by simulating an honest execution of a virtual party $\DParty_\is$ over a random input $\tilde{x}_\is$ toward all other honest parties (the red execution), and furthermore, runs in its mind a virtual execution over the initial input $x_\is$ and random inputs $\tilde{x}_i$ for $i\neq \is$ (the blue execution). This phase continues until $\DParty_\is$ has degree $\beta n/4$ in the red execution (no parties other than $\Party_\is$ are being corrupted). If all other parties in the red execution have high degree, then the adversary finds the partition of the red graph as in the previous attack (the partition is guaranteed by \cref{lem:number_of_cuts}). \radded{Note that only $\Party_\is$ is corrupted, hence all messages that are sent by other parties will be considered in the final communication graph, as well as messages sent by $\Party_\is$ that are processed by the receivers.}

In \phaseII, the adversary continues simulating the corrupted $\Party_\is$ toward the real honest parties until the degree of $\Party_\is$ in the real execution is $\beta n/4$; however, his communication is based on the view in the blue execution at the end of \phaseI (this is no longer an honest-looking execution). During this phase, $\AdvneCorrupt$ blocks any message that is sent between two subsets in the partition.

In \phaseIII, that begins when $\Party_\is$ has high degree (in the real execution), $\AdvneCorrupt$ adds $\Party_\is$ to one of the subsets in the partition, in which $\Party_\is$ has many neighbors, and continues to block messages between different subsets in the partition until the conclusion of the protocol.

\subsubsection{Proving High Entropy of \texorpdfstring{$X_\Is$}{Lg}\SUBSUBSEC}\label{sec:LB_lemma_entropy}

We now proceed to prove the first of the two core lemmata.
\begin{lemma}\label{lem:lb:entropy}
Consider an execution of $\pi_n$ on random inputs $(X_1,\ldots,X_n)$ for the parties with adversary $\AdvneCorrupt$, and the events $\E_1$ and $\E_2$ as defined in \cref{sec:lower_bound_proof}. Then, it holds that:
\begin{enumerate}
    \item
    $\pr{\E_1 \cap \E_2}\geq 1/{2n^2} -\negl(n)$.
    \item
    Conditioned on the event $\E_1 \cap \E_2$, there exists an honest party $\Party_\js$ such that
    \[
    H(X_\Is \mid \E_1 \cap \E_2, \VIEW^\corrupt_\js)\geq n/2,
    \]
    where $\VIEW^\corrupt_\js$ is the random variable representing the view of $\Party_\js$ at the end of the protocol.
\end{enumerate}
\end{lemma}

\begin{proof}
We start by showing that for a randomly chosen $\Is\in[n]$, if party $\Party_\Is$ does not send messages to sufficiently many parties, the adversary $\AdvneCorrupt$ can violate \emph{validity} of $\pi_n$.

\begin{claim}\label{claim:LB_few_neighbors}
Let $\GphaseIIC$ denote the random variable representing the graph of the red execution $(\Party_1,\ldots,\DParty_\Is,\ldots,\Party_n)$ at the conclusion of \phaseI with adversary $\AdvneCorrupt$. Then
\[
\pr{\deg_{\GphaseIIC}(\Is)<(\beta_1/2) \cdot n}\leq \negl(n).
\]
\end{claim}

\begin{proof}
Consider an execution of $\pi_n$ with adversary $\AdvneCorrupt$ on random inputs $(X_1,\ldots,X_n)$, with random coins $(R_1,\ldots,R_n)$ for the parties and $(\tilde{X}_1,\ldots,\tilde{X}_n,\tilde{R}_1,\ldots,\tilde{R}_n,\Is)$ for the adversary.
Denote $\pr{\deg_{\GphaseIIC}(\Is)<(\beta_1/2) \cdot n} = \epsilon(n)$.

In this case, the view of all honest parties (\ie all parties but $\Party_\Is$) is identically distributed as their view in an honest execution of the protocol $\pi_n$ on input $(X_1,\ldots,X_{\Is-1},\tilde{X}_\Is,X_{\Is+1},\ldots,X_n)$, and in particular, all honest parties output $Y^\corrupt_\Is=\tilde{X}_\Is$ as the $\Is$'th coordinate of the common output, except for a negligible probability.

This is not sufficient for contradicting validity, since $\Party_\Is$ is corrupted, hence, all that is required is \emph{agreement} on the coordinate $Y^\corrupt_\Is$. However, consider the adversary $\AdvneHonest$ (\cref{fig:LBattack_honest_one,fig:LBattack_honest_two}) that instead of corrupting $\Party_\Is$, isolates him by corrupting all his neighbors.\footnote{For now, we are only interested in \phaseI of $\AdvneHonest$.}

Because the protocol is defined in the plain model, and the parties do not share correlated randomness (such as PKI), the adversary $\AdvneHonest$ indeed manages to isolate party $\Party_\Is$ such that:
\begin{enumerate}
    \item
    The view of all honest parties, except for party $\Party_\Is$, is distributed identically as in an honest execution of $\pi_n$ on inputs $(X_1,\ldots,X_{\Is-1},\tilde{X}_\Is,X_{\Is+1},\ldots,X_n)$ and random coins $(R_1,\ldots,R_{\Is-1},\tilde{R}_\Is,R_{\Is+1},\ldots,R_n)$.
    Denote the communication graph of this distribution by $\GredH$.
    \item
    The view of party $\Party_\Is$ is distributed identically as in an honest execution $\pi_n$ on inputs $(\tilde{X}_1,\ldots,\tilde{X}_{\Is-1},X_\Is,\tilde{X}_{\Is+1},\ldots,\tilde{X}_n)$ and random coins $(\tilde{R}_1,\ldots,\tilde{R}_{\Is-1},R_\Is,\tilde{R}_{\Is+1},\ldots,\tilde{R}_n)$.
    Denote the communication graph of this distribution by $\GblueH$.
\end{enumerate}

\noindent
Since both executions are distributed like honest executions on random inputs, it holds that
\[
\pr{\deg_{\GredH}(\Is)<(\beta_1/2)\cdot n}=\pr{\deg_{\GblueH}(\Is)<(\beta_1/2)\cdot n} = \epsilon(n).
\]
Therefore, with a non-negligible probability, the number of parties corrupted by $\AdvneHonest$ is bounded by
\[
\deg_{\GredH}(\Is)+\deg_{\GblueH}(\Is)<\beta_1\cdot n.
\]
Hence, with a non-negligible probability, the execution of $\pi_n$ with $\AdvneHonest$ will complete in \phaseI, where the view of the set of honest parties, except for $\Party_\Is$ is identically distributed as in an honest execution where $\Party_\Is$ has input $\tilde{X}_\Is$, and therefore the $\Is$'th coordinate of the common output will be $Y^\honest_\Is=\tilde{X}_\Is$. However, since $\beta_1\cdot n < \beta\cdot n=t$, it follows from the \emph{validity} property, that the $\Is$'th coordinate of the common output is $Y^\honest_\Is=X_\Is$ except for negligible probability. Furthermore, for every party that is not a neighbour of $\Party_\Is$ it holds that the party is honest and that its view is identically distributed in an execution with $\AdvneCorrupt$ as in an execution with $\AdvneHonest$, therefore, by \emph{agreement}, $Y^\corrupt_\Is=Y^\honest_\Is$ except for negligible probability.
Since both $X_\Is$ and $\tilde{X}_\Is$ are random elements in $\zn$, it holds that $X_\Is \neq \tilde{X}_\Is$ with a noticeable probability, and we derive a contradiction.
\QED
\end{proof}

From \cref{claim:LB_few_neighbors} it follows that with overwhelming probability, $\AdvneCorrupt$ will not complete the attack in \stepref{step:lb:finish_phase_one} of \phaseI.
We now turn to bound from below the probability of event $\E_1\cap\E_2$. Recall events $\E_1$ and $\E_2$ as defined in \cref{sec:lower_bound_proof}. Loosely speaking, $\E_1$ is the event that $\Party_\Is$ is the last party to reach high-degree threshold in both red and blue executions. $\E_2$ is the event that $\Party_\Is$ reaches the degree threshold in the red execution before the blue execution.

\begin{claim}\label{claim:LB_many_neighbors}
Consider an execution of protocol $\pi_n$ on random inputs with adversary $\AdvneCorrupt$. Then
\begin{enumerate}
    \item
    $\pr{\E_1} \geq 1/n^2 - \negl(n)$.
    \item
    $\pr{\E_2 \mid \E_1} = 1/2$.
    \item
    $\pr{\E_2 \cap \E_1} \geq 1/{2n^2} - \negl(n)$.
\end{enumerate}
\end{claim}
\begin{proof}
By \cref{claim:LB_few_neighbors}, with overwhelming probability, a random party will have degree $(\beta_1/2)\cdot n$ in an honest execution over random inputs, by the end of the protocol.
Thus, \emph{all} parties must reach degree $(\beta_1/2)\cdot n$ with overwhelming probability.
By choosing $\Is$ uniformly from $[n]$, we conclude that $\Party_\Is$ will be the \emph{last} party to do so with probability $1/n-\negl(n)$. Since both the red and the blue executions are independent honest executions over random inputs, it follows that $\Party_\Is$ will be last in \emph{both} executions with probability $1/n^2-\negl(n)$.
The second part follows by symmetry, and the third part by definition of conditional probability.
\QED
\end{proof}

From \cref{claim:LB_many_neighbors} it follows that with a noticeable probability, the adversary $\AdvneCorrupt$ will proceed to \phaseIII. We now show that the size of the partition $\Gamma_1$ (set at the beginning of \phaseII) in this case is constant, and only depends on $\beta$.
In particular, we wish to argue that the identity of the final remaining sublinear cut in the graph cannot reveal too much information about the input $X_\Is$. (Recall that $\redExec$ is the execution with $X_\Is$ replaced by $\tilde{X}_\Is$.)
The proof follows from the graph-theoretic theorem (\cref{lem:number_of_cuts}), stated in \cref{sec:GT_lemma}. We start by looking at the communication graph at the beginning of \phaseII \emph{without} the chosen party $\Party_\Is$ (which depends only on the red execution). Party $\Party_\Is$ is added to one of the sets in the partition based on his edges at the end of \phaseII. Finally, we show that given the red-execution graph there are at most $2^{2c}$ possible choices for the final cut.

\begin{claim}\label{claim:LB_number_of_cuts}
Let $c > 1$ be a constant integer satisfying $\beta_1/2\geq 1/c$. Then, for sufficiently large $n$, in an execution with $\AdvneCorrupt$, conditioned on the event $\E_1\cap\E_2$, it holds that:
\begin{enumerate}
    \item\label{item:const_size_partition}
    There exists an $(\alpha(n),n/c)$-partition $\Gamma_1$ of $\GphaseIIC\setminus\set{\Is}$  of size at most $c$.
    \item\label{item:istar_set_in_partition}
    At the end of \phaseII there exists a set $\island_k\in\Gamma_1$ such that $\ssize{\edges_{\GphaseIIIC}\left(\set{\Is},\island_k\right)}>\alpha(n)$.
    \item
    Conditioned on the event $\E_1\cap \E_2$ and on $\redExec$, the random variable $\finalCut$ has at most $2c$ bits of entropy, \ie
    \[
    H\left(\finalCut\mid \E_1\cap \E_2, \redExec\right)\leq 2c.
    \]
\end{enumerate}
\end{claim}

\begin{proof}
For $n\in\N$, denote by $\GphaseIIC(n)$ the ``red graph'' that is obtained by $\pi_n$ running on input/randomness from $\redExec$ with $\AdvneCorrupt$.
From the definition of the events $\E_1$ and $\E_2$, it holds that the degree of every node in the red graph $\GphaseIIC(n)$ is greater or equal to $(\beta_1/2)\cdot n \geq n/c$. By removing the red node $\Is$ (corresponding to the virtual party $\DParty_\Is$ running on input $\tilde{X}_\Is$), we obtain the graph $G_{n-1}=\GphaseIIC(n)\setminus\sset{\Is}$ of size $n-1$.
Since by removing $\Is$, the degree of each node reduces by at most $1$, it holds that the degree of each vertex in $G_{n-1}$ is at least $n/c-1$. By applying \cref{lem:number_of_cuts} on the ensemble of graphs $\sset{G_{n-1}}_{n\in\N}$, for sufficiently large $n$, there exists a $(\alpha(n),n/c)$-partition of $G_n$, denoted $\Gamma_1=\sset{\island_1,\ldots,\island_\ell}$, of size $\ell<c$.\footnote{To be more precise, applying \cref{lem:number_of_cuts} on $\sset{G_{n-1}}_{n\in\N}$ gives an $(\alpha(n-1),(n-1)/c)$-partition. For clarity, we abuse the notation and write $(\alpha(n),n/c)$. This will not affect the subsequent calculations.}

From the definition of the events $\E_1$ and $\E_2$, it holds that the degree of $\Party_\Is$ in the blue execution will reach the threshold $(\beta_1/2)\cdot n$ (yet not before $\Party_\Is$ reaches the threshold in the red execution). It follows that the attack will enter \phaseIII. If for every $k\in[\ell]$, it holds that $\ssize{\edges_{\GphaseIIIC}\left(\sset{\Is},\island_k\right)}\leq\alpha(n)$, then the total number of neighbors of $\Party_\Is$ is bounded by $\ell\cdot\alpha(n)\in o(n)$, and we get a contradiction.

Finally, since $\Gamma_1$ is an $(\alpha(n),n/c)$-partition of the $(n-1)$-size graph $\GphaseIIC(n)\setminus\sset{\Is}$, it holds that every $\alpha(n)$-cut $\sset{S_n,\comp{S}_n}$ of $\GphaseIIC(n)\setminus\sset{\Is}$ can be represented as $S_n=\bigcup_{k\in A} U_k$ and $\comp{S}_n=\bigcup_{k\in [\ell]\setminus A} U_k$ for some $A\subsetneq[\ell]$. There are at most $2^c$ such subsets. Next, consider the $n$-size graph $\GphaseIIC(n)$. When adding the node $\Is$ (and its edges) back to $\GphaseIIC(n)\setminus\sset{\Is}$, the number of potential $\alpha(n)$-cuts at most doubles, since for every potential $\alpha(n)$-cut $\sset{S,\comp{S}}$ in $\GphaseIIC(n)\setminus\sset{\Is}$, either $\sset{S\cup{\Is},\comp{S}}$ is an $\alpha(n)$-cut in $\GphaseIIC(n)$, or $\sset{S,\comp{S}\cup{\Is}}$ is an $\alpha(n)$-cut in $\GphaseIIC(n)$
(or neither options induces an $\alpha(n)$-cut in $\GphaseIIC(n)$). Therefore, there are at most $2^{2c}$ potential $\alpha(n)$-cuts in $\GphaseIIC(n)$. It follows that $2c$ bits are sufficient to describe the support of $\finalCut$ given $\redExec$ (which in particular fully specifies the graph $\GphaseIIC(n)$).
\QED
\end{proof}

In \cref{claim:LB_honestparty,claim:LB_view_simulatable,claim:LB_mutualinfo,claim:LB_entropy}, we prove that there exists an honest party $\Party_\js$ that receives a bounded amount of information about $X_\Is$ by the end of the protocol's execution.
For $n\in\N$, denote by $\GendC(n)$ the final communication graph of $\pi_n$ when running on input/randomness from $\inputCoins$ with $\AdvneCorrupt$.

\begin{claim}\label{claim:LB_honestparty}
Condition on the event $\E_1 \cap \E_2$. Then, at the conclusion of the protocol execution with $\AdvneCorrupt$ there exists an $\alpha(n)$-cut, denoted $\sset{S_n,\comp{S}_n}$, in the induced graph $\GendC(n)$, and $\js\in[n]$ such that party $\Party_\js$ is honest, $\Is\in S_n$, and $\js\in\comp{S}_n$.
\end{claim}
\begin{proof}
By assumption, the communication graph of $\pi$ is not $\alpha(n)$-edge-connected; therefore, by definition, for sufficiently large $n$ there exists a cut $\sset{S_n,\comp{S}_n}$ of weight at most $\alpha(n)$ in $\GendC(n)$. Let $\sset{S_n,\comp{S}_n}$ be such a cut, and assume without loss of generality that $\Is\in S_n$. We now prove that there exists an honest party in the complement set $\comp{S}_n$.

Let $\Gamma_1=\sset{\island_1,\ldots,\island_\ell}$ be the $(\alpha(n),n/c)$-partition of $\GphaseIIC(n)\setminus\sset{\Is}$ defined in \phaseII of the attack. By Item~\ref{item:istar_set_in_partition} of \cref{claim:LB_number_of_cuts}, at the end of the \phaseIII there exists $k\in[\ell]$ for which $\ssize{\edges(\sset{\Is},U_k)}\geq \alpha(n)$.

Since $\sset{S_n,\comp{S}_n}$ is an $\alpha(n)$-cut of $\GendC(n)$, it holds that $\sset{S_n\setminus\sset{\Is},\comp{S}_n}$ is an $\alpha(n)$-cut of $\GphaseIIC\setminus\sset{\Is}$ (since any $\alpha$-cut in a graph remains an $\alpha$-cut when removing additional edges). It holds that $S_n\setminus\sset{\Is}=\bigcup_{k\in A} \island_k$ for some $\emptyset\neq A\subsetneq[\ell]$. Similarly, $\comp{S}_n=\bigcup_{k\in [\ell]\setminus A} \island_k$.
Recall that in an execution with $\AdvneCorrupt$ we have the following corruption pattern:
\begin{itemize}
    \item
    During \phaseI only $\Party_\Is$ gets corrupted.
    \item
    During \phaseII, for every pair of $\island_k,\island_{k'}\in\Gamma_1$, at most $\alpha(n)$ communicating parties get corrupted.
    \item
    During \phaseIII, for every pair of $\Visland_k,\Visland_{k'}\in\Gamma_2$, at most $\alpha(n)$ communicating parties get corrupted. (Recall that $\Visland_k\in\Gamma_2$ either equals $\island_k$ or equals $\island_k\cup\sset{\Is}$.)
\end{itemize}
It follows that there are at most $\ell^2\cdot \alpha(n)$ corrupted parties in $\comp{S}_n$.
By definition, $\ssize{\comp{S}_n}\geq n/c$ (since $\ssize{\island_k}\geq n/c$, for each $k\in[\ell]$).
By Item~\ref{item:const_size_partition} of \cref{claim:LB_number_of_cuts}, $\ell\leq c$ for some constant $c$, meaning that there is a linear number of parties in each side of the cut, but  only a sublinear number of corruptions, and the claim follows.
\QED
\end{proof}

We now prove that the view of an honest $\Party_\js\in\comp{S}_n$ can be perfectly simulated given: (1) All inputs and random coins that were used in the \emph{red} execution. This information is captured in the random variable $\redExec$, and deterministically determines the partition $\Gamma_1$ at the beginning of \phaseII. (2) The identities of the parties in $\comp{S}_n$. This information is captured in the random variable $\finalCut$.

\begin{claim}\label{claim:LB_view_simulatable}
Conditioned on the event $\E_1 \cap \E_2$, the view $\VIEW^\corrupt_\js$ of honest party $\Party_\js$ at the end of the protocol is simulatable by $\redExec$ and by $\finalCut$.
\end{claim}
\begin{proof}
Let $\sset{S_n,\comp{S}_n}$ be the $\alpha(n)$-cut that is guaranteed to exist at the end of the protocol (the ``minimal'' such cut according to some canonical ordering). Assume that $\Is\in S_n$, and let $\js\in \comp{S}_n$ be an index of the honest party that exists by \cref{claim:LB_honestparty}.
The view of $\Party_\js$ is defined as its input, random coins, and the messages it received during the protocol.\footnote{Formally, in every round every party receives a vector of $n$ messages from $\fpsmt$, where some may be the empty string. Therefore, a party also knows the identity of the sender of every message.}

During \phaseI the view of $\Party_\js$ is identically distributed as a view of an honest party in an honest execution over inputs $(X_1,\ldots,\tilde{X}_\Is, \ldots, X_n)$ and random coins $(R_1,\ldots,\tilde{R}_\Is, \ldots, R_n)$ (without loss of generality, we can assume that the joint view of all honest parties is exactly $(X_1,\ldots,\tilde{X}_\Is, \ldots, X_n)$ and $(R_1,\ldots,\tilde{R}_\Is, \ldots, R_n)$, \ie all the information that deterministically defines the red execution). Indeed, this can be easily simulated given the random variable $\redExec$ by running an honest execution until all parties have degree $(\beta_1/2)\cdot n$. Furthermore, the partition $\Gamma_1=(\island_1,\ldots,\island_\ell)$ is deterministically determined by $\redExec$, as well as the view of every honest party at the end of \phaseI (except for $\Party_\Is$).

Next, consider the parties in $\comp{S}_n$ (which are determined by the random variable $\finalCut$). By \cref{lem:number_of_cuts}, for large enough $n$ it holds that $\comp{S}_n=\cup_{k\in A} \island_k$ for some nonempty $A\subsetneq[\ell]$. That is, given $\Gamma_1=\sset{\island_1,\ldots, \island_\ell}$, the random variable $\finalCut$ is fully specified by $A\subseteq[\ell]$, which can be described as an element of $\zo^\ell$. As before, and without loss of generality, for every $j\in\comp{S}_n$, the view of party $\Party_j$ at the end of \phaseI can also be set as $(X_1,\ldots,\tilde{X}_\Is, \ldots, X_n)$ and $(R_1,\ldots,\tilde{R}_\Is, \ldots, R_n)$.

Observe that given the final cut $\sset{S_n,\comp{S}_n}$, it holds that during Phases~II and III every message that is being sent to a party in $\comp{S}_n$ by a party in $S_n$ is ignored. Therefore, it may seem that the simulation can be resumed by continuing an honest execution of all parties in $\comp{S}_n$ based on their joint view at the end of \phaseI. However, this approach will not suffice since the adversary in the real execution blocks any message that is sent by parties in different $\island_k$'s, even if both parties are members of $\comp{S}_n$. To simulate this behavior, it is important to know exactly when to block a message sent between different $\island_k$'s and when to keep it. Indeed, this is the reason for keeping track of the number of edges between every pair of $\island_k$'s and blocking message only until the threshold $\alpha(n)$ is reached.

The simulation of Phases~II and III therefore proceeds by running the protocol honestly for every $\Party_j$, with $j\in\comp{S}_n$, however, with the following two exceptions. First, every party is simulated as if he does not receive any message from parties outside of $\comp{S}_n$, and whenever he is instructed to send a message to a party outside of $\comp{S}_n$, a dummy party is simulated as receiving this message. Second, any message that is sent between two parties $\Party_{j_1}$ and $\Party_{j_2}$ is discarded whenever $j_1\in\island_k$ and $j_2\in\island_{k'}$ (respectively, $j_1\in\Visland_k$ and $j_2\in\Visland_{k'}\setminus\sset{\Is}$), for some $k\neq k'$, and it holds that $\ssize{\edges(\island_k,\island_{k'})}<\alpha(n)$ (respectively, $\ssize{\edges(\Visland_k,\Visland_{k'})}<\alpha(n)$).

The above simulation identically emulates the view of very honest party in $\comp{S}_n$, since the adversary indeed discards every message that is sent from some party in $S_n$ to some party in $\comp{S}_n$, as well as from parties in different $\island_k$'s, but otherwise behaves honestly.
\QED
\end{proof}

The core of the proof (\cref{claim:LB_entropy}) is showing that a constant number of bits suffices to describe $\finalCut$. We note that while $\redExec$ is independent of $X_\Is$ conditioned on $\E_1$ and on $\Is$ (as shown in \cref{claim:LB_mutualinfo} below), there is some correlation between $\finalCut$ and $X_\Is$.
\begin{claim} \label{claim:LB_mutualinfo}
Conditioned on $\Is$ and on the event $\E_1$, the random variables $X_\Is$ and $\redExec$ are independent. That is,
\[
\mutualinfo{X_\Is;\redExec \mid \E_1,\Is}=0.
\]
\end{claim}
\begin{proof}
Recall that over the sampling of
\[
\inputCoins=\left(X_1,\ldots,X_n,R_1,\ldots,R_n,\tilde{X}_1,\ldots,\tilde{X}_n,\tilde{R}_1,\ldots,\tilde{R}_n,\Is\right),
\]
the random variable $\redExec$ is defined as
\[
\redExec=\left(X_{-\Is},\tilde{X}_\Is,R_{-\Is},\tilde{R}_\Is\right).
\]
Similarly, define
\[
\blueExec=\left(\tilde{X}_{-\Is},X_\Is,\tilde{R}_{-\Is},R_\Is\right).
\]
Note that (even) given the value of $\Is$, $\redExec$ and $\blueExec$ are simply uniform distributions, independent of each other.
That is, for every $\is\in[n]$
\[
\mutualinfo{\blueExec~;~\redExec \mid \Is=\is}=0.
\]
This in turn implies that
\[
\mutualinfo{\blueExec~;~\redExec \mid \Is}=0.
\]

We observe that given $\Is$, the event $\E_1$ decomposes as a conjunction of independent events.
Namely, consider the (deterministic) predicate
\[
\lastParty\left(\left(x_{-\is},x_\is,r_{-\is},r_\is\right),\is\right),
\]
which simulates an honest execution of $\pi_n$ with corresponding inputs and randomness, and outputs~$1$ if party $\Party_\is$ is the last party to reach degree $\beta n/4$. Then,
\[
\E_1 = \left(\lastParty\left(\blueExec,\Is\right)=1\right) \wedge \left(\lastParty\left(\redExec,\Is\right)=1\right).
\]
It follows that

\begin{align*}
\mutualinfo{X_\Is~;~\redExec \mid \Is, \E_1}
&\leq
\mutualinfo{\blueExec~;~\redExec \mid \Is, \E_1}\\
&=
\mutualinfo{
\begin{array}{c|c}
\Centerstack{
\blueExec~;~\redExec
}
&
\Centerstack{
$\Is$\\
$\lastParty\left(\blueExec,\Is\right)=1$\\
$\lastParty\left(\redExec,\Is\right)=1$
}\\
\end{array}
}.
\end{align*}

\noindent
We will show that the last term is equal to zero. This proves our claim, as mutual information cannot be negative.
\begin{align*}
0
&=
\mutualinfo{\blueExec~;~\redExec \mid \Is}\\
&=
\mutualinfo{\blueExec, \lastParty\left(\blueExec,\Is\right)~;~\redExec, \lastParty\left(\redExec,\Is\right) \mid \Is}\\
&\geq
\mutualinfo{
\begin{array}{c|c}
\Centerstack{
\blueExec~;~\redExec
}
&
\Centerstack{
$\Is$\\
$\lastParty\left(\blueExec,\Is\right)$\\
$\lastParty\left(\redExec,\Is\right)$
}\\
\end{array}
},
\end{align*}

\noindent
where the second equality follows from the fact that $\lastParty$ is a deterministic function of
inputs and randomness of all the parties and $\Is$. This implies that
\begin{align*}
\mutualinfo{
\begin{array}{c|c}
\Centerstack{
\blueExec~;~\redExec
}
&
\Centerstack{
$\Is$\\
$\lastParty\left(\blueExec,\Is\right)$\\
$\lastParty\left(\redExec,\Is\right)$
}\\
\end{array}
}
&=
0.
\end{align*}

\noindent
Finally, since the event $\E_1$ occurs with non-zero probability, it holds that
\begin{align*}
\mutualinfo{
\begin{array}{c|c}
\Centerstack{
\blueExec~;~\redExec
}
&
\Centerstack{
$\Is$\\
$\lastParty\left(\blueExec,\Is\right)=1$\\
$\lastParty\left(\redExec,\Is\right)=1$
}\\
\end{array}
}
&=
0,
\end{align*}

\noindent
This concludes the proof of \cref{claim:LB_mutualinfo}.
\QED
\end{proof}

\begin{claim}\label{claim:LB_entropy}
For sufficiently large $n$, conditioned on the event $\E_1 \cap \E_2$, the input $X_\Is$ retains $n/2$ bits of entropy given the view of honest party $\Party_\js$, \ie
\[
H(X_\Is \mid \E_1 \cap \E_2, \VIEW^\corrupt_\js)\geq n/2.
\]
\end{claim}
\begin{proof}

First, since $\pr{\E_1}\geq 1/n^2$, it holds that
\begin{align}
H(X_\Is \mid \E_1) \geq n-\log(n^2). \label{eq:LB1}
\end{align}
Indeed, for an arbitrary $x\in\zn$,
\[
\pr{X_\Is=x \mid \E_1} = \frac{\pr{X_\Is=x,\E_1}}{\pr{\E_1}} \leq \frac{\pr{X_\Is=x}}{\pr{\E_1}} \leq \frac{n^2}{2^n},
\]
where the last inequality uses the fact that $X_\Is$ is uniformly distributed in $\zn$ and that $\pr{\E_1}\geq 1/n^2$.
\cref{eq:LB1} follows by the following computation.
\begin{align*}
H(X_\Is \mid \E_1)
&= \sum_{x\in\zn}\pr{X_\Is=x \mid \E_1} \cdot \log\left(\frac{1}{\pr{X_\Is=x \mid \E_1}}\right) \\
&\geq \sum_{x\in\zn}\pr{X_\Is=x \mid \E_1}\cdot\log\left(\frac{2^n}{n^2}\right) \\
&= \log\left(\frac{2^n}{n^2}\right)\cdot \underbrace{\sum_{x\in\zn}\pr{X_\Is=x \mid \E_1}}_{=\ 1} \\
&= n-\log(n^2).
\end{align*}

\noindent
Since $\Is$ represents an element of the set $[n]$, the support of $\Is$ is $\zo^{\log(n)}$, and it holds that
\begin{align}
H(X_\Is \mid \E_1, \Is)
&\geq H(X_\Is \mid \E_1) - \log(n). \label{eq:LB2-Is}
\end{align}

\noindent
In addition, since conditioned on $\E_1$ and $\Is$, the random variables $\redExec$ and $X_\Is$ are independent (\cref{claim:LB_mutualinfo}), it holds that
\begin{align}
H(X_\Is \mid \E_1, \Is, \redExec)
&= H(X_\Is \mid \E_1, \Is) - \mutualinfo{X_\Is;\redExec \mid \E_1, \Is} \notag \\
&= H(X_\Is \mid \E_1, \Is). \label{eq:LB2}
\end{align}

\noindent
Next, since $\pr{\E_2 \mid \E_1}=1/2$ (\cref{claim:LB_many_neighbors}), it follows that
\begin{align}
H(X_\Is \mid \E_1 \cap \E_2, \Is, \redExec) \geq 2\cdot H(X_\Is \mid \E_1, \Is, \redExec)-n -2. \label{eq:LB3}
\end{align}

\noindent
To prove this, we define an indicator random variable $\Ind_2$ for the event $\E_2$, \ie $\pr{\Ind_2=1}=\pr{\E_2}$, and use the fact that since $\Ind_2$ is an indicator random variable, the mutual information $\mutualinfo{X_\Is;\Ind_2 \mid \redExec,\E_1, \Is}$ cannot be more than $1$. Therefore,
\begin{align}
H(X_\Is \mid \E_1, \Is, \redExec,\Ind_2)
&= H(X_\Is \mid \E_1, \Is, \redExec) - \mutualinfo{X_\Is;\Ind_2 \mid \E_1, \Is, \redExec} \notag\\
&\geq H(X_\Is \mid \E_1, \Is, \redExec) - 1. \notag
\end{align}

\noindent
\cref{eq:LB3} now follows by the following computation.
\begin{align*}
H(X_\Is \mid \E_1, \Is, \redExec)-1
&\leq H(X_\Is \mid \E_1, \Is, \redExec,\Ind_2) \\
&= \pr{\Ind_2=1 \mid \E_1}\cdot H(X_\Is \mid \E_1, \Is, \redExec,\Ind_2=1) \\
& \quad + \pr{\Ind_2=0 \mid \E_1}\cdot \underbrace{H(X_\Is \mid \E_1, \Is, \redExec,\Ind_2=0)}_{\leq\ n} \\
&\leq \underbrace{\pr{\E_2 \mid \E_1}}_{=\ 1/2} \cdot H(X_\Is \mid \E_1 \cap \E_2, \Is, \redExec) + \underbrace{\pr{\neg \E_2 \mid \E_1}}_{=\ 1/2}\cdot n \\
&\leq 1/2 \cdot H(X_\Is \mid \E_1 \cap \E_2, \Is, \redExec) + n/2.
\end{align*}

\noindent
By \cref{claim:LB_number_of_cuts}, the support of $\finalCut$ is $\zo^{2c}$ and it holds that
\begin{align}
H(X_\Is &\mid \E_1 \cap \E_2, \Is, \redExec,\finalCut) \nonumber\\
&= H(X_\Is \mid \E_1 \cap \E_2, \Is, \redExec) - \mutualinfo{X_\Is \mid \E_1 \cap \E_2, \Is, \redExec, \finalCut} \nonumber\\
&\geq H(X_\Is \mid \E_1 \cap \E_2, \Is, \redExec) - H(\finalCut \mid \E_1 \cap \E_2, \Is, \redExec) \nonumber\\
&\geq H(X_\Is \mid \E_1 \cap \E_2, \Is, \redExec) - 2c. \label{eq:LB4}
\end{align}

\noindent
The first inequality holds since the conditional mutual information is upper-bounded by the conditional entropy term.
Finally, by \cref{claim:LB_view_simulatable} it holds that
\begin{align}
H(X_\Is \mid \E_1 \cap \E_2, \VIEW^\corrupt_\js)\geq H(X_\Is \mid \E_1 \cap \E_2,\Is, \redExec,\finalCut). \label{eq:LB5}
\end{align}

\noindent
By combining \cref{eq:LB1,eq:LB2,eq:LB2-Is,eq:LB3,eq:LB4,eq:LB5} together, it holds that
\begin{align*}
H(X_\Is \mid \E_1 \cap \E_2, \VIEW^\corrupt_\js)
&\geq H(X_\Is \mid \E_1 \cap \E_2,\Is, \redExec,\finalCut)\\
&\geq H(X_\Is \mid \E_1 \cap \E_2,\Is, \redExec)-2c\\
&\geq 2 \cdot H(X_\Is \mid \E_1,\Is, \redExec)-n -2-2c\\
&= 2 \cdot H(X_\Is \mid \E_1,\Is)-n -2-2c\\
&\geq 2 \cdot H(X_\Is \mid \E_1)-2 \cdot \log(n)-n -2-2c\\
&\geq 2 n-2 \cdot \log(n^2)-2 \cdot \log(n) -n -2-2c\\
&\geq n-2 \cdot (\log(n^2)+\log(n)+1 +c).
\end{align*}

\noindent
The claim follows for sufficiently large $n$.
\QED
\end{proof}

This concludes the proof of \cref{lem:lb:entropy}.
\QED
\end{proof} 

\subsubsection{Proving the Common Output Contains \texorpdfstring{$X_\Is$}{Lg}\SUBSUBSEC}\label{sec:LB_lemma_correct_output}

We now turn to the second main lemma of the proof.
We show that although party $\Party_\Is$ is corrupted in the execution with $\AdvneCorrupt$, all honest parties must output its initial input value $X_\Is$ at the conclusion of the protocol. This is done by analyzing an execution with the dual adversary, $\AdvneHonest$ (described in \cref{sec:LB_adv}), that does not corrupt the chosen party $\Party_\Is$ but simulates the attack by $\AdvneCorrupt$ toward the honest parties.

\begin{lemma}\label{lem:lb:correct_output}
Consider an execution of $\pi_n$ on random inputs $(X_1,\ldots,X_n)$ for the parties with adversary $\AdvneHonest$. Then, conditioned on the event $\E_1 \cap \E_2$ it holds that:
\begin{enumerate}
    \item
    The $\Is$'th coordinate of the common output $Y^\honest_\Is$ equals the initial input $X_\Is$ of $\Party_\Is$, except for negligible probability, \ie
    \[
    \pr{Y^\honest_\Is=X_\Is \mid \E_1 \cap \E_2}\geq 1-\negl(n).
    \]
    \item
    The $\Is$'th coordinate of the common output $Y^\honest_\Is$ in an execution with $\AdvneHonest$ equals the $\Is$'th coordinate of the common output $Y^\corrupt_\Is$ in an execution with $\AdvneCorrupt$, except for negligible probability, \ie
    \[
    \pr{Y^\honest_\Is=Y^\corrupt_\Is \mid \E_1 \cap \E_2}\geq 1-\negl(n).
    \]
\end{enumerate}
\end{lemma}

\begin{proof}
The first part of the lemma follows (almost) from the \emph{validity} property of parallel broadcast. It is only left to prove that the event $\E_1 \cap \E_2$ is non-negligible in an execution with $\AdvneHonest$. However, as $\E_1$ and $\E_2$ are events on honest executions (mirrored by both $\AdvneHonest$ and $\AdvneCorrupt$ in \phaseI), this follows by \cref{claim:LB_many_neighbors}. We spell this out in \cref{claim:LB_few_neighbors_honest,claim:LB_E1E2}.

\begin{claim}\label{claim:LB_few_neighbors_honest}
Let $\GphaseIIH$ denote the random variable representing the graph of the red execution $(\Party_1,\ldots,\DParty_\Is,\ldots,\Party_n)$ at the conclusion of \phaseI with adversary $\AdvneHonest$. Then
\[
\pr{\deg_{\GphaseIIH}(\Is)<(\beta_1/2) \cdot n}\leq \negl(n).
\]
\end{claim}
\begin{proof}
As proven in \cref{claim:LB_few_neighbors}, in an execution on random inputs with adversary $\AdvneHonest$, it holds that the degree of every party reaches $(\beta_1/2) \cdot n$ except for negligible probability.
\QED
\end{proof}

\begin{claim}\label{claim:LB_E1E2}
Consider an execution of protocol $\pi_n$ on random inputs with adversary $\AdvneHonest$. Then
\[
\pr{\E_2 \cap \E_1} \geq 1/{2n^2} - \negl(n).
\]
\end{claim}
\begin{proof}
By \cref{claim:LB_few_neighbors_honest} every party has degree $(\beta_1/2) \cdot n$ with overwhelming probability. The rest of the proof follows exactly as the proof of \cref{claim:LB_many_neighbors}.
\QED
\end{proof}

\noindent
This complete the proof of the first part of the \cref{lem:lb:correct_output}.

For proving the second part of the lemma, we show that in an execution with $\AdvneHonest$ there exists an honest party whose view, and in particular his output, is identically distributed as in an execution with $\AdvneCorrupt$. This follows from a simple counting argument.

First, denote by $C_n^\honest$ the random variable representing the set of corrupted parties in an execution with $\AdvneHonest$.
Since $\beta<1/3$ and there are at most $\beta n$ corrupted parties, it holds that $\ssize{C_n^\honest}<n/3$.
Note that by construction, $\AdvneHonest$ corrupts during \phaseI all neighbors of the virtual party $\DParty_\Is$ in the red execution and all neighbors of $\Party_\Is$ in the blue execution, and in Phases II and III, all parties that belong to some $\island_k\in\Gamma_1$ (respectively, $\Visland_k\in\Gamma_2$) and receive messages from parties in $\island_{k'}$ (respectively, $V_{k'}$) for $k\neq k'$ (until there are $\alpha(n)$ such edges, and except for $\Party_\Is$).

Second, denote by $C_n^\corrupt$ the random variable representing the following set of parties in an execution with $\AdvneCorrupt$ (note that these parties are not necessarily corrupted here):
All neighbors of the virtual party $\DParty_\Is$ in the red execution and all neighbors of $\Party_\Is$ in the blue execution during \phaseI, and all parties that get corrupted during Phases II and III, \ie all parties that belong to some $\island_k\in\Gamma_1$ (respectively, $\Visland_k\in\Gamma_2$) and receive messages from parties in $\island_{k'}$ (respectively, $\Visland_{k'}$) for $k\neq k'$ (until there are $\alpha(n)$ such edges, and except for $\Party_\Is$).
By symmetry, it holds that $\ssize{C_n^\corrupt}<n/3$.

For every value of \inputCoins, the size of the union of these sets $C_n^\honest \cup C_n^\corrupt$ is at most $2n/3$. By the constructions of $\AdvneHonest$ and $\AdvneCorrupt$, the view of any party outside of this union is identically distributed in an execution with $\AdvneHonest$ as in an execution with $\AdvneCorrupt$, conditioned on the event $\E_1\cap\E_2$. Indeed, in both executions the view of every party outside of $C_n^\honest \cup C_n^\corrupt$ is distributed according to an honest execution on random inputs during \phaseI. During \phaseII, the view is distributed according the continuation of the protocol under omission failures between every pair $U_k,U_{k'}\in\Gamma_1$. During \phaseIII, the view is distributed according the continuation of the protocol under omission failures between every pair $V_k,V_{k'}\in\Gamma_2$, with the exception that the internal state of a random party $\Party_\Is$ is replaced by a view of $\Party_\Is$ in a honest execution with an independently distributed random inputs for all parties.

This concludes the proof of \cref{lem:lb:correct_output}.
\QED
\end{proof} 

\ifdefined\IsProofInAppendix\else
\subsection{Proof of the Graph-Theoretic Theorem (\texorpdfstring{\cref{lem:number_of_cuts}}{Lg})}\label{sec:GT_lemma_cont}

In this section, we prove \cref{lem:number_of_cuts}, stating that an $(\alpha,d)$-partition exists and can be efficiently computed when the minimal degree of a vertex in the graph is sufficiently large.

\begin{customthm}{5.6}
Let $c>1$ be a constant integer, let $\alpha(n)\in o(n)$ be a fixed sublinear function in $n$, and let $\sset{G_n}_{n\in\N}$ be a family of graphs, where $G_n=([n],E_n)$ is defined on $n$ vertices, and every vertex of $G_n$ has degree at least $\frac{n}{c}-1$.
Then, for sufficiently large $n$ it holds that:
\begin{enumerate}
    \item
    There exists an $(\alpha(n),n/c)$-partition of $G_n$, denoted $\Gamma$; it holds that $\size{\Gamma} \leq c$.
    \item
    An $(\alpha(n),n/c)$-partition $\Gamma$ of $G_n$ can be found in (deterministic) polynomial time.
\end{enumerate}
\end{customthm}

In \cref{lem:GT_existence}, we prove the existence of such a partition; and in \cref{lem:GT_algorithm}, we present a deterministic polynomial-time algorithm for computing it.
Recall that given a graph $G$ with $n$ vertices and a subset $S\subseteq[n]$, we denote by $\edges(S)=\edges(S,\comp{S})$ the set of edges from $S$ to its complement.

\begin{lemma}\label{lem:GT_existence}
Consider the setting of \cref{lem:number_of_cuts}. Then, there exists an $(\alpha(n),n/c)$-partition $\Gamma$ of $G_n$, and it holds that $\size{\Gamma} \leq c$. Furthermore, the number of $\alpha(n)$-cuts in $G_n$ is at most $2^{c-1}$.
\end{lemma}
\begin{proof}
Let $n\in\N$ and let
\[
\set{\sset{S_1,\comp{S}_1},\ldots,\sset{S_\ell,\comp{S}_\ell}}
\]
be the set of all $\alpha(n)$-cuts in $G_n$, \ie for every $i\in[\ell]$ it holds that
\begin{equation}\label{eq:valid-cut}
\size{\edges_{G_n}(S_i)} \leq \alpha(n).
\end{equation}

\noindent
We proceed by defining the following family of subsets
\[
\Gamma=\set{\bigcap_{i=1}^\ell S_i^{b_i} : (b_1,b_2,\ldots,b_\ell)\in\zo^\ell}\setminus\set{\emptyset},
\]
where for every $i\in[\ell]$ and $b\in\zo$, the set $S_i^b$ is defined as
\[
S_i^b \assign
\begin{cases}
S_i        & \text{\em if } b=0, \\
\comp{S_i} & \text{\em if } b=1.
\end{cases}
\]

\noindent
We will show that for sufficiently large $n$, the set $\Gamma$ is an $(\alpha(n),n/c)$-partition of $G_n$ and that $\ssize{\Gamma}\leq c$.
We start by proving two useful claims.
\begin{claim}\label{claim:large-sets}
For every $S\subseteq[n]$, if $1\leq \ssize{S}\leq \frac{n}{c}-1$ then $\ssize{\edges_{G_n}(S)}\geq \frac{n}{c}-1$.
\end{claim}

\begin{proof}
Consider an arbitrary set $S\subseteq[n]$. The claim follows from the following set of inequalities:
\begin{align*}
    \ssize{\edges_{G_n}(S)} & = (\text{total degree of the vertices in $S$}) \\
    & \hspace{3cm}- 2\cdot (\text{total number of edges whose both vertices are inside $S$}) \\
    & \geq \ssize{S}\cdot \left(\frac{n}{c}-1\right) - 2\cdot \binom{\ssize{S}}{2} \\
    & = \ssize{S}\cdot \left(\frac{n}{c}-1\right) - \ssize{S}\cdot \left(\ssize{S}-1\right) \\
    & = \ssize{S}\cdot \left(\frac{n}{c}-\ssize{S}\right) \\
    & \stackrel{(\ast)}{\geq} \frac{n}{c}-1.
\end{align*}
The last inequality $(\ast)$ follows since for a constant $a$ and $1\leq x\leq a-1$, if $f(x)=x(a-x)$ then $a-1\leq f(x)\leq \frac{a^2}{4}$.
In our case $1\leq \ssize{S}\leq \frac{n}{c}-1$; therefore
\[
\frac{n}{c}-1 \leq \size{S}\cdot \left(\frac{n}{c}-\size{S}\right) \leq \frac{n^2}{4c^2}.
\]
This concludes the proof of \cref{claim:large-sets}.
\QED
\end{proof}

\begin{claim}\label{claim:large-intersection}
For every $(b_1,b_2,\ldots,b_\ell)\in\zo^\ell$, either $\bigcap_{i=1}^\ell S_i^{b_i}=\emptyset$ or $\ssize{\bigcap_{i=1}^\ell S_i^{b_i}} \geq n/c$.
\end{claim}

\begin{proof}
Suppose, to the contrary, that there exists an $\ell$-bit vector $(b_1,b_2,\ldots,b_\ell)\in\zo^\ell$ such that $1\leq \ssize{\bigcap_{i=1}^\ell S_i^{b_i}} \leq \frac{n}{c}-1$. Let us consider the following nested sequence of sets
\[
A_\ell\subseteq \ldots\subseteq A_2\subseteq A_1,
\]
where for every $j\in[\ell]$ the set $A_j$ is defined as $A_j\assign\bigcap_{i=1}^{j}S_i^{b_i}$.
Since $A_1=S_1^{b_1}$ and $S_1$ satisfies \cref{eq:valid-cut}, it holds that
\begin{equation}\label{eq:edges-A_1}
\ssize{\edges_{G_n}(A_1)} \leq \alpha(n).
\end{equation}
Also, since $A_\ell=\bigcap_{i=1}^\ell S_i^{b_i}$ and $|\bigcap_{i=1}^\ell S_i^{b_i}| \leq \frac{n}{c}-1$ (by assumption), it follows from \cref{claim:large-sets} that
\begin{equation}\label{eq:edges-A_k}
\ssize{\edges_{G_n}(A_\ell)} \geq \frac{n}{c}-1.
\end{equation}

\noindent
The following two cases can occur:

\paragraph{Case 1.}
For every $j\in[\ell-1]$, either $A_j\setminus A_{j+1}=\emptyset$ or $\ssize{A_j\setminus A_{j+1}} \geq n/c$:
Consider the set $\mathcal{I} = \{i\in[\ell-1]: A_i\setminus A_{i+1} \neq \emptyset\}$.
It follows from the assumption that for every $i\in\mathcal{I}$, $\ssize{A_i\setminus A_{i+1}}\geq n/c$.
Since these sets are disjoint and $A_{\ell}$ is non-empty, we have $\ssize{\mathcal{I}}\leq c-1$.
Note that there are at least $n/c-1$ edges
which are coming out of $A_{\ell}$.
Since $\ssize{\edges(A_1)}\leq\alpha(n)$, at least $n/c-1-\alpha(n)$ of these edges (whose one end-point is inside $A_{\ell}$) must have their other end-point inside $A_1\setminus A_{\ell}$.
Observe that $A_1\setminus A_{\ell}=\bigcup_{i=1}^{\mathcal{I}}A_i\setminus A_{i+1}$, which implies $\ssize{\edges(A_{\ell}, A_k\setminus A_{k+1})}\geq \frac{(n/c)-1-\alpha(n)}{c-1}$ for some $k\in\mathcal{I}$.
Since $A_{\ell}\subseteq S_{k+1}^{b_{k+1}}$, which is disjoint from $A_k\setminus A_{k+1}$, we have
$\ssize{\edges(S_{k+1}^{b_{k+1}})}\geq\frac{(n/c)-1-\alpha(n)}{c-1}$. Now, since $\alpha(n)$ is a
sub-linear function in $n$, for sufficiently large $n$ we have that $\ssize{\edges(S_{k+1}^{b+1})} > \alpha(n)$,
which is a contradiction to \cref{eq:valid-cut}.

\paragraph{Case 2.} $0<\ssize{A_j\setminus A_{j+1}} \leq n/c-1$ for some $j\in[\ell-1]$:
As in the previous case, consider the set $\mathcal{I} = \{i\in[\ell-1]: A_i\setminus A_{i+1} \neq \emptyset\}$.
Let $k\in\mathcal{I}$ be the first index such that $0<\ssize{A_k\setminus A_{k+1}} \leq n/c-1$.
Define another set $\mathcal{I}'=\{i\in\mathcal{I}: i\leq k\}$.
Since $k\in\mathcal{I}$ is the first index such that $0<\ssize{A_k\setminus A_{k+1}} \leq n/c-1$,
the number of $i$'s such that $i\in\mathcal{I}$ and $i\leq k$ must be less than $c$,
which implies that $\ssize{\mathcal{I}'}\leq c-1$.
Note that $A_k$ and $A_k\setminus A_{k+1}$ can be written as $A_k=\bigcap_{i\in\mathcal{I}'}S_i^{b_i}$ and
$A_k\setminus A_{k+1}=\bigcap_{i\in\mathcal{I}'}S_i^{b_i}\bigcap S_{k+1}^{b_{k+1}\oplus 1}$, respectively.
Let $s\assign\ssize{\mathcal{I}'}$.
For simplicity of notation, let us renumber these cuts from 1 to $s+1$, i.e.,
let $\{S_1,S_2,\hdots,S_{s+1}\}\assign\{S_i:i\in\mathcal{I}'\}\bigcup S_{k+1}$.
With this notation, there exists some $(\hat{b}_1,\hat{b}_2,\ldots,\hat{b}_{s+1})\in\zo^{s+1}$, such that $A_k\setminus A_{k+1} = \bigcap_{i=1}^{s+1} S_i^{\hat{b}_i}$
Since $0<\ssize{A_k\setminus A_{k+1}} \leq n/c-1$, we have from \cref{claim:large-sets} that $\ssize{\edges(\bigcap_{i=1}^{s+1} S_i^{\hat{b}_i})}\geq n/c-1$.

Note that for any edge $(u,v)\in \edges(\bigcap_{i=1}^{s+1} S_i^{\hat{b}_i})$, it holds that $v\in\bigcap_{i=1}^{s+1} S_i^{e_i}$ for some $(e_1,e_2,\ldots,e_{s+1})\in\zo^{s+1}$.
Since $\ssize{\edges(\bigcap_{i=1}^{s+1} S_i^{\hat{b}_i})}\geq n/c-1$, there exists $(\hat{e}_1,\hat{e}_2,\ldots,\hat{e}_{s+1})\in\zo^{s+1}$ such that $\ssize{\edges(\bigcap_{i=1}^{s+1} S_i^{\hat{b}_i},\bigcap_{i=1}^{s+1} S_i^{\hat{e}_i})} \geq \frac{(n/c)-1}{2^{s+1}}$. Since $s+1=\ssize{\mathcal{I}'}+1\leq c$, we have
\begin{align}\label{eq:large-intersection-interim4}
\size{\edges\left(\bigcap_{i=1}^{s+1} S_i^{\hat{b}_i},\bigcap_{i=1}^{s+1} S_i^{\hat{e}_i}\right)} \geq \frac{(n/c)-1}{2^{c}}.
\end{align}
Since $(\hat{b}_1,\hat{b}_2,\ldots,\hat{b}_{s+1})\neq(\hat{e}_1,\hat{e}_2,\ldots,\hat{e}_{s+1})$, there exists $l\in[s+1]$ such that $\hat{b}_l \neq \hat{e}_l$. This, together with \cref{eq:large-intersection-interim4}, implies that $\ssize{\edges(S_l)} \geq \frac{(n/c)-1}{2^c}$.
Since $c$ is a constant, it holds $\frac{(n/c)-1}{2^c}> \alpha(n)$ for sufficiently large $n$, implying that $\ssize{\edges(S_l)}>\alpha(n)$, which is a contradiction to \cref{eq:valid-cut}.
This concludes the proof of \cref{claim:large-intersection}.
\QED
\end{proof}

Now we show that $\Gamma$ is an $(\alpha(n),n/c)$-partition of $G_n$.
First observe that the union of all the sets in $\Gamma$ is indeed $[n]$.
By \cref{claim:large-intersection} for every $S\in\Gamma$ it holds that $\ssize{S}\geq n/c$. Furthermore, for every $\alpha(n)$-cut $\sset{S_j, \comp{S}_j}$ it holds that $S_j$ and $\comp{S}_j$ can be represented as a union of some sets from $\Gamma$, more precisely
\[
S_j=\bigcup_{(\vec{b}\in\zo^\ell \ : \ b_j=0)}\bigcap_{i=1}^\ell S_i^{b_i}
\quad \text{ and }\quad
\comp{S}_j=\bigcup_{(\vec{b}\in\zo^\ell \ : \ b_j=1)}\bigcap_{i=1}^\ell S_i^{b_i}.
\]
In addition, it is easy to see that the sets in $\Gamma$ are pairwise disjoint, and since each is of size at least $n/c$ it holds that $\ssize{\Gamma}\leq c$.

To show that distinct sets from $\Gamma$ have at most $\alpha(n)$ crossing edges, consider two different sets $U_i,U_j\in \Gamma$. There exist distinct binary vectors $\vec{b},\vec{e}\in\zo^\ell$ such that $U_i=\bigcap_{k=1}^\ell S_k^{b_k}$ and $U_j=\bigcap_{k=1}^\ell S_k^{e_k}$.
Let $\hat{k}\in[\ell]$ be an index such that $b_{\hat{k}}\neq e_{\hat{k}}$.
Since (1) $U_i\subseteq S_{\hat{k}}^{b_{\hat{k}}}$, (2) $U_j\subseteq S_{\hat{k}}^{e_{\hat{k}}}$, (3)
$S_{\hat{k}}^{b_{\hat{k}}}\cap S_{\hat{k}}^{e_{\hat{k}}}=\emptyset$, and (4)
$\ssize{\edges(S_{\hat{k}}^{b_{\hat{k}}},S_{\hat{k}}^{e_{\hat{k}}})}\leq\alpha(n)$, it holds that $\ssize{\edges(U_i,U_j)}\leq\alpha(n)$.

Finally, since $\ssize{\Gamma}\leq c$ and for every $\alpha(n)$-cut $\sset{S_j, \comp{S}_j}$ it holds that $S_j$ and $\comp{S}_j$ can be represented as a union of some sets from $\Gamma$, we have that $\ell\leq 2^c-1$ (which is the total number of nonempty subsets of $\Gamma$).
However, a cut is defined by two such subsets; therefore it holds that the total number of $\alpha(n)$-cuts is at most $2^{c-1}$.

This completes the proof of \cref{lem:GT_existence}.
\QED
\end{proof}

\cref{lem:GT_existence} proved existence of a partition. In \cref{lem:GT_algorithm}, we show how to efficiently find such a partition.
A core element of our algorithm is the algorithm for enumerating all cuts of a weighted graph due to Vazirani and Yannakakis \cite{VaziraniYa92} that runs in polynomial time with $\tilde{O}(n^2m)$ delay between two successive outputs (\ie once the \ith cut was found, this is the running time needed to produce the $(i+1)$'th cut). Yeh et al.\ \cite{YehWangSu10} reduced the delays to $\tilde{O}(nm)$.

\begin{theorem}[\cite{VaziraniYa92,YehWangSu10}]\label{thm:GT_findcuts}
Let $G=(V,E)$ be a weighted undirected graph on $n$ vertices and $m$ edges.
There exists a deterministic polynomial-time algorithm to enumerate all the cuts of $G$ by non-decreasing weights with $\tilde{O}(n m)$ delay between two successive outputs.
\end{theorem}

\begin{lemma}\label{lem:GT_algorithm}
Consider the setting of \cref{lem:number_of_cuts}. Then, for sufficiently large $n$, an $(\alpha(n),n/c)$-partition $\Gamma$ of $G_n$ can be found in polynomial time (precisely in $\tilde{O}(n^3)$ time).
\end{lemma}
\begin{proof}
We start by assigning weight $1$ to every edge of our graph $G_n$ and run the algorithm from \cref{thm:GT_findcuts} on this weighted graph.
We enumerate all cuts by non-decreasing weights until we hit a cut whose weight is more than $\alpha(n)$.
Since the total number of $\alpha(n)$-cuts is at most $2^{c-1}$, we stop this algorithm after it has enumerated
at most $2^{c-1}$ cuts.

Let $\sset{S_1,\comp{S}_1},\ldots,\sset{S_{\ell},\comp{S}_{\ell}}$ denote all the $\alpha(n)$-cuts of $G_n$.
We define the partition $\Gamma$ as above
\[
\Gamma=\set{\bigcap_{i=1}^\ell S_i^{b_i} : (b_1,b_2,\ldots,b_\ell)\in\zo^\ell}\setminus\set{\emptyset}.
\]
Following the proof of \cref{lem:GT_existence}, the partition $\Gamma$ is an $(\alpha(n),n/c)$-partition of $G_n$.

\paragraph{Running time analysis.}
Since $c$ is a constant and total number of edges in our graph is $O(n^2)$, we can generate all the $\alpha(n)$-cuts in $\tilde{O}(n^3)$ time using the algorithm of Yeh et al.\ \cite{YehWangSu10}.
Once we have found all the $\alpha(n)$-cuts, generating the $(\alpha(n),n/c)$-partition takes $O(n)$ time: since the total number of cuts $\ell$ is constant, computing $\bigcap_{i=1}^\ell S_i^{b_i}$ for any fixed vector $(b_1,b_2,\ldots,b_\ell)\in\zo^\ell$ takes $O(n)$ time. We are computing $2^{\ell}\leq 2^{2^{c-1}}$ such intersections corresponding to the $2^\ell$ vectors. Because $c$ is a constant, the total time is still $O(n)$.
Hence, the total time required by our procedure for finding an $(\alpha(n),n/c)$-partition is $\tilde{O}(n^3)$.
This completes the proof of \cref{lem:GT_algorithm}.
\QED
\end{proof}

\fi

{\small{
\bibliographystyle{alpha}
\bibliography{crypto}

\newcommand{\etalchar}[1]{$^{#1}$}
\begin{thebibliography}{DDWY93}

\bibitem[ACD{\etalchar{+}}19]{ACDNPRS19}
Ittai Abraham, T.{-}H.~Hubert Chan, Danny Dolev, Kartik Nayak, Rafael Pass,
  Ling Ren, and Elaine Shi.
\newblock Communication complexity of byzantine agreement, revisited.
\newblock In {\em Proceedings of the 38th Annual ACM Symposium on Principles of
  Distributed Computing (PODC)}, pages 317--326, 2019.

\bibitem[ACdH06]{ACH06}
Saurabh Agarwal, Ronald Cramer, and Robbert de~Haan.
\newblock Asymptotically optimal two-round perfectly secure message
  transmission.
\newblock In {\em Advances in Cryptology -- CRYPTO 2006}, pages 394--408, 2006.

\bibitem[ALM17]{ALM17}
Adi Akavia, Rio LaVigne, and Tal Moran.
\newblock Topology-hiding computation on all graphs.
\newblock In {\em Advances in Cryptology -- CRYPTO 2017, part {I}}, pages
  447--467, 2017.

\bibitem[AM17]{AM17}
Adi Akavia and Tal Moran.
\newblock Topology-hiding computation beyond logarithmic diameter.
\newblock In {\em Advances in Cryptology -- EUROCRYPT 2017, part {III}}, pages
  609--637, 2017.

\bibitem[BBC{\etalchar{+}}19]{BBCMM19}
Marshall Ball, Elette Boyle, Ran Cohen, Tal Malkin, and Tal Moran.
\newblock Is information-theoretic topology-hiding computation possible?
\newblock In {\em Proceedings of the 17th Theory of Cryptography Conference,
  TCC 2019, part {I}}, pages 502--530, 2019.

\bibitem[BBMM18]{BBMM18}
Marshall Ball, Elette Boyle, Tal Malkin, and Tal Moran.
\newblock Exploring the boundaries of topology-hiding computation.
\newblock In {\em Advances in Cryptology -- EUROCRYPT 2018, part {III}}, pages
  294--325, 2018.

\bibitem[BCDH18]{BCDH18}
Elette Boyle, Ran Cohen, Deepesh Data, and Pavel Hub{\'{a}}{\v{c}}ek.
\newblock Must the communication graph of {MPC} protocols be an expander?
\newblock In {\em Advances in Cryptology -- CRYPTO 2018, part {III}}, pages
  243--272, 2018.

\bibitem[BCG21]{BCG21}
Elette Boyle, Ran Cohen, and Aarushi Goel.
\newblock Breaking the {O}({\(\surd\)} n)-bit barrier: {Byzantine} agreement
  with polylog bits per party.
\newblock In {\em Proceedings of the 40th Annual ACM Symposium on Principles of
  Distributed Computing (PODC)}, pages 319--330, 2021.

\bibitem[BCP15]{BCP15}
Elette Boyle, Kai{-}Min Chung, and Rafael Pass.
\newblock Large-scale secure computation: Multi-party computation for
  (parallel) {RAM} programs.
\newblock In {\em Advances in Cryptology -- CRYPTO 2015, part {II}}, pages
  742--762, 2015.

\bibitem[Bea91]{Beaver91}
Donald Beaver.
\newblock Foundations of secure interactive computing.
\newblock In {\em Advances in Cryptology -- CRYPTO '91}, pages 377--391, 1991.

\bibitem[Bei07]{Beimel07}
Amos Beimel.
\newblock On private computation in incomplete networks.
\newblock {\em Distributed Computing}, 19(3):237--252, 2007.

\bibitem[BF99]{BF99}
Amos Beimel and Matthew~K. Franklin.
\newblock Reliable communication over partially authenticated networks.
\newblock {\em Theoretical Computer Science}, 220(1):185--210, 1999.

\bibitem[BG93]{BG93}
Piotr Berman and Juan~A. Garay.
\newblock Fast consensus in networks of bounded degree.
\newblock {\em Distributed Computing}, 7(2):67--73, 1993.

\bibitem[BGH13]{BGH13}
Nicolas Braud{-}Santoni, Rachid Guerraoui, and Florian Huc.
\newblock Fast byzantine agreement.
\newblock In {\em Proceedings of the 32th Annual ACM Symposium on Principles of
  Distributed Computing (PODC)}, pages 57--64, 2013.

\bibitem[BGLS03]{BGLS03}
Dan Boneh, Craig Gentry, Ben Lynn, and Hovav Shacham.
\newblock Aggregate and verifiably encrypted signatures from bilinear maps.
\newblock In {\em Advances in Cryptology -- EUROCRYPT 2003}, pages 416--432,
  2003.

\bibitem[BGT11]{BGK11}
Elette Boyle, Shafi Goldwasser, and Yael {Tauman Kalai}.
\newblock Leakage-resilient coin tossing.
\newblock In {\em Proceedings of the 25th International Symposium on
  Distributed Computing (DISC)}, pages 181--196, 2011.

\bibitem[BGT13]{BGT13}
Elette Boyle, Shafi Goldwasser, and Stefano Tessaro.
\newblock Communication locality in secure multi-party computation - how to run
  sublinear algorithms in a distributed setting.
\newblock In {\em Proceedings of the 10th Theory of Cryptography Conference,
  TCC 2013}, pages 356--376, 2013.

\bibitem[BGW88]{BGW88}
Michael {Ben-Or}, Shafi Goldwasser, and Avi Wigderson.
\newblock Completeness theorems for non-cryptographic fault-tolerant
  distributed computation (extended abstract).
\newblock In {\em Proceedings of the 20th Annual ACM Symposium on Theory of
  Computing (STOC)}, pages 1--10, 1988.

\bibitem[BJLM06]{BJLM06}
Markus Bl{\"{a}}ser, Andreas Jakoby, Maciej Li{\'{s}}kiewicz, and Bodo Manthey.
\newblock Private computation: k-connected versus 1-connected networks.
\newblock {\em Journal of Cryptology}, 19(3):341--357, 2006.

\bibitem[BJLM11]{BJLM11}
Markus Bl{\"{a}}ser, Andreas Jakoby, Maciej Li{\'{s}}kiewicz, and Bodo Manthey.
\newblock Privacy in non-private environments.
\newblock {\em Theory Comput. Syst.}, 48(1):211--245, 2011.

\bibitem[BM05]{BM05}
Amos Beimel and Lior Malka.
\newblock Efficient reliable communication over partially authenticated
  networks.
\newblock {\em Distributed Computing}, 18(1):1--19, 2005.

\bibitem[BMR90]{BMR90}
Donald Beaver, Silvio Micali, and Phillip Rogaway.
\newblock The round complexity of secure protocols (extended abstract).
\newblock In {\em Proceedings of the 22nd Annual ACM Symposium on Theory of
  Computing (STOC)}, pages 503--513, 1990.

\bibitem[Can00]{Canetti00}
Ran Canetti.
\newblock Security and composition of multiparty cryptographic protocols.
\newblock {\em Journal of Cryptology}, 13(1):143--202, 2000.

\bibitem[Can01]{Canetti01}
Ran Canetti.
\newblock {Universally Composable Security: A New Paradigm for Cryptographic
  Protocols}.
\newblock In {\em Proceedings of the 42nd Annual Symposium on Foundations of
  Computer Science (FOCS)}, pages 136--145, 2001.

\bibitem[CCD88]{CCD88}
David Chaum, Claude Cr{\'{e}}peau, and Ivan Damg{\aa}rd.
\newblock Multiparty unconditionally secure protocols (extended abstract).
\newblock In {\em Proceedings of the 20th Annual ACM Symposium on Theory of
  Computing (STOC)}, pages 11--19, 1988.

\bibitem[CCG{\etalchar{+}}15]{CCGGOZ15}
Nishanth Chandran, Wutichai Chongchitmate, Juan~A. Garay, Shafi Goldwasser,
  Rafail Ostrovsky, and Vassilis Zikas.
\newblock The hidden graph model: Communication locality and optimal resiliency
  with adaptive faults.
\newblock In {\em Proceedings of the 6th Annual Innovations in Theoretical
  Computer Science (ITCS) conference}, pages 153--162, 2015.

\bibitem[CCGZ19]{CCGZ16}
Ran Cohen, Sandro Coretti, Juan~A. Garay, and Vassilis Zikas.
\newblock Probabilistic termination and composability of cryptographic
  protocols.
\newblock {\em Journal of Cryptology}, 32(3):690--741, 2019.

\bibitem[CCGZ21]{CCGZ17}
Ran Cohen, Sandro Coretti, Juan~A. Garay, and Vassilis Zikas.
\newblock Round-preserving parallel composition of probabilistic-termination
  cryptographic protocols.
\newblock {\em Journal of Cryptology}, 34(2):12, 2021.

\bibitem[CDD{\etalchar{+}}99]{CDDHR99}
Ronald Cramer, Ivan Damg{\aa}rd, Stefan Dziembowski, Martin Hirt, and Tal
  Rabin.
\newblock Efficient multiparty computations secure against an adaptive
  adversary.
\newblock In {\em Advances in Cryptology -- EUROCRYPT '99}, pages 311--326,
  1999.

\bibitem[CDF01]{CDF01}
Ronald Cramer, Ivan Damg{\aa}rd, and Serge Fehr.
\newblock On the cost of reconstructing a secret, or {VSS} with optimal
  reconstruction phase.
\newblock In {\em Advances in Cryptology -- CRYPTO 2001}, pages 503--523, 2001.

\bibitem[CFOR12]{CFOR12}
Alfonso Cevallos, Serge Fehr, Rafail Ostrovsky, and Yuval Rabani.
\newblock Unconditionally-secure robust secret sharing with compact shares.
\newblock In {\em Advances in Cryptology -- EUROCRYPT 2012}, pages 195--208,
  2012.

\bibitem[CGO10]{CGO10}
Nishanth Chandran, Juan~A. Garay, and Rafail Ostrovsky.
\newblock Improved fault tolerance and secure computation on sparse networks.
\newblock In {\em Proceedings of the 37th International Colloquium on Automata,
  Languages, and Programming (ICALP), part {II}}, pages 249--260, 2010.

\bibitem[CGO12]{CGO12}
Nishanth Chandran, Juan~A. Garay, and Rafail Ostrovsky.
\newblock Edge fault tolerance on sparse networks.
\newblock In {\em Proceedings of the 39th International Colloquium on Automata,
  Languages, and Programming (ICALP), part {II}}, pages 452--463, 2012.

\bibitem[CGO15]{CGO15}
Nishanth Chandran, Juan~A. Garay, and Rafail Ostrovsky.
\newblock Almost-everywhere secure computation with edge corruptions.
\newblock {\em Journal of Cryptology}, 28(4):745--768, 2015.

\bibitem[CHOR18]{CHOR18}
Ran Cohen, Iftach Haitner, Eran Omri, and Lior Rotem.
\newblock Characterization of secure multiparty computation without broadcast.
\newblock {\em Journal of Cryptology}, 31(2):587--609, 2018.

\bibitem[CL17]{CL17}
Ran Cohen and Yehuda Lindell.
\newblock Fairness versus guaranteed output delivery in secure multiparty
  computation.
\newblock {\em Journal of Cryptology}, 30(4):1157--1186, 2017.

\bibitem[DDWY93]{DDWY93}
Danny Dolev, Cynthia Dwork, Orli Waarts, and Moti Yung.
\newblock Perfectly secure message transmission.
\newblock {\em Journal of the ACM}, 40(1):17--47, 1993.

\bibitem[DI05]{DI05}
Ivan Damg{\aa}rd and Yuval Ishai.
\newblock Constant-round multiparty computation using a black-box pseudorandom
  generator.
\newblock In {\em Advances in Cryptology -- CRYPTO 2005}, pages 378--394, 2005.

\bibitem[DKM{\etalchar{+}}17]{DKMSZ17}
Varsha Dani, Valerie King, Mahnush Movahedi, Jared Saia, and Mahdi Zamani.
\newblock Secure multi-party computation in large networks.
\newblock {\em Distributed Computing}, 30(3):193--229, 2017.

\bibitem[Dol82]{Dolev82}
Danny Dolev.
\newblock The byzantine generals strike again.
\newblock {\em J. Algorithms}, 3(1):14--30, 1982.

\bibitem[DPPU88]{DkPPU88}
Cynthia Dwork, David Peleg, Nicholas Pippenger, and Eli Upfal.
\newblock Fault tolerance in networks of bounded degree.
\newblock {\em SIAM Journal on Computing}, 17(5):975--988, 1988.

\bibitem[DS83]{DS83}
Danny Dolev and Raymond Strong.
\newblock Authenticated algorithms for {Byzantine} agreement.
\newblock {\em SIAM Journal on Computing}, 12(4):656--666, 1983.

\bibitem[DW10]{DW10}
Zeev Dvir and Avi Wigderson.
\newblock Monotone expanders: Constructions and applications.
\newblock {\em Theory of Computing}, 6(1):291--308, 2010.

\bibitem[Fei99]{Feige99}
Uriel Feige.
\newblock Noncryptographic selection protocols.
\newblock In {\em Proceedings of the 40th Annual Symposium on Foundations of
  Computer Science (FOCS)}, pages 142--153, 1999.

\bibitem[FFGS07]{FFGV07}
Matthias Fitzi, Matthew~K. Franklin, Juan~A. Garay, and Harsha~Vardhan
  Simhadri.
\newblock Towards optimal and efficient perfectly secure message transmission.
\newblock In {\em Proceedings of the Fourth Theory of Cryptography Conference,
  TCC 2007}, pages 311--322, 2007.

\bibitem[FLM86]{FLM86}
Michael~J. Fischer, Nancy~A. Lynch, and Michael Merritt.
\newblock Easy impossibility proofs for distributed consensus problems.
\newblock {\em Distributed Computing}, 1(1):26--39, 1986.

\bibitem[FM97]{FM97}
Pesech Feldman and Silvio Micali.
\newblock An optimal probabilistic protocol for synchronous byzantine
  agreement.
\newblock {\em SIAM Journal on Computing}, 26(4):873--933, 1997.

\bibitem[FY04]{FY04}
Matthew~K. Franklin and Moti Yung.
\newblock Secure hypergraphs: Privacy from partial broadcast.
\newblock {\em {SIAM} Journal on Discrete Mathematics}, 18(3):437--450, 2004.

\bibitem[GKKZ11]{GKKZ11}
Juan~A. Garay, Jonathan Katz, Ranjit Kumaresan, and Hong{-}Sheng Zhou.
\newblock Adaptively secure broadcast, revisited.
\newblock In {\em Proceedings of the 30th Annual ACM Symposium on Principles of
  Distributed Computing (PODC)}, pages 179--186, 2011.

\bibitem[GL90]{GL90}
Shafi Goldwasser and Leonid~A. Levin.
\newblock Fair computation of general functions in presence of immoral
  majority.
\newblock In {\em Advances in Cryptology -- CRYPTO '90}, pages 77--93, 1990.

\bibitem[GM93]{GM93}
Juan~A. Garay and Yoram Moses.
\newblock Fully polynomial byzantine agreement in t+1 rounds.
\newblock In {\em Proceedings of the 25th Annual ACM Symposium on Theory of
  Computing (STOC)}, pages 31--41, 1993.

\bibitem[GMW87]{GMW87}
Oded Goldreich, Silvio Micali, and Avi Wigderson.
\newblock How to play any mental game or a completeness theorem for protocols
  with honest majority.
\newblock In {\em Proceedings of the 19th Annual ACM Symposium on Theory of
  Computing (STOC)}, pages 218--229, 1987.

\bibitem[GO08]{GO08}
Juan~A. Garay and Rafail Ostrovsky.
\newblock Almost-everywhere secure computation.
\newblock In {\em Advances in Cryptology -- EUROCRYPT 2008}, pages 307--323,
  2008.

\bibitem[Gol04]{Goldreich04}
Oded Goldreich.
\newblock {\em Foundations of Cryptography -- VOLUME 2: Basic Applications}.
\newblock Cambridge University Press, 2004.

\bibitem[Gol11]{G11}
Oded Goldreich.
\newblock A sample of samplers: {A} computational perspective on sampling.
\newblock In {\em Studies in Complexity and Cryptography}, pages 302--332.
  Springer, 2011.

\bibitem[GVZ06]{GVZ06}
Ronen Gradwohl, Salil~P. Vadhan, and David Zuckerman.
\newblock Random selection with an adversarial majority.
\newblock In {\em Advances in Cryptology -- CRYPTO 2006}, pages 409--426, 2006.

\bibitem[HIJ{\etalchar{+}}16]{HIJKR16}
Shai Halevi, Yuval Ishai, Abhishek Jain, Eyal Kushilevitz, and Tal Rabin.
\newblock Secure multiparty computation with general interaction patterns.
\newblock In {\em Proceedings of the 7th Annual Innovations in Theoretical
  Computer Science (ITCS) conference}, pages 157--168, 2016.

\bibitem[HIK07]{HIK07}
Danny Harnik, Yuval Ishai, and Eyal Kushilevitz.
\newblock How many oblivious transfers are needed for secure multiparty
  computation?
\newblock In {\em Advances in Cryptology -- CRYPTO 2007}, pages 284--302, 2007.

\bibitem[HLP11]{HLP11}
Shai Halevi, Yehuda Lindell, and Benny Pinkas.
\newblock Secure computation on the web: Computing without simultaneous
  interaction.
\newblock In {\em Advances in Cryptology -- CRYPTO 2011}, pages 132--150, 2011.

\bibitem[HLW06]{HLW06}
Shlomo Hoory, Nathan Linial, and Avi Wigderson.
\newblock Expander graphs and their applications.
\newblock {\em Bull. Amer. Math. Soc.}, 43(4):439--561, 2006.

\bibitem[HMTZ16]{HMTZ16}
Martin Hirt, Ueli Maurer, Daniel Tschudi, and Vassilis Zikas.
\newblock Network-hiding communication and applications to multi-party
  protocols.
\newblock In {\em Advances in Cryptology -- CRYPTO 2016, part {II}}, pages
  335--365, 2016.

\bibitem[HZ10]{HZ10}
Martin Hirt and Vassilis Zikas.
\newblock Adaptively secure broadcast.
\newblock In {\em Advances in Cryptology -- EUROCRYPT 2010}, pages 466--485,
  2010.

\bibitem[IOZ14]{IOZ14}
Yuval Ishai, Rafail Ostrovsky, and Vassilis Zikas.
\newblock Secure multi-party computation with identifiable abort.
\newblock In {\em Advances in Cryptology -- CRYPTO 2014, part {II}}, pages
  369--386, 2014.

\bibitem[KGSR02]{KGSR02}
M.~V. N.~Ashwin Kumar, Pranava~R. Goundan, K.~Srinathan, and C.~Pandu Rangan.
\newblock On perfectly secure cmmunication over arbitrary networks.
\newblock In {\em Proceedings of the 21th Annual ACM Symposium on Principles of
  Distributed Computing (PODC)}, pages 193--202, 2002.

\bibitem[KKK{\etalchar{+}}08]{KKKSS08}
Bruce~M. Kapron, David Kempe, Valerie King, Jared Saia, and Vishal Sanwalani.
\newblock Fast asynchronous byzantine agreement and leader election with full
  information.
\newblock In {\em Proceedings of the 19th Annual {ACM-SIAM} Symposium on
  Discrete Algorithms (SODA)}, pages 1038--1047, 2008.

\bibitem[KLST11]{KLST11}
Valerie King, Steven Lonargan, Jared Saia, and Amitabh Trehan.
\newblock Load balanced scalable byzantine agreement through quorum building,
  with full information.
\newblock In {\em Proceedings of the 12th International Conference on
  Distributed Computing and Networking (ICDCN)}, pages 203--214, 2011.

\bibitem[KRS16]{KRS16}
Ranjit Kumaresan, Srinivasan Raghuraman, and Adam Sealfon.
\newblock Network oblivious transfer.
\newblock In {\em Advances in Cryptology -- CRYPTO 2016, part {II}}, pages
  366--396, 2016.

\bibitem[KS09a]{KS09}
Valerie King and Jared Saia.
\newblock From almost everywhere to everywhere: Byzantine agreement with
  {\~{o}}(n\({}^{\mbox{3/2}}\)) bits.
\newblock In {\em Proceedings of the 23th International Symposium on
  Distributed Computing (DISC)}, pages 464--478, 2009.

\bibitem[KS09b]{KurSuz09}
Kaoru Kurosawa and Kazuhiro Suzuki.
\newblock Truly efficient 2-round perfectly secure message transmission scheme.
\newblock {\em IEEE Transactions on Information Theory}, 55(11):5223--5232,
  2009.

\bibitem[KS10]{KS10}
Valerie King and Jared Saia.
\newblock Breaking the \emph{O}(\emph{n}\({}^{\mbox{2}}\)) bit barrier:
  scalable byzantine agreement with an adaptive adversary.
\newblock In {\em Proceedings of the 29th Annual ACM Symposium on Principles of
  Distributed Computing (PODC)}, pages 420--429, 2010.

\bibitem[KSSV06]{KSSV06}
Valerie King, Jared Saia, Vishal Sanwalani, and Erik Vee.
\newblock Scalable leader election.
\newblock In {\em Proceedings of the 17th Annual {ACM-SIAM} Symposium on
  Discrete Algorithms (SODA)}, pages 990--999, 2006.

\bibitem[LMRS04]{LMRS04}
Anna Lysyanskaya, Silvio Micali, Leonid Reyzin, and Hovav Shacham.
\newblock Sequential aggregate signatures from trapdoor permutations.
\newblock In {\em Advances in Cryptology -- EUROCRYPT 2004}, pages 74--90,
  2004.

\bibitem[LOS{\etalchar{+}}13]{LOSSW13}
Steve Lu, Rafail Ostrovsky, Amit Sahai, Hovav Shacham, and Brent Waters.
\newblock Sequential aggregate signatures, multisignatures, and verifiably
  encrypted signatures without random oracles.
\newblock {\em Journal of Cryptology}, 26(2):340--373, 2013.

\bibitem[LSP82]{LSP82}
Leslie Lamport, Robert~E. Shostak, and Marshall~C. Pease.
\newblock The {Byzantine} generals problem.
\newblock {\em ACM Transactions on Programming Languages and Systems},
  4(3):382--401, 1982.

\bibitem[MOR01]{MOR01}
Silvio Micali, Kazuo Ohta, and Leonid Reyzin.
\newblock Accountable-subgroup multisignatures: extended abstract.
\newblock In {\em Proceedings of the 8th {ACM} Conference on Computer and
  Communications Security (CCS)}, pages 245--254, 2001.

\bibitem[MOR15]{MOR15}
Tal Moran, Ilan Orlov, and Silas Richelson.
\newblock Topology-hiding computation.
\newblock In {\em Proceedings of the 12th Theory of Cryptography Conference,
  TCC 2015, part {I}}, pages 159--181, 2015.

\bibitem[MR91]{MR91}
Silvio Micali and Phillip Rogaway.
\newblock Secure computation (abstract).
\newblock In {\em Advances in Cryptology -- CRYPTO '91}, pages 392--404, 1991.

\bibitem[PSL80]{PSL80}
Marshall~C. Pease, Robert~E. Shostak, and Leslie Lamport.
\newblock Reaching agreement in the presence of faults.
\newblock {\em Journal of the ACM}, 27(2):228--234, 1980.

\bibitem[RB89]{RB89}
Tal Rabin and Michael {Ben-Or}.
\newblock Verifiable secret sharing and multiparty protocols with honest
  majority (extended abstract).
\newblock In {\em Proceedings of the 30th Annual Symposium on Foundations of
  Computer Science (FOCS)}, pages 73--85, 1989.

\bibitem[SA96]{SA96}
Hasan~Md. Sayeed and Hosame Abu{-}Amara.
\newblock Efficient perfectly secure message transmission in synchronous
  networks.
\newblock {\em Information and Control}, 126(1):53--61, 1996.

\bibitem[SASM10]{SASM10}
Takenobu Seito, Tadashi Aikawa, Junji Shikata, and Tsutomu Matsumoto.
\newblock Information-theoretically secure key-insulated multireceiver
  authentication codes.
\newblock In {\em Progress in Cryptology - {AFRICACRYPT} 2010}, pages 148--165,
  2010.

\bibitem[SHZI02]{SHZI02}
Junji Shikata, Goichiro Hanaoka, Yuliang Zheng, and Hideki Imai.
\newblock Security notions for unconditionally secure signature schemes.
\newblock In {\em Advances in Cryptology -- EUROCRYPT 2002}, pages 434--449,
  2002.

\bibitem[SNR04]{SNR04}
K.~Srinathan, Arvind Narayanan, and C.~Pandu Rangan.
\newblock Optimal perfectly secure message transmission.
\newblock In {\em Advances in Cryptology -- CRYPTO 2004}, pages 545--561, 2004.

\bibitem[SS11]{SS11}
Colleen Swanson and Douglas~R. Stinson.
\newblock Unconditionally secure signature schemes revisited.
\newblock In {\em Proceedings of the 5th International Conference on
  Information Theoretic Security {ICITS}}, pages 100--116, 2011.

\bibitem[SZ16]{SZ16}
Gabriele Spini and Gilles Z{\'{e}}mor.
\newblock Perfectly secure message transmission in two rounds.
\newblock In {\em Proceedings of the 14th Theory of Cryptography Conference,
  TCC 2016-B, part {I}}, pages 286--304, 2016.

\bibitem[Upf92]{Upfal92}
Eli Upfal.
\newblock Tolerating linear number of faults in networks of bounded degree.
\newblock In {\em Proceedings of the 11th Annual ACM Symposium on Principles of
  Distributed Computing (PODC)}, pages 83--89, 1992.

\bibitem[VY92]{VaziraniYa92}
Vijay~V. Vazirani and Mihalis Yannakakis.
\newblock Suboptimal cuts: Their enumeration, weight and number (extended
  abstract).
\newblock In {\em Proceedings of the 19th International Colloquium on Automata,
  Languages, and Programming (ICALP)}, pages 366--377, 1992.

\bibitem[YWS10]{YehWangSu10}
Li{-}Pu Yeh, Biing{-}Feng Wang, and Hsin{-}Hao Su.
\newblock Efficient algorithms for the problems of enumerating cuts by
  non-decreasing weights.
\newblock {\em Algorithmica}, 56(3):297--312, 2010.

\bibitem[Zuc97]{Z97}
David Zuckerman.
\newblock Randomness-optimal oblivious sampling.
\newblock {\em Random Struct. Algorithms}, 11(4):345--367, 1997.

\end{thebibliography}
}}

\appendix

\section{Preliminaries (Cont'd)}\label{sec:preliminaries_contd}

In \cref{sec:primitives}, we define the cryptographic primitives used in the paper, in \cref{sec::Def:model}, we define the MPC model, and in \cref{sec:PKI}, we present the correlated-randomness functionalities that are used in the paper.

\subsection{Cryptographic Primitives}\label{sec:primitives}

\subsubsection{Error-Correcting Secret Sharing}
We present the definition of error-correcting secret sharing, also known in the literature as robust secret sharing.
\begin{definition}\label{def:ECSS}
A \emph{$(t, n)$ error-correcting secret-sharing scheme (ECSS)} over a message space $\mathcal{M}$ consists of a pair of algorithms $(\Share, \Recon)$ satisfying the following properties:
\begin{enumerate}
    \item\textbf{$t$-privacy:}
    For every $m\in \mathcal{M}$, and every subset $\IS\subseteq[n]$ of size $\ssize{\IS}\leq t$, the distribution of $\sset{s_i}_{i\in \IS}$ is independent of $m$, where $(s_1,\ldots,s_n)\gets \Share(m)$.
    \item\textbf{Reconstruction from up to $t$ erroneous shares:}
    For every $m\in \mathcal{M}$, every $\vs = (s_1,\ldots,s_n)$, and every $\vs' = (s'_1,\ldots,s'_n)$ such that $\ppr{\vS\gets \Share(m)}{\vS=\vs}>0$ and $\ssize{\sset{i \mid s_i=s'_i}}\geq n-t$, it holds that $m=\Recon(\vs')$ (except for a negligible probability).
\end{enumerate}
\end{definition}

\noindent
ECSS can be constructed information-theoretically, with a negligible positive error probability, when $t<n/2$~\cite{RB89,CDF01,CFOR12}.

\subsubsection{Committee Election}
A committee-election protocol is a protocol for electing a subset (committee) of $n'$ parties out of a set of $n$ parties. In this work we consider electing uniformly at random committees of size $n'=\omega(\log{n})$. If the fraction of corrupted parties is constant at the original $n$-party set, then the fraction of corrupted parties in the committee is only slightly larger. This follows immediately by the analysis of  Boyle et al.\ \cite[Lem.\ 2.6]{BGK11} of Feige' lightest-bin protocol~\cite{Feige99}.

\begin{lemma}[\cite{BGK11}]\label{lem:Feige}
For any $n'<n$ and $0<\beta<1$, Feige's lightest-bin protocol is a $1$-round, $n$-party protocol for electing a committee $\committee$, such that for any set of corrupted parties $\IS\subseteq [n]$ of size $t=\beta n$, the following holds.
\begin{enumerate}
    \item
    $\size{\committee}\leq n'$.
    \item
    For every constant $\epsilon>0$, $\pr{\size{\committee\setminus\IS}\leq (1-\beta-\epsilon)n'} < \frac{n}{n'} e^{-\frac{\epsilon^2 n'}{2(1-\beta)}}$.
    \item
    For every constant $\epsilon>0$, $\pr{\frac{\size{\committee\cap \IS}}{\size{\committee}}\geq \beta+\epsilon} < \frac{n}{n'} e^{-\frac{\epsilon^2 n'}{2(1-\beta)}}$.
\end{enumerate}
\end{lemma}

The following corollary follows.
\begin{corollary}\label{cor:elect}
Let $\committee\subseteq[n]$ be a uniformly random subset of size $n'=\omega(\log{n})$. Let $\IS\subseteq[n]$ be a set of corrupted parties of size $t=\beta\cdot n$, for a constant $0<\beta<1$. Then, except for a negligible probability (in $n$), it holds that for an arbitrary small $\epsilon>0$
\[
\size{\committee\cap\IS} \leq (\beta+\epsilon) \cdot n'.
\]
\end{corollary}

\subsubsection{Information-Theoretic Signatures}\label{sec:itsign}
Parts of the following section are taken almost verbatim from~\cite{IOZ14}.

\paragraph{$\PS$-verifiable Information-Theoretic Signatures.}\label{appendix:signatures}
We recall the definition and construction of information-theoretic signatures~\cite{SHZI02,SASM10} but slightly modify the terminology to what we consider to be more intuitive. The signature scheme (in particular the key-generation algorithm) needs to know the total number of verifiers or alternatively the list $\PS$ of their identities. Furthermore, as usually for information-theoretic primitives, the key-length needs to be proportional to the number of times that the key is used. Therefore, the scheme is parameterized by two natural numbers $\sigcals$ and $\vercals$ which will be upper bounds on the number of
signatures that can be generated and verified, respectively, without violating the security.

A {\em $\PS$-verifiable signature scheme} consists of a triple of randomized algorithms $(\Gen,\Sign,\Verify)$, where:
\begin{enumerate}
    \item
    $\Gen(1^\secParam,n,\sigcals,\vercals)$ outputs a pair $(\sk,\vvk)$, where $\sk\in\zo^{\secParam}$ is a signing key, $\vvk=(\vk[1],\ldots,\vk[n])\in(\zo^\secParam)^n$ is a verification-key-vector consisting of (private) verification sub-keys, and $\sigcals,\vercals\in\mathbb{N}$.
    \item
    $\Sign(m,\sk)$ on input a message $m$ and the signing-key \sk outputs a signature $\sigma\in\zo^{\poly(\secParam)}$.
    \item
    $\Verify(m,\sigma,\vk[i])$ on input a message $m$, a signature $\sigma$ and a verification sub-key $\vk[i]$, outputs a decision bit $d\in\zo$.
\end{enumerate}

\begin{definition}
A $\PS$-verifiable signature scheme $(\Gen,\Sign,\Verify)$ is said to be {\em information-theoretically $(\sigcals,\vercals)$-secure} if it satisfies the following properties:
\begin{itemize}
    \item (completeness)
    A valid signature is accepted from any honest receiver:
    \[
    \Pr[\Gen\outp(\sk,(\vk[1],\ldots,\vk[n]));\ \text{ for } i\in[n] : (\Verify(m,\Sign(m,\sk),\vk[i])=1)] = 1.
    \]
    \item
    Let $\oSig[\sk]$ denote a signing oracle (on input $m$, $\oSig[\sk]$ outputs $\sigma=\Sign(m,\sk)$) and $\oVer[\vvk]$ denote a verification oracle (on input $(m,\sigma,i)$, $\oVer[\vvk]$ outputs $\Verify(m,\sigma,\vk[i])$). Also, let $\Adv^{\oSig[\sk],\oVer[\vvk]}$ denote a computationally unbounded adversary that makes at most \sigcals calls to $\oSig[\sk]$ and at most \vercals calls to $\oVer[\vvk]$, and gets to see the verification keys indexed by some subset $\IS\subsetneq [n]$. The following properties hold:
    \begin{itemize}
        \item (unforgeability)
        $\Adv^{\oSig[\sk],\oVer[\vvk]}$ cannot generate a valid signature on message $m'$ of his choice, other than the one he queries $\oSig[\sk]$ on (except with negligible probability). Formally,
        \[
        \Pr\left[\begin{array}{c}
            \Gen\outp(\sk,\vvk); \text{ for some $\IS\subsetneq[n]$ chosen by } \Adv^{\oSig[\sk],\oVer[\vvk]}:\\
            \biggl(A^{\oSig[\sk],\oVer[\vvk]}\bigl(\vvk|_{\IS}\bigr)\outp (m,\sigma,j)\biggr)\ \ \wedge\ \ (m \text{ was not queried to }
            \oSig[\sk])\ \ \wedge\\
            (j\in [n]\setminus \IS)\ \ \wedge \ \ \bigl(\Verify(m,\sigma,\vk[j])=1\bigr)
        \end{array}\right] = \negl(\secParam).
        \]

        \item (consistency)\footnote{This property is often referred to as transferability.}
        $\Adv^{\oSig[\sk],\oVer[\vvk]}$ cannot create a signature that is accepted by some (honest) party and rejected by some other even after seeing $\sigcals$ valid signatures and verifying $\vercals$ signatures (except with negligible probability). Formally,
        \[
        \Pr\left[\begin{array}{c}
            \Gen\outp(\sk,\vvk); \text{ for some $\IS\subsetneq[n]$ chosen by } \Adv^{\oSig[\sk],\oVer[\vvk]}(\sk):\\
            A^{\oSig[\sk],\oVer[\vvk]}(\sk,\vvk|_\IS)\outp (m,\sigma) \\
            \exists i,j\in[n]\setminus \IS \text{ s.t. } \Verify(m,\sigma,\vk[i])\neq\Verify(m,\sigma,\vk[j])
        \end{array}\right] = \negl(\secParam).
        \]
    \end{itemize}
\end{itemize}
\end{definition}

In~\cite{SHZI02,SS11} a signature scheme satisfying the above notion of security was constructed. These signatures have a deterministic signature generation algorithm $\Sign$. In the following (\cref{scheme:itsigs}) we describe the construction from~\cite{SHZI02} (as described by \cite{SS11} but for a single signer). We point out that the keys and signatures in the described scheme are elements of a sufficiently large finite field $\Field$ (\ie $\ssize{\Field}=O(2^{\poly(\secParam)}$)); one can easily derive a scheme for strings of length $\ell=\poly(\secParam)$ by applying an appropriate encoding: \eg map the \ith element of $\Field$ to the \ith string (in the lexicographic order) and vice versa. We say that a value $\sigma$ is {\em a valid signature} on message $m$ (with respect to a given key setup $(sk,\vvk)$), if for every honest $\Party_i$ it holds that $\Verify(m,\sigma,\vk[i])=1$.

\begin{theorem}[\cite{SS11}]\label{thm:itsign}
Assuming $\ssize{\Field}=\Omega(2^\secParam)$ and $\sigcals=poly(\secParam)$ the above signature scheme (\cref{scheme:itsigs}) is an information-theoretically $(\sigcals,poly(\secParam))$-secure \emph{$\PS$-verifiable signature scheme}.
\end{theorem}

\subsubsection{Averaging Samplers}\label{sec:samplers}

Samplers~\cite{Z97,G11} were used in distributed protocols as a technique for \emph{universe reduction} \cite{GVZ06,KLST11,BGH13}.
Specifically, they allow to sample points of a given universe such that the probability of hitting any particular subset approximately matches its density.

\begin{definition}[\cite{KLST11,BGH13}]
A function $\samp \colon X \rightarrow Y$ is a $(\theta,\delta)$-sampler if for any set $S \subseteq Y$, at most a $\delta$ fraction of the inputs $x\in X$ satisfy
\[
\frac{\size{\samp(x)\cap S}}{\size{S}} > \frac{\size{S}}{\size{Y}} + \theta.
\]
\end{definition}

The constructions of samplers in \cite{KLST11,BGH13} provide the additional guarantee that the sampled subsets do not have ``large'' intersections.
This is an important property when the sampler is used to select committees (quorums), so that no committee member ends up being overloaded.
Specifically, let $H(x,i)=\samp(x\cdot n +i)$ for $x \in X$ and $i \in [n]$, and denote by $H^{-1}(x,i)$ the set of nodes $y$ such that $i\in H(x,y)$.
We say that a node $i$ is \textsf{$d$-overloaded} by $H$ if for some constant $a$, there is exists $x \in X$ such that $\ssize{H^{-1}(x,i)}> a\cdot d$.
Samplers that are not overloading can be constructed with the following parameters.

\begin{lemma}[\cite{KLST11,BGH13}]
For every constant $c$, for $\delta = \ssize{X}^{-1}$, and any $\theta > 0$, there is a $(\theta,\delta)$-sampler $H : X \times [n] \rightarrow [n]^d$ with $d = O\left(\frac{\log(1/\delta)}{\theta^2}\right)$ such that for all $x \in X$, no $i \in [n]$ is $d$-overloaded.
\end{lemma}

\begin{nfbox}{Construction of information-theoretic signatures~\cite{SS11}}{scheme:itsigs}
\begin{description}
    \item[\bf Key Generation: ]
    The algorithm for key generation $\Gen(1^\secParam,n,\sigcals)$ is as follows:
    \begin{enumerate}
        \item
        For $(j,k)\in\{0,\ldots,n-1\}\times\{0,\ldots,\sigcals\}$, choose $a_{ij}\in_R\Field$ uniformly at random and set the signing key to be (the description of) the multi-variate polynomial
        \[
        \sk\assign f(y_1,\ldots,y_{n-1},x)=\sum_{k=0}^{\sigcals}a_{0,k}x^k+\sum_{j=1}^{n-1}\sum_{k=0}^{\sigcals}a_{j,k}y_jx^k.
        \]
        \item
        For $i\in[n]$, choose vector $\vec{v}_i=(v_{i,1},\ldots,v_{i,n-1})\in_R\Field^{n-1}$ uniformly at random and set the \ith verification key to be
        \[
        \vk[i]=(\vec{v}_i,f_{\vec{v}_i}(x)),
        \]
        where $f_{\vec{v}_i}(x)=f(v_{i,1},\ldots,v_{i,n-1},x)$.
    \end{enumerate}
    \item[\bf Signature Generation:]
    The algorithm for signing a message $m\in\Field$, given the above signing key, is (a description of) the following polynomial
    \[
    \Sign(m,\sk)\assign g(y_1,\ldots,y_{n-1})\assign f(y_1,\ldots,y_{n-1},m)
    \]
    \item[\bf Signature Verification:]
    The algorithm for verifying a signature $\sigma=g(y_1,\ldots,y_n)$ on a given message $m$ using the \ith verification key is
    \[
    \Verify(m,\sigma,\vk[i])=
    \left\{\begin{array}{l}
        1, \text{ if } g(\vec{v}_i)=f_{\vec{v}_i}(m)\\
        0, \text{ otherwise }
    \end{array}\right.
    \]
\end{description}
\end{nfbox}

\subsection{Model Definition}\label{sec::Def:model}

We provide the basic definitions for secure multiparty computation according to the real/ideal paradigm, for further details see~\cite{Goldreich04} (which in turn follows~\cite{GL90,Beaver91,MR91,Canetti00}).
Informally, a protocol is secure according to the real/ideal paradigm, if whatever an adversary can do in the real execution of protocol, can be done also in an ideal computation, in which an uncorrupted trusted party assists the computation. We consider security with \emph{guaranteed output delivery}, meaning that the ideal-model adversary cannot prematurely abort the ideal computation. For the sake of clarity, we focus in this section on the simpler case of \emph{static} adversaries, that decide on the set of corrupted parties before the protocol begins. The case of \emph{adaptive} adversaries, that can decide which party to corrupt based on information gathered during the course of the protocol, follows in similar lines, but is more technically involved. We refer the reader to~\cite{Canetti00} for a precise definition of adaptively secure MPC.

\begin{definition}[functionalities]\label{def:func}
An $n$-party \textsf{functionality} is a random process that maps vectors of $n$ inputs to vectors of $n$ outputs.
Given an $n$-party functionality $f \colon (\zs)^n \rightarrow (\zs)^n$, let $f_i(\vx)$ denote its \ith output coordinate, \ie $f_i(\vx) = f(\vx)_i$.
\end{definition}

\paragraph{Real-model execution.}
An $n$-party protocol $\pi= (\Party_1,\ldots,\Party_n)$ is an $n$-tuple of probabilistic interactive Turing machines. The term \emph{party $\Party_i$} refers to the $i$'th interactive Turing machine. Each party $\Party_i$ starts with input $x_i\in\zs$ and random coins $r_i\in\zs$.
An \emph{adversary} \Adv is another probabilistic interactive Turing machine describing the behavior of the corrupted parties. It starts the execution with input that contains the identities of the corrupted parties, their private inputs, and an additional auxiliary input.
The parties execute the protocol in a \emph{synchronous} network. That is, the execution proceeds in rounds: each round consists of a \emph{send phase} (where parties send their messages for this round) followed by a \emph{receive phase} (where they receive messages from other parties).
The adversary is assumed to be \emph{rushing}, which means that it can see the messages the honest parties send in a round before determining the messages that the corrupted parties send in that round.

We consider the \emph{point-to-point (communication) model}, where all parties are connected via a \emph{fully connected point-to-point network}. We emphasize that although every party has the ability to send a message to every other party, and to receive a message from every other party, we will focus on protocols where each party will only communicate with a subset of the parties.
We consider three models for the communication lines between the parties: In the \emph{authenticated-channels} model, the communication lines are assumed to be ideally authenticated but not private (and thus the adversary cannot modify messages sent between two honest parties, but can read them). In the \emph{secure-channels} model, the communication lines are assumed to be ideally private (and thus the adversary cannot read or modify messages sent between two honest parties, but he learns the \emph{size} of the message that was sent on the channel). In the \emph{hidden-channels} model, the communication lines are assumed to hide the very fact that a message has been sent on the channel (and thus the adversary is not aware that a message has been sent between two honest parties).
We do not assume the existence of a \emph{broadcast channel}, however, we will occasionally assume the availability of a trusted preprocessing phase, that is required for executing a broadcast protocol.

Throughout the execution of the protocol, all the honest parties follow the instructions of the prescribed protocol, whereas the corrupted parties receive their instructions from the adversary. The adversary is considered to be \emph{malicious}, meaning that it can instruct the corrupted parties to deviate from the protocol in any arbitrary way. At the conclusion of the execution, the honest parties output their prescribed output from the protocol, the corrupted parties output nothing, and the adversary outputs an (arbitrary) function of its view of the computation (containing the views of the corrupted parties). The view of a party in a given execution of the protocol consists of its input, its random coins, and the messages it sees throughout this execution.

\begin{definition} [real-model execution]\label{def:RealModel}
Let $\pi= (\Party_1,\ldots, \Party_n)$ be an $n$-party protocol and let $\IS \subseteq [n]$ denote the set of indices of the parties corrupted by $\Adv$. The {\sf joint execution of $\pi$ under $(\Adv,\IS)$ in the real model}, on input vector $\vx= (x_1,\ldots, x_n)$, auxiliary input $\aux$, and security parameter $\secParam$, denoted $\REAL_{\pi,\IS,\Adv(\aux)}(\vect{x},\secParam)$, is defined as the output vector of $\Party_1,\ldots,\Party_n$ and $\Adv(\aux)$ resulting from the protocol interaction, where for every $i \in \IS$, party $\Party_i$ computes its messages according to $\Adv$, and for every $j \notin \IS$, party $\Party_j$ computes its messages according to $\pi$.
\end{definition}

\paragraph{Ideal-model execution.}
An ideal computation of an $n$-party functionality $f$ on input $\vx=(x_1,\ldots,x_n)$ for parties $(\Party_1,\ldots,\Party_n)$ in the presence of an ideal-model adversary $\Adv$ controlling the parties indexed by $\IS\subseteq[n]$, proceeds via the following steps.
\begin{itemize}
  \item[\emph{Sending inputs to trusted party:}]
  An honest party $\Party_i$ sends its input $x_i$ to the trusted party.
  The adversary may send to the trusted party arbitrary inputs for the corrupted parties. Let $x_i'$ be the value actually sent as the input of party $\Party_i$.

  \item[\emph{Trusted party answers the parties:}]
  If $x_i'$ is outside of the domain for $\Party_i$, for some index $i$, or if no input was sent for $\Party_i$, then the trusted party sets $x_i'$ to be some predetermined default value.
  Next, the trusted party computes $(y_1, \ldots, y_n)=f(x_1', \ldots, x_n')$ and sends $y_i$ to party $\Party_i$ for every $i$.

  \item[\emph{Outputs:}]
  Honest parties always output the message received from the trusted party and the corrupted parties output nothing.
  The adversary $\Adv$ outputs an arbitrary function of the initial inputs $\set{x_i}_{i\in\IS}$, the messages received by the corrupted parties from the trusted party $\set{y_i}_{i\in\IS}$, and its auxiliary input.
\end{itemize}

\begin{definition}[ideal-model execution]\label{def:IdealModel}
Let $f\colon(\zs)^n \rightarrow (\zs)^n$ be an $n$-party functionality and let $\IS\subseteq [n]$. The {\sf joint execution of $f$ under $(\Adv, I)$ in the ideal model}, on input vector $\vect{x}=(x_1, \ldots, x_n)$, auxiliary input $\aux$ to $\Adv$, and security parameter $\secParam$, denoted $\IDEAL_{f,\IS,\Adv(\aux)}(\vect{x},\secParam)$, is defined as the output vector of $\Party_1, \ldots, \Party_n$ and $\Adv(\aux)$ resulting from the above described ideal process.
\end{definition}

\paragraph{Security definition.}

Having defined the real and ideal models, we can now define security of protocols according to the real/ideal paradigm.
\begin{definition}\label{def:SecureProtocol}
Let $f\colon(\zs)^n \rightarrow (\zs)^n$ be an $n$-party functionality, and let $\pi$ be a probabilistic polynomial-time protocol computing $f$. The protocol $\pi$ is a \textsf{protocol that $t$-securely computes $f$ with computational security}, if for every probabilistic polynomial-time real-model adversary \Adv, there exists a probabilistic polynomial-time adversary $\Sim$ for the ideal model, such that for every $\IS\subseteq [n]$ of size at most $t$, it holds that
\[
\set{\bigbrack \REAL_{\pi, \IS, \Adv(\aux)}(\vx, \secParam)}_{(\vx, \aux)\in(\zs)^{n+1}, \secParam\in\N}
\compindist
\set{\bigbrack \IDEAL_{f, \IS, \Sim(\aux)}(\vx, \secParam)}_{(\vx, \aux)\in(\zs)^{n+1}, \secParam\in\N}.
\]
The protocol $\pi$ is a \textsf{protocol that $t$-securely computes $f$ with information-theoretic security}, if for every real-model adversary \Adv, there exists an adversary $\Sim$ for the ideal model, whose running time is polynomial in the running time of $\Adv$, such that for every $\IS\subseteq [n]$ of size at most $t$, it holds that
\[
\set{\bigbrack \REAL_{\pi, \IS, \Adv(\aux)}(\vx, \secParam)}_{(\vx, \aux)\in(\zs)^{n+1}, \secParam\in\N}
\statclose
\set{\bigbrack \IDEAL_{f, \IS, \Sim(\aux)}(\vx, \secParam)}_{(\vx, \aux)\in(\zs)^{n+1}, \secParam\in\N}.
\]
\end{definition}

\paragraph{The Hybrid Model.}
The hybrid model is a model that extends the real model with a trusted party that provides ideal computation for specific functionalities. The parties communicate with this trusted party in exactly the same way as in the ideal model described above.

Let $f$ and $g$ be $n$-party functionalities. Then, an execution of a protocol $\pi$ computing $g$ in the $f$-hybrid model, involves the parties sending normal messages to each other (as in the real model) and in addition, having access to a trusted party computing $f$. It is essential that the invocations of $f$ are done sequentially, meaning that before an invocation of $f$ begins, the preceding invocation of $f$ must finish. In particular, there is at most one call to $f$ per round, and no other messages are sent during any round in which $f$ is called.

Let $\Adv$ be an adversary with auxiliary input $\aux$ and let $\IS\subseteq[n]$ be the set of corrupted parties.
We denote by $\HYBRID^f_{\pi,\IS,\Adv(\aux)} (\vx,\secParam)$ the random variable consisting of the view of the adversary and the output of the honest parties, following an execution of $\pi$ with ideal calls to a trusted party computing $f$ on input vector $\vx = (x_1, \ldots, x_n)$, auxiliary input $\aux$ to $\Adv$, and security parameter $\secParam$.

In this work, we will employ the sequential composition theorem of Canetti \cite{Canetti00}.
\begin{theorem}[\cite{Canetti00}]\label{thm:Composition}
Let $f$ and $g$ be $n$-party functionalities. Let $\rho$ be a protocol that $t$-securely computes $f$, and let $\pi$ be a protocol that $t$-securely computes $g$ in the $f$-hybrid model. Then, protocol $\pi^{f\rightarrow\rho}$, that is obtained from $\pi$ by replacing all ideal calls to the trusted party computing $f$ with the protocol $\rho$, is a protocol that $t$-securely computes $g$ in the real model.
\end{theorem}

\paragraph{Extended functionalities and extended protocols.}
As mentioned above, in the sequential composition theorem (\cref{thm:Composition}) an $n$-party protocol $\pi$ is considered, in which all $n$ parties invoke the trusted party for computing an $n$-party functionality $f$. Next, the adjusted protocol $\pi^{f\rightarrow\rho}$, where all hybrid calls to $f$ are replaced by an $n$-party protocol $\rho$ for computing $f$, is proven secure.
It is essential that the same set of $n$ parties, defined by $\pi$, will run all executions of the sub-protocol $\rho$ in order to claim security of $\pi^{f\rightarrow\rho}$. Looking ahead, in some of our constructions (in \cref{sec:ne_mpc}) we will use sub-protocols that are executed only by a subset of the parties. During the rounds in which the sub-protocol takes place, the remaining parties (that do not patriciate in the sub-protocol) should remain idle, \ie not send any message and not receive any message. Toward this goal, we show how to extend functionality and protocols, that are defined for a subset of the $n$ parties, into functionalities and protocols that are defined for the entire party-set, such that the parties outside of the original subset remain idle.

Let $m<n$. Given an $m$-party functionality $f$, an $m$-party protocol $\rho$ that $t$-securely computes $f$, and a subset $\PS=\sset{i_1,\ldots,i_m}\subseteq[n]$ of size $m$, we define the ``extended functionality'' $\extend{[n]}{\PS}{f}$ as the $n$-party functionality in which the output vector $(y_1,\ldots,y_n)$ is defined as follows: for every $i\notin\PS$ the output $y_i=\emptystr$ is defined to be the empty string, and for every $i\in\PS$, the output is computed via $(y_{i_1},\ldots,y_{i_m})=f(x_{i_1},\ldots,x_{i_m})$. In addition, we define the protocol $\extend{[n]}{\PS}{\rho}$ as the $n$-party protocol where every party $\Party_i$, with $i\in\PS$, follows the code of $\rho$, and ignores messages from parties outside of $\PS$, and every party $\Party_i$, with $i\notin\PS$, doesn't send any message, ignores all incoming messages, and outputs $\emptystr$. The following lemma is straightforward.

\begin{lemma}
If $\rho$ is an $m$-party protocol that $t$-securely computes the $m$-party functionality $f$, then $\extend{[n]}{\PS}{\rho}$ is an $n$-party protocol that $t$-securely computes the $n$-party functionality $\extend{[n]}{\PS}{f}$.
\end{lemma}

\subsection{Correlated Randomness Functionalities}\label{sec:PKI}

The positive results presented in \cref{sec:ne_mpc} are defined in the a model where the parties receive correlated randomness generated by a trusted setup phase (formally, in the correlated-randomness hybrid model). The correlated randomness that we consider is \emph{public-key infrastructure (PKI)}, which is ``minimal'' in a sense, and commonly used in MPC protocols. In the computational setting, we consider PKI that is based on digital signatures (which exist assuming one-way functions exist), and in the information-theoretic setting a PKI that relies on information-theoretic signatures (see \cref{sec:itsign}).

\paragraph{Public-key infrastructure.}

The PKI functionality $\fpki$ (\cref{fig:fpki}) generates a pair of signing and verification keys for every party, and hands each party its signing key along with all verification keys. For simplicity, when we say that a protocol $\pi$ is in the PKI model, we mean that $\pi$ is defined in the $\fpki$-hybrid model.

\begin{nfbox}{The PKI functionality}{fig:fpki}
\begin{center}
    \textbf{The functionality} $\fpki$
\end{center}
The $n$-party functionality $\fpki$ is parametrized by a signature scheme $(\Gen,\Sign,\Verify)$ and proceeds with parties $\PS_1=\sset{\Party_1,\ldots,\Party_n}$ as follows.
\begin{enumerate}
    \item
    For every $i\in[n]$ generate $(\sk[i],\vk[i])\gets\Gen(1^\secParam)$.
    \item
    The output for party $\Party_i$ is the signing key $\sk[i]$ and the vector of verification keys $(\vk[1],\ldots,\vk[n])$.
\end{enumerate}
\end{nfbox}

\paragraph{Information-theoretic PKI.}

The $(\sigcals,\vercals)$-IT-PKI functionality $\fitpki^{(\sigcals,\vercals)}$ (\cref{fig:fitpki}) generates $n$ tuples of signing and verification keys, and hands each party its signing key along with all corresponding verification keys. For simplicity, when we say that a protocol $\pi$ is in the $(\sigcals,\vercals)$-IT-PKI model, we mean that $\pi$ is defined in the $\fitpki^{(\sigcals,\vercals)}$-hybrid model. By the IT-PKI model we mean $(\sigcals,\vercals)$-IT-PKI model where $\sigcals$ and $\vercals$ are polynomial in $\secParam$.

\begin{nfbox}{The information-theoretic PKI functionality}{fig:fitpki}
\begin{center}
    \textbf{The functionality} $\fitpki$
\end{center}
The $n$-party functionality $\fitpki^{(\sigcals,\vercals)}$ is parametrized by an information-theoretically $(\sigcals,\vercals)$-secure signature scheme $(\Gen,\Sign,\Verify)$ and proceeds with parties $\PS_1=\sset{\Party_1,\ldots,\Party_n}$ as follows.
\begin{enumerate}
    \item
    For every $i\in[n]$ generate $(\sk[i],\vvk[i])\gets\Gen(1^\secParam,n,\sigcals,\vercals)$, where $\vvk[i]=(\vk[1]^i,\ldots,\vk[n]^i)$.
    \item
    The output for party $\Party_i$ is the signing key $\sk[i]$ and the vector of verification keys $(\vk[i]^1,\ldots,\vk[i]^n)$.
\end{enumerate}
\end{nfbox}

\ifdefined\IsProofInAppendix
\section{MPC with Non-Expanding Communication Graph (Cont'd)}\label{sec:ne_mpc_cont}

We now provide complementary material for \cref{sec:ne_mpc}.

\subsection{Proof of \texorpdfstring{\cref{prop:security}}{Lg}}\label{sec:prop:security_proof}

\subsection{Proof of \texorpdfstring{\cref{lem:mpc_no_expander_adaptive}}{Lg}}\label{sec:lem:mpc_no_expander_adaptive_proof}

\section{Expansion is Necessary for Correct Computation (Cont'd)}\label{sec:LB_Expander_cont}

We now provide complementary material for \cref{sec:LB_Expander}.

\fi

\end{document}